\pdfoutput=1
\documentclass[12pt]{article}

\usepackage{etex} 

\usepackage{cite}
\usepackage{listings}
\usepackage[utf8]{inputenc}
\usepackage[text={7in,8.0in},centering]{geometry}
\usepackage{amsmath}
\usepackage{amssymb}
\usepackage{amsthm}
\usepackage{bbm}
\usepackage{tabularx}
\usepackage{url}
\usepackage{fancyhdr}
\usepackage{algpseudocode}
\usepackage{algorithm}
\usepackage{empheq}
\usepackage{color}
\usepackage{graphicx}
\usepackage{hyperref}
\usepackage[font={footnotesize},labelfont=bf]{caption}
\usepackage{subcaption}
\usepackage{authblk}
\usepackage{tikz}

\definecolor{lightyellow}{cmyk}{0,0,0.3,0}
\graphicspath{{./Figures/}}

\newcommand{\vc}[1]{{\pmb{#1}}}

\newcommand{\rhostar}{\rho^*}

\newcommand{\sv}{\vc{s}}

\newcommand{\uv}{\vc{u}}
\newcommand{\Uv}{\vc{U}}

\newcommand{\xv}{\vc{x}}
\newcommand{\yv}{\vc{y}}
\newcommand{\Yv}{\vc{Y}}
\newcommand{\wv}{\vc{w}}

\newcommand{\vv}{\vc{v}}
\newcommand{\zv}{\vc{z}}
\newcommand{\Zv}{\vc{Z}}

\newcommand{\alphav}{\vc{\alpha}}
\newcommand{\sigmav}{\vc{\sigma}}

\newcommand{\qhat}{\hat{q}}
\newcommand{\qhatrb}{\hat{q}^{rb}}

\newcommand{\Pcal}{\mathcal{P}}
\newcommand{\Fcal}{\mathcal{F}}
\newcommand{\Rcal}{\mathcal{R}}

\newcommand{\Ycal}{\mathcal{Y}}
\newcommand{\samp}[1]{^{(#1)}}

\newcommand{\Reals}{\mathbb{R}}
\newcommand{\Scube}{\mathbb{S}}

\newcommand{\Expect}[1]{\mathbb{E}_{#1}}
\newcommand{\indic}{\mathbbm{1}}
\newcommand{\nbrs}{\mathcal{N}}

\newcommand{\fbx}[1]{$\framebox{${#1}$}$}  

\newcommand{\Acal}{\mathcal{A}}

\newcommand{\algoname}{\textsc}
\newcommand{\instancename}{\texttt}

\newsavebox{\Figurebox}
\usepackage{pstricks}

\usepackage{tikz-qtree}
\usepackage{tikz}
\usetikzlibrary{calc}

\newtheorem{proposition}{Proposition}

\newtheorem{guideline}{Guideline}

\begin{document}
\title{Discrete Equilibrium Sampling with Arbitrary Nonequilibrium Processes}
\author[1]{Firas Hamze}
\author[1]{Evgeny Andryash}
\affil[1]{D-Wave Systems Inc.}
\maketitle

\begin{abstract}

  We present a novel framework for performing statistical sampling,
  expectation estimation, and partition function approximation using
  \emph{arbitrary} heuristic stochastic processes defined over
  discrete state spaces. Using a highly parallel construction we call
  the
  \emph{sequential constraining
    process}, we are able to simultaneously generate states with the
  heuristic process and accurately estimate their probabilities, even
  when they are far too small to be realistically inferred by
  direct counting. After showing that both theoretically correct
  importance
  sampling and Markov chain Monte Carlo are possible
  using the sequential constraining process, we integrate it into a
  methodology
  called \emph{state space sampling}, which extends the classical
  ideas of state space search from computer science to the sampling
  context. The methodology comprises a \emph{dynamic} data
  structure that constructs a robust Bayesian model of the statistics
  generated by the
  heuristic process subject to an accuracy constraint, the
  \emph{posterior
    Kullback-Leibler divergence}.
  Sampling from the dynamic structure will generally yield \emph{partial}
  states, at which point
  the heuristic is called recursively to refine the structure at the
  needed point in the state space until \emph{complete} states, along
  with
  their probability estimates, are returned. Importance sampling is
  thereby enabled.
  Our experiments on various
  Ising models strongly suggest that
  state space sampling enables heuristic state generation with
  dramatically accurate probability estimates,
  demonstrated by illustrating the convergence of a simulated
  annealing process
  to the Boltzmann
  distribution with increasing run length. Consequently, heretofore
  unprecedented direct importance sampling using the \emph{final}
  (marginal) distribution of a generic stochastic process is allowed,
  potentially augmenting considerably the range of algorithms
  at the Monte Carlo practitioner's disposal.
\end{abstract}

\section{Introduction}
\label{sec:Intro}

We consider in this paper the problem of computing expectations of
functions with respect to discrete distributions
\begin{equation}
\Expect{\pi}[ h(\Yv) ]
\label{eq:expectationGenereal}
\end{equation}
Without loss of generality, the function $\pi$ defined over a discrete
state space can be assumed to be in \emph{Boltzmann distribution} form
\begin{equation}
\pi(\yv) = \frac{e^{-E(\yv)}}{Z_\pi}
\end{equation}
where $E(\yv)$ and $Z_\pi = \sum_{\yv} e^{-E(\yv)}$ are the
\emph{energy function} and \emph{partition function} respectively.
In this paper, we assume that $\yv$ is a vector of $m$
\emph{Ising}-valued variables, i.e
\begin{eqnarray*}
  y_i \in \{ -1,1 \} \\
  \yv = y_{1:m} \in \Scube^m
\end{eqnarray*}
where $\Scube^m \triangleq \{-1,1\}^m$ is the set of all possible
Ising configurations on $m$ variables. Modification of the paper's
ideas to the cases of Boolean $(\{0,1\})$ or more general discrete
values is straightforward.

The problem defined in Equation (\ref{eq:expectationGenereal}) arises
in a wide range of applications from statistics and machine learning
to statistical physics; some examples include computing spin-spin
correlations, local magnetizations, and average energies. When
performing \emph{Boltzmann machine learning}, access to accurate
estimates of these quantities enables \emph{maximum likelihood}
parameter estimation \cite{sejnowski2001learning}; in physics, their
behaviour may signal phenomena such as \emph{phase transitions}
\cite{chandler1987introduction}. Additionally, in several settings
such as \emph{approximate
  counting}\cite{jerrum1996markov} and \emph{Bayesian model selection}
\cite{bernardo2009bayesian} it is of interest to approximate $Z_\pi$.

The task of approximating expectations can, of course, in principle be
treated by generating samples from $\pi(\yv)$ and computing the
empirical average of $h$ over those samples. In situations where the
problem possesses some simplifying structure such as a low treewidth,
this is indeed feasible. In general, drawing such samples for systems
with large $m$ is highly nontrivial. A common approach is to use a
\emph{Markov Chain Monte Carlo} (MCMC) algorithm instantiating a
generally correlated sequence of samples whose distribution converges
asymptotically to $\pi$; examples include \emph{Metropolis}
\cite{metropolis1953equation}, \emph{Metropolis-Hastings}
\cite{hastings1970monte}, and Gibbs (\emph{heat bath}) samplers
\cite{geman1984stochastic,glauber1963time}. The simplest versions of
these algorithms are \emph{local} (univariate), but more sophisticated
generalizations such \emph{parallel tempering}
\cite{hukushima1996exchange,geyer1991computing} and
\emph{cluster-based} methods
\cite{swendsen1987nonuniversal,wolff1989collective,houdayer2001cluster,zhu2015efficient}
have been devised to improve efficiency.  When designing such
algorithms, special care must be put into ensuring that they yield the
correct stationary distribution. Consequently, it is often difficult
to engineer a correct algorithm that makes large jumps in the state
space and avoids getting trapped in local optima.  In addition,
approximating $Z_\pi$ using MCMC is not automatically obvious.

An alternative to MCMC is \emph{importance sampling,} where one
generates samples according to some \emph{trial} (or \emph{proposal})
distribution and appropriately weights their contribution in the
estimator to ensure that a statistical average with respect to $\pi$
is obtained. This method can suffer from terrible inaccuracy in
high-dimensional problems. The main issue is that it becomes
exceedingly difficult to design a trial distribution that is
simultaneously \emph{tractable}, i.e. allowing efficient sampling and,
for the purpose of reweighting, evaluation, and possesses sufficient
\emph{overlap} with the statistically-dominant regions of the state
space under $\pi$. To illustrate two extremes, using a uniform trial
distribution certainly satisfies the tractability requirement, but
will typically require a completely unrealistic number of samples to
yield useful estimates for most interesting problems; on the other
hand, one can conceivably design an elaborate, randomized heuristic
process which covers the important areas of $\pi$ quite well, but
computing the probabilities of states yielded by such a process may be
infeasible. As a concrete example of the latter case, one may be
interested in probing the statistics of the global energy minima, or
\emph{ground states}, of $E$. These states occur with uniform
probability under $\pi$ at zero temperature. The question of how to
perform estimation when given access to a heuristic that generates
ground states with unknown and generally nonuniform probabilities is
a subset of the problems tackled by the approach taken in this paper.

The class of algorithms known as \emph{sequential Monte Carlo} (SMC)
(or \emph{population annealing}) methods
\cite{jarzynski1997nonequilibrium,del2006sequential,hukushima2003population,machta2010population,
  wang2015population} are in some sense a synthesis of importance
sampling and MCMC. In general, applying a finite number of MCMC moves
will yield a distribution that is \emph{out of equilibrium}, in other
words, has not yet converged to $\pi$. Were one able to evaluate this
nonequilibrium probability, an estimator using reweighted samples
could compensate for the finite stopping-time bias incurred by
na\"ively using the raw empirical average. Computing the
nonequilibrium distribution is generally out of the question; SMC
methods circumvent this difficulty by \emph{augmenting} the state
space under consideration to a \emph{joint, growing} one. For example
suppose we could tractably sample from distribution $\pi_0(\yv_0)$
and generate updated state $\yv_1$ according to an MCMC kernel $T$
converging to $\pi(\yv_1)$; of course, the variables $(\Yv_0, \Yv_1)$
will be jointly distributed according to
$\pi_0(\yv_0) T( \yv_1 | \yv_0) $. If we suitably define an
\emph{augmented} target distribution to be
$\pi(\yv_1) L (\yv_0 | \yv_1) $, where $L (\yv_0 | \yv_1)$ is a
conditional distribution which can in general be chosen with great
flexibility, we could use the importance-weighted samples \emph{over
  the joint state space} to approximate the expectation of $h(\yv_1)$
with respect to $\pi$. The idea can be extended to a \emph{sequence}
of target distributions $(\pi_n)$, typically obtained by
\emph{annealing} a temperature parameter scaling the energy such that
the final distribution in the sequence is the target $\pi$. The
variance of estimator grows with the length of the sequence, and hence
variance reduction techniques, such as the (asymptotically) unbiased
strategy of \emph{resampling} (or \emph{particle interaction}) are
often applied. A particularly desirable attribute of SMC methods is
that they enable estimation of the partition function $Z_\pi$ if that
of the initial distribution $\pi_0(\yv)$ is known.

SMC methods have been a considerable advance in Monte Carlo
simulation, but they still impose certain restrictions on the
practitioner. One particularly unappealing aspect of a
commonly-applied variant of SMC is that, due to the usage of a growing
state space, a sequence of distributions with considerable overlap is
required even if the proposal kernel $T$ happens to yield
\emph{perfect} samples from the distribution it targets. Furthermore,
while it is possible in principle to use a more generic move process
than an MCMC operator, designing an appropriate joint distribution via
$L(\yv_0|\yv_1)$ then becomes nontrivial. For example, a tempting
possibility is to simply choose $L(\yv_0|\yv_1) = \pi(\yv_0)$,
yielding a simple-looking expression for the importance weights.
Unfortunately, a serious issue is that the proposal can typically only
be evaluated exactly for \emph{local} moves. This in turn implies that
the overlap between $\pi$ and $T(\xv_1|\xv_0) $ can be negligible,
leading to tremendous variance. The issue would be mitigated one could
somehow instantiate a proposal that achieved good global coverage of
$\pi$, but in that case its evaluation becomes impossible. Indeed it
is precisely this central issue that motivates the present work and
its ramifications.

Our approach departs from the SMC methodology of importance sampling
on a growing state space and reverts to directly using the
distribution of the final state of an arbitrary heuristic process. The
methodology could considerably expand the repertoire of tools at
researchers' disposal with which to perform statistical sampling. For
example, one may now employ, in place of commonly-used MCMC methods
relying on single variable updates, randomized combinatoric methods
(such as \emph{primal-dual algorithms} and \emph{linear
  programming with rounding}) whose distributions are exceedingly
difficult to characterize, but which nonetheless may be excellent
trial processes with which to sample from a target distribution.
Alternatively, well-known MCMC algorithms may be modified in such a way
as to improve their mobility in the state space, despite the generally
crucial property of \emph{detailed
  balance} being consequently broken; for regular MCMC methods, it
provides a compelling diagnostic of convergence. Finally, physical
sampling devices with real-world imperfections causing unknown changes
to their distributions may be brought to bear
\cite{johnson2011quantum, zhu2015best}. As mentioned, our focus in
this work is on distributions over discrete state spaces, but the
general ideas are certainly transferable to the continuum in
principle.

In Section \ref{sec:SCP}, we introduce the \emph{Sequential
  Constraining Process}, the central workhorse of our algorithms.
Given a stochastic sampling heuristic as a ``black box'', this process
instantiates states along with relatively accurate estimates of their
probabilities. It does so by constructing, using a sample
\emph{population} from the heuristic, a \emph{statistical model}
subject to an accuracy criterion. In general, this model will not be
able to represent the full state space for a realistically-sized
population. Nonetheless, when a sample is drawn from the model
corresponding to an incomplete state, the heuristic can be used
\emph{recursively} to construct another accuracy-respecting model with
the incomplete state as a \emph{constraining condition}. In this
manner, full states and probabilities are generated by stochastically
\emph{querying} the proposal. Like most population-based methods, the
sequential constraining process is a highly parallelisable algorithm,
and can hence straightforwardly leverage the ubiquitous computing
power available at the time of this paper's writing. The issue of
parameter estimation, i.e. construction of the modelling
distributions, is a crucial one and is discussed in Section
\ref{sec:SCP:qnEstimators}. We opt for a \emph{robust Bayesian}
approach to estimating the heuristic distributions, with the
\emph{posterior Kullback-Leibler divergence} as the loss function
quantifying estimator accuracy. Rao-Blackwellisation as a means of
\emph{variance reduction} is presented in
\ref{sec:SCP:qnEstimators:RaoBlack}.

In Section \ref{sec:MCusingSCP}, we show the crucial fact that the
sequential constraining process allows \emph{both} importance sampling
and MCMC to be performed. This is of particular interest as our
algorithms are \emph{theoretically correct} but employ as proposals
\emph{approximations} to a complicated distribution; it is \emph{not}
generally true that valid Monte Carlo results when such
approximations are made.

Section \ref{sec:SSS} coalesces the introductory ideas of the first
half of the paper into a practical importance sampling algorithm we
call \emph{state space sampling}. The name is motivated by the class of
algoritms used in artificial intelligence known as \emph{state space
  search}, a generalization of which it can be seen\footnote{Though
  correctness of an MCMC algorithm using approximations derived from
  the heuristic process is shown in principle, practical exploration
  of this is left to future work.}. State space sampling is a process
by which discrete states are drawn by sampling from a \emph{dynamical}
statistical model of an unknown proposal heuristic, estimated from
populations generated using the constraining process. For binary state
spaces, the dynamical model is naturally represented as a \emph{binary
  tree} whose nodes correspond to partial states within the state
space. The sampling procedure can be seen as a stochastic
\emph{exploration} of this binary tree; when a node is reached
corresponding to an incomplete state, the model tree is extended from
that point by invoking the proposal process with the partial state
constraint for a fresh population, growing a new model, and resuming
the sampling. Ideally, the model constructed would be retained (or
\emph{cached}) as long as samples are being drawn, but just as is done
in the search context, a bound on the tree size must imposed to
respect practical memory constraints. We propose implementing this in
a manner such that regions of the model that have lower probability of
being used have higher \emph{priority} for being discarded. A simple
example illustrating the main features of state space sampling is
presented in Section \ref{sec:SSS:comicBook}.

Section \ref{sec:Experiments} discusses experimental results obtained
by applying state space sampling to several types of Ising model using
short-run simulated annealing processes as the stochastic heuristic.
While the point of devising the methodology in this paper was to allow
for the use of more general heuristics, we nonetheless selected
simulated annealing in this preliminary illustration because its
behaviour is relatively well-understood. More specifically, if we claim
that our methodology allows for accurate estimation of distributions
derived from arbitrary heuristics, which would then in turn enable
reliable importance sampling should the heuristic distribution be a
good match to the target, then applying simulated annealing should
unambiguously show convergence of the states' estimated probabilities
to the final target distribution $\pi$ as the run length is increased.
The results presented in Section \ref{sec:Experiments} indeed provide
such confirmation, showing convergence of the ensemble distribution to
$\pi$ in a fairly direct manner. Throughout the paper, we strive to
maximize computational efficiency, either generically or, using
problem-dependent features as discussed in Appendix
\ref{sec:Appendix:partitioning}.

Our analysis and results show state space sampling, and the general
methodology presented in this paper, to be a potentially powerful
framework that takes advantage of copious contemporary computing
power and considerably enlarges the scope of algorithms that can be
used to simulate complex systems.

\section{Notational Preliminaries}
\label{sec:Preliminaries}
Following common convention, random variables appear in upper case,
with their particular values in lower case. Variables appearing in
bold are vector-valued; e.g. $\yv \in \Reals^m$. When appearing with
subscripts, bold symbols refer to elements of a sequence of
vector-valued variables; e.g. $ \yv_n \in \Reals^m$. Plain symbols can
refer to either scalar or multivariate values depending on the form of
their subscripts. For example $y_n$ refers to the $n^{\textrm{th}}$
element of $\yv$, but $y_{1:n}$ refers to first $n$ elements of $\yv$
(hence $y_{1:m} = \yv)$. More generally, $y_I$ specifies the variables
in $\yv$ indexed by subset $I \subseteq \{1, \ldots, m \}$, e.g.  if
$I = \{ 1, 3, 10 \}$, then $y_I$ is short for $\{y_1,y_3,y_{10}\}$. On the
other hand, $U_{n,I}$ denotes the variates indexed by $I$ in the
$n^{\textrm{th}}$ vector-valued $\Uv_n$ appearing in a
sequence. Finally, superscripts in parentheses refer to members of a
set of random variables instantiated from some distribution. For
example to denote drawing $N$ samples from distribution $f$, we write
$\Yv\samp{j} \sim f$ for $j = 1, \ldots, N$. Illustrating some of
these conventions, the conditional distribution of variable
$\mathbf{X}$ given the first $n$ elements of $\yv\samp{j}$, the
\emph{observed} value of the $j^{\textrm{th}}$ sample, would be
written as $f( \xv | y_{1:n}\samp{j} )$.

\section{The Sequential Constraining Process}
\label{sec:SCP}

\subsection{Sampling with Nonequilibrium Distributions}
\label{sec:SCP:sampNEQDist}

Suppose we have available a stochastic procedure, referred to as the
\emph{proposal process} \( \Pcal \),
generating samples from some discrete probability distribution.
Particular simulation-based examples are Markov Chain Monte Carlo
(MCMC) algorithms run for a finite number of iterations, where
consequently, the asymptotic distribution is generally not reached,
non-Markovian \emph{stochastic local search} (SLS) algorithms
\cite{hoos2004stochastic} that may either start from a completely
uniform distribution or, as in the case of \emph{iterated} local
search, generate configurations conditioned on the current state, and
methods that allow for interactions among members of a population such
as \emph{genetic} or \emph{evolutionary} algorithms
\cite{john1992adaptation,goldberg2006genetic}. Alternatively,
\( \Pcal \)
can correspond to deriving configurations from a physical system such
as one designed to implement a thermal or quantum annealing algorithm.

With exceptions such as MCMC, little can be theoretically said of any
statistical or convergence properties of these processes. Thus, while
they may be useful as randomized heuristics to generate low-energy (or
``low-cost'') configurations of a system, in raw form they are
unsuitable for performing correct stochastic simulation of a specified
probability distribution $\pi$ with respect to that system. In the
following, we discuss this issue in such a way as to motivate the
contributions of this paper.

Let $ f(\yv | \xv ) $ be the probability of sampling
$ \yv \in \Scube^m $ by running $ \Pcal $ to completion from current
state $ \xv \in \Scube^m $. If $ \Pcal $ corresponded to some
``perturbative'' process such as a sequence of MCMC or SLS moves, then
$ f(\yv|\xv) $ is the \emph{marginal distribution} of the state at the
final time step of $ \Pcal $. For certain kinds of $ \Pcal $, such as
finite-length MCMC from a random initial state,
$ f(\yv|\xv) = f(\yv) $, and in such cases we will omit the dependence
in the notation. We emphasize to readers familiar with SMC not to
confuse $f(\yv|\xv)$ with the \emph{joint} probability of a sample
path resulting from a sequence of proposal moves.

Given that we can sample from $f(.)$, two well-known methodologies
that would \emph{in principle} allow for estimation of expectations
are \emph{Importance Sampling} (IS) and the \emph{Metropolis-Hastings}
(MH) MCMC algorithm. In the case \( f(\yv|\xv) = f(\yv) \),
we could have

\begin{equation}
  \Expect{\pi}[h(\Yv)] = \Expect{f}[ w(\Yv)h(\Yv) ]
  \label{eq:plainIS}
\end{equation}
where the \emph{importance weights} are given by
\begin{equation*}
  w(\yv) \triangleq \frac{\pi(\yv)}{f(\yv)}
\end{equation*}

For a correlated proposal, arising for example from applying an
iterated local search idea, the standard MH algorithm can be used to
asymptotically sample from $\pi $, which would in turn allow
estimation. The point $ \yv $ generated according to $ f(.|\xv) $
would be accepted with probability \begin{equation}
  \alpha( \xv, \yv ) \triangleq \min \Big [ 1, \frac{ \pi(\yv) f(\xv|\yv) } {
    \pi(\xv) f(\yv|\xv) }  \Big ]
  \label{eq:plainMH}
\end{equation}
When \( f \) itself corresponds to running a
MCMC algorithm for some number of time steps, this can be viewed as an
"outer" MCMC loop. 

Unfortunately, these two approaches are utterly impractical for many
interesting choices of proposal distribution. The difficulty lies in the
requirement to evaluate \( f(\yv) \) for IS and of the \emph{forward} and
\emph{reverse} proposal probabilities, \( f(\yv|\xv) \) and \(
f(\xv|\yv) \) respectively, to perform MH sampling. Consider the simple case
that \( f \) was instantiated by a sequence of \( k+1 \) Markov
operations \( \{T_k\} \) (such as local-variable simulated
annealing) so that
\begin{equation*}
f(\yv|\xv) = \sum_{\sv_1 \ldots \sv_k} T_{k+1}( \yv|\sv_k)\ldots T_{1}(\sv_1|\xv)
\end{equation*}
Even if the Markovian property were exploited, computing the
summations would require exponential (in \( m\)
) time; the task is certainly no simpler for generic \( f \).
When the proposal arises from a physical system exhibiting unknown or
intractable complexity, for example a manufactured implementation
\cite{johnson2011quantum} of quantum
annealing\cite{kadowaki1998quantum}, precise characterization of
\( f \) may be all but impossible.

One may ask if there exists a sensible strategy to approximate \( f(\yv) \)
at a sampled point \( \yv \). A na\"ive option is to draw a set
of \( N \) samples and to approximate \( f(\yv) \) by its histogram
\begin{equation*}
f(\yv) \approx \frac{1}{N} \sum_i \indic_{\yv}[\yv\samp{j}]
\end{equation*}
where $\indic_{\yv}[.]$ is the indicator function. In high dimensions,
this is generally a flawed approach. Indeed, suppose that a
hypothetical ``oracle'' generated configurations distributed exactly
according to target distribution \( \pi \),
but that one was unaware of this and tried to perform importance
sampling relative to the histogram estimator instead. For
computationally reasonable values of \( N \),
and large, non-trivial systems, each sampled point will tend to occur
only a few times in the population, and thus its proposal probability
will simply evaluate to \( \frac{1}{N} \),
leading to an importance weight proportional to \( \pi \);
specifically, in a population where each state is hit a single time,
the weights are
\( w(\yv\samp{j}) = \pi(\yv\samp{j})/\sum_j\pi(\yv\samp{j}) \).
This will result in a highly skewed estimator as the correct weights
in this scenario are unity by construction. In essence, the
\emph{effective} constructed proposal fails to cover the target; only
for extremely large \( N \)
will the histogram estimates eventually approach their correct values
and this effect be mitigated. The main point is that merely being able
to sample from a good proposal is not sufficient; one must also be
able to accurately \emph{evaluate} it.

Consider now the unknown proposal factorized into its
univariate conditionals
\begin{equation}
f(\yv|\xv) = f_1(y_{1} | \xv ) f_2(y_{2} | y_{1},\xv )\ldots f_m(y_{m}|
y_{1:m-1},\xv )  
\label{eq:propFactored}
\end{equation}
While the variables \( y_{1:m} \)
refer to the \emph{final} configurations jointly yielded by the
process, it is instructive to consider the task of generating them in
a sequential manner by successively sampling from the conditionals.
Suppose we conceptually decomposed the problem of generating the
samples from \( f \)
into \( m \)
steps. At step \( 1 \),
we sample a state \( \Uv \)
from the joint \( f(\uv|\xv) \)
and set \( Y_1 \gets u_1 \).
Clearly, \( Y_1 \)
is then a sample from \( f_1(y_1|\xv) \).
At step 2, we sample another state \( \Uv \)
from the joint \( f(\uv|\xv) \).
If \( u_1 = y_1 \)
in this sample, then we set \( Y_2 \gets u_2 \)
and continue; otherwise, we continue to generate additional samples
$\{ \Uv \}$ from the joint until \(u_1 = y_1\).
It can readily be shown that this procedure, a simple form of
\emph{rejection sampling,} will yield a sample from the conditional
\( f_2(y_2|y_1,\xv ) \).
In general, at step \( n \),
given values \( y_{1:n-1} \),
one can in principle sample $Y_n$ from the conditional
\(f_n(y_n | y_{1:n-1},\xv ) \)
by repeatedly drawing \( \{\Uv\} \sim f(\uv|\xv) \)
until the condition \( u_{1:n-1} = y_{1:n-1} \)
is met, setting \(Y_n \gets u_n \)
for the sample $\Uv$ that meets the condition. Performing this
operation sequentially for $n = (1, \ldots, m)$, the random variable
\( \Yv \triangleq Y_{1:m} \)
is therefore distributed as \(f(\yv|\xv) \).

At first sight, this seems like a puzzling and contrived way of
sampling from \( f(\yv|\xv) \);
indeed, any sample used in the procedure is assumed to come from that
very distribution! The reason for its potential interest is that, were
one able to carry it out, it would enable \emph{estimation} of the
conditional components \( \{f_n(y_n|y_{1:n-1},\xv)\} \)
and thus of the full joint \( f(y_{1:m}|\xv) \).

Estimation of the first few conditionals in (\ref{eq:propFactored})
indeed turns out to be realistic. For example, unlike attempting to
estimate \( f(\yv|\xv) \)
over the full state space \( \Scube^m \),
it is quite reasonable to approximate the scalar, binary function
\( f_1(y_{1} | \xv ) \)
using either a histogram or other methods to be discussed later. It is
possible in principle to estimate \( f_n(y_{n} | y_{1:n-1},\xv) \)
using the same idea, but unfortunately, the probability of generating
a state from the joint matching a particular signature \( y_{1:n-1} \)
decreases exponentially with $n$, and so unsurprisingly, the
dimensionality issues recur.

There is no obvious way around this issue if one were to insist on
sampling from \( f(\yv|\xv) \), but there is a potential solution if
one were willing to instantiate another proposal \emph{derived} from
\( \Pcal \) whose behaviour can be expected to be \emph{similar} to that
of \( \Pcal \). In the next section, we discuss this idea, which is one of the main
building blocks of our algorithms. Crucial practical matters will be
discussed in later sections.

\subsection{Sequential Constraining Process Introduced}
\label{sec:SCP:SeqConstProc}
We begin by illustrating how to recursively invoke the heuristic in
such a way as to implicitly define a distribution which can both be
sampled from and evaluated. This in some sense sidesteps the rejection
sampling difficulties discussed in the last section. Consider first
running the proposal process with respect to a \emph{constrained
  subsystem} of the overall distribution. More precisely, the
heuristic is modified so that some subset of the variables is ensured
of having specific values \( \{y_i\}\).
For reasons that will become clear soon, we refer to the variates
generated by such a process as $\Uv$ in distinction to $\Yv$. An
example of a proposal targeting a constrained subsystem is using a
SLS or MCMC heuristic, holding constant from the initial conditions
the values of variables $U_I $ to $ y_I$ for some subset $I$ of the
variable indices, and running the process to completion on the
remaining variables. Alternatively, it may be preferable to
\emph{gradually} apply a local bias (or ``field'') whose magnitude
increases over the course of the simulation such that the final values
of $U_I$ are constrained to being $y_I$.

The joint proposal distribution with a given set of
constraints imposed (either by \emph{hard clamping} from initial
conditions or by increasing their effect over the course of
\( \Pcal \)) is defined for \(n \geq 2\) to be
\begin{equation}
q_n( \uv | y_{1:n-1}, \xv ) = \Bigg[ \prod^{n-1}_{k=1}
\delta_{y_k}(u_k) \Bigg ]q_n( u_{n:m} | u_{1:n-1}, y_{1:n-1}, \xv)  
\label{eq:jointConstrProp}
\end{equation}
Equation \ref{eq:jointConstrProp} makes explicit that in a simulation
from \( q_n \),
variables \( U_{1:n-1} \)
will have prescribed values \( y_{1:n-1} \)
with probability one. The conditional on the RHS of
\ref{eq:jointConstrProp} is only defined for
\( u_{1:n-1} = y_{1:n-1} \)
as all \( \uv \)
without this property have a joint probability of zero. We can thus
shorten the notation and use \( q_n(u_{n:m} | y_{1:n-1}, \xv ) \)
to refer implicitly to this conditional; likewise
\( q_n(u_n | y_{1:n-1}, \xv ) \)
refers to the conditional of variable $u_n$ when constraints
$y_{1:n-1}$ are imposed.  Finally, the case of \( n=1 \)
is defined to be such that no constraint is imposed. In other words,
\( q_1( \uv | \xv) = f(\uv|\xv) \), which is simply the original proposal.
The distinction between the conditionals of the original process
\[ \{ f_n( y_n | y_{1:n-1},\xv ) \} \]
and those of the constrained process
\[ \{ q_n( u_n | y_{1:n-1},\xv )\} \]
should be borne in mind; except at \( n= 1 \),
they are generally not the same. The \( \{ f_n \} \)
are simply those of the proposal when \( \Pcal \)
is run without any of the aforementioned constraining mechanisms
and random variables \( Y_{1:n-1} \)
\emph{happened to have} values \( y_{1:n-1} \).
In the constrained case, we are forcefully changing the ``dynamics''
of \( \Pcal \)
such that the constraints are met. In many situations, the two
families of conditionals can be expected to be close: for example if
\( \Pcal \)
corresponds to an MCMC process with invariant distribution
\( \pi \),
then both \( f_n( y_n | y_{1:n-1},\xv ) \)
and \( q_n( y_n | y_{1:n-1},\xv) \Rightarrow \pi( y_n | y_{1:n-1}) \)
as the process length increases.

Simulating constrained subsystems suggests an algorithm to sample from
\emph{and approximate} a surrogate to \(f\).
We now describe a rudimentary version, which we call the
\emph{Sequential Constraining Process} (SCP.) First, note that once
constraint $y_{1:n-1}$ is imposed, we are free in principle to
generate as much data as desired from $q_n(u_{n:m} | y_{1:n-1}, \xv )$
by repeatedly calling the constrained heuristic. In particular, a
population of sufficient size can be drawn so as to reliably estimate
the univariate conditional $q_n(u_n | y_{1:n-1}, \xv )$. This is in
distinction to sampling from $f_n(y_n|y_{1:n-1},\xv)$, which would, as
discussed, require a generally intractable rejection sampling
procedure. In this paper, to estimate $q_n(u_n|y_{1:n-1},\xv)$ we shall
consider only \emph{consistent} statistical estimators, i.e. those
converging in probability to the true value of the estimand with
increasing population size $N$. For the problems of present interest,
the univariate conditional estimation problem is simply that of
approximating a binomial success parameter given $N$ binary
observations. Examples of consistent estimators for this parameter are
the MLE histogram and its Bayesian analogues under a prior
distribution over the parameter values. Estimation will be discussed
in detail in Section \ref{sec:SCP:qnEstimators}, but for now, given a
population $\{\uv_n\samp{j}\}$ of size $N$ from
$q_n(u_{n:m} | y_{1:n-1}, \xv )$, we approximate the distribution of
the $n^{\textrm{th}}$ variable with the consistent estimator:
\[ \qhat_n(y_n | \{\uv_n\samp{j}\}, y_{1:n-1},\xv ) \approx q_n(y_n |
y_{1:n-1}, \xv ) \]
After constructing $\qhat_n$, we trivially sample a value $y_n$ from
the estimator and repeat the process for the next variable. Following
a suitable initialization, this methodology defines the SCP; it is
summarized in Algorithm \ref{alg:SCPBasic}. A complete binary state
$\yv$ is generated in this manner, and importantly, \emph{so is an
  approximation
  to its joint
  probability
  under the SCP.} More specifically, we have the joint approximation
\begin{equation}
  \qhat(\yv|\xv) = \prod_{n=1}^{m}\qhat_n(y_n | \{\uv_n\samp{j}\},
  y_{1:n-1},\xv )
  \label{eq:SCPProduct}
\end{equation}
Figure \ref{fig:proposalDiagram} visually shows the statistical
structure of the procedure for the case of sampling three binary
variables. We now see why it was necessary to distinguish between
variables $\Yv$ and $\Uv$. The latter, corresponding to variates
generated by simulating constrained subsystems, are used to construct
estimators, but are ultimately not the states returned by the
algorithm. The final ``usable'' states, $\Yv$, are sampled from the
estimators

We remark that there is considerable algorithmic flexibility here. The
heuristics used in different clamping stages need not have the same
parameters. Suppose, for instance, that a simulated annealing-type
process is used as the heuristic. As variables get sequentially
clamped, the state spaces effectively simulated over are of course
decreasing in size. One can accordingly use increasingly fast
annealing schedules for each process in the sequence. It is generally
the case that heuristics tend to more effectively probe the low-energy
structure of a problem landscape as the number of problem variables is
reduced.

The flexibility inherent to SCP is in fact more general than allowing
for the heuristic parameters to be tailored to the relative order in
the sequence: the formulation allows us to altogether move away from
reliance on a \emph{single} heuristic process for sampling. The
heuristics themselves invoked in the sequence are allowed to be
completely different, and may even depend on the particular state of
the system. A general rule is that if one can derive a favorable
heuristic for a system of a certain size, then one can also find
\emph{some} heuristic for a constrained subsystem; simulating small
systems is typically no harder than simulating larger ones. Hence, if
a more efficient heuristic is known for simulating a subsystem than
for simulating the full system, a SCP can be defined that applies the
appropriate heuristic to the system under consideration. The state
$\yv$ can in that case be interpreted as arising from a ``compound''
heuristic process. As an example of this sort of heuristic mixing,
suppose that at some point in the sequence which began using a local
search proposal, the set of variables has fractured into small enough
independent sets that \emph{exact} samples can be feasibly generated.
The exact sampling algorithm can then serve as the ``heuristic''
called recursively at that point, and the values of the corresponding
(exact) conditional probabilities can be used in place of an
approximately determined $\qhat_n$. In this introductory paper we
focus on presenting and validating the basic ideas, leaving in-depth
experimentation with large number of plausible heuristics to a later
analysis. Section \ref{sec:Discussion} discusses some additional
heuristic possibilities.

In summary, the advantage of the SCP methodology is that samples from
a favorable heuristic proposal (or class of proposals) can be
faithfully generated and their generation probabilities accurately
approximated. As discussed in Section \ref{sec:MCusingSCP}, this
enables the heuristics to be used for Monte Carlo estimation.

\begin{algorithm}
  \caption{Basic Sequential Constraining Process (SCP)}
  \label{alg:SCPBasic}
  \begin{algorithmic}
    \State \textbf{Initialize}
    \State \indent Draw $N$ samples $ \{ \Uv_1\samp{j} \}  \sim q(\uv|\xv ) $
    \State \indent Construct $\qhat_1(y_1 | \{ \uv_1\samp{j}\},\xv  )$ (see
    Section \ref{sec:SCP:qnEstimators} )
    \State \indent Sample \( Y_1 \sim \qhat_1(y_1 | \{
    \uv_n\samp{j}\}, \xv ) \)
    \Statex
    \For{$n = 2, \ldots, m $}
      \State Draw $N$ samples \( \{ \Uv_n\samp{j} \}  \sim q_n(u_{n:m} | y_{1:n-1},\xv) \)
      \State Construct
      \( \qhat_n(y_n | \{ \uv_n\samp{j}\}, y_{1:n-1},\xv ) \)
      (see Section \ref{sec:SCP:qnEstimators} )
      \State Sample \( Y_n \sim \qhat_n(y_n | \{
      \uv_n\samp{j}\}, y_{1:n-1},\xv ) \)
    \EndFor
  \end{algorithmic}
\end{algorithm}

\def\svgwidth{1.0\columnwidth}
\begin{figure}
\centering
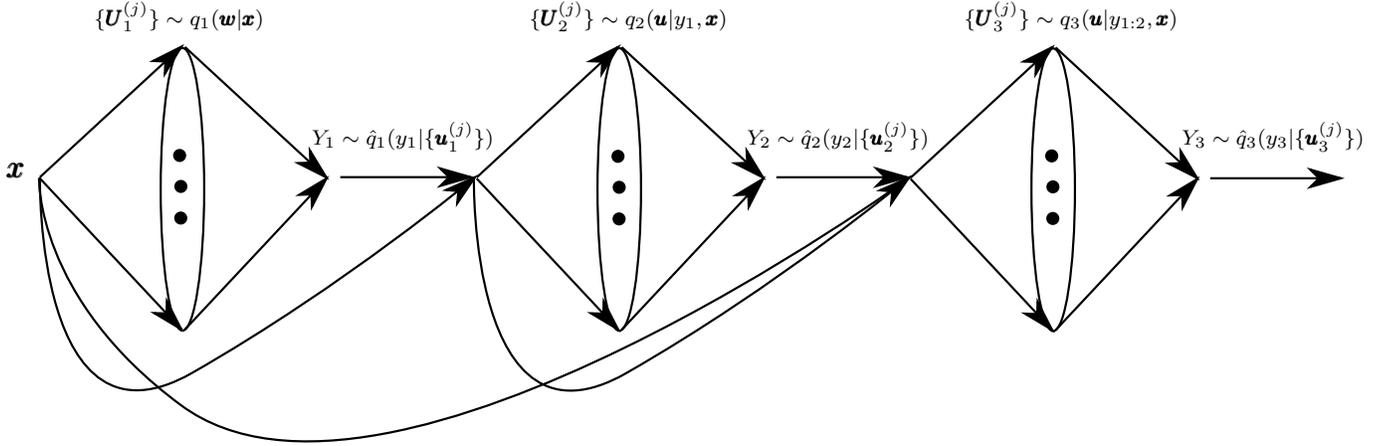
\caption{Illustration of the sequence of statistical operations
  implementing a simple variant of the \emph{Sequential Constraining Process}
  (SCP) to generate samples using an arbitrary stochastic process. At each
  stage, a possibly dependent population of samples $\{ \uv_n\samp{j}
  \}$ is generated by running the process with variables $y_{1:n-1}$
  clamped to their preceding values (none for the initial stage.) A consistent estimator
  $\qhat_n(y_n | \{ \uv_n\samp{j} \}, y_{1:n-1} ) $ is then constructed, and variable
  $y_n$ is sampled from this estimator. The result is a binary state
  $\yv$ and an approximate value of its probability
  $\qhat(\yv|\xv ) = \prod_{n=1}^3 \qhat_n( y_n|...) $ under the SCP.}
\label{fig:proposalDiagram}
\end{figure}

We presently discuss a few aspects of the SCP. First, we remark that
the members of the populations $\{\Uv_n\samp{j}\}$ generated by the
constrained runs of $\Pcal$ \emph{are not required to be statistically
  independent.} This is of particular importance if considering the
use of heuristics that involve \emph{interaction} among the samples.
Practically significant instances of such interaction are
\emph{stochastic
  resampling}
\cite{hukushima2003population,del2006sequential,machta2010population}
and \emph{genetic crossover}
\cite{john1992adaptation,goldberg2006genetic}.

The computational inefficiency of the presented variant of the SCP
should be apparent. In its described form, it requires sampling a
population of size $N$ for \emph{each discrete
  variable} in the target distribution to compute the estimators
$\{\qhat_n\}$ for the generation of a \emph{single sample} $\yv$.
Given that generating each sample for estimator construction requires
running a heuristic, even if the trivially parallel nature of the task
were exploited, this seems like a wasteful undertaking. Some speedup
can be straightforwardly achieved by estimating conditionals of
\emph{groups} (or \emph{blocks}) of variables; instead of estimating
binomial success probabilities, the task becomes that of
\emph{multinomial} estimation. But this approach needs to be applied
with caution: if the binary state spaces corresponding to the block
variables are too large relative to the population size $N$, the
resulting estimators $\qhat_n$ (where $n$ is now a block index rather
than a single variable index) will degrade rapidly, even if perfect
equilibrium samples were generated by the proposal. The compromise we
ultimately use, leveraging ideas of Bayesian robustness, is to
construct a discrete state space of approximately maximum size,
corresponding to a \emph{partition} of the set of binary outcomes
realizable by a block of variables, such that a maximum tolerable
statistical risk is never exceeded. Algorithmically, we strive to
build a maximum binary tree subject to a statistical loss constraint.
In Section \ref{sec:SSS}, we will present a novel, adaptive
methodology for discrete variable generation using the notion of
sequential constrained sampling that alleviates some of the
computational shortcomings using our robust Bayesian tree construction
and incorporating some ideas of \emph{state
  space
  search} from classical artificial intelligence. We reassure the
reader that an example clearly illustrating the combination of these
ideas will be presented in \ref{sec:SSS:comicBook}, and
apologize for the technical labyrinth that intercedes.

In Section \ref{sec:MCusingSCP}, we will show that one can in
principle perform correct Monte Carlo estimation with the na\"ive
variant of the SCP presented here using both importance sampling and
MCMC. In the next section, however, we take a diversion into the
important subject of construction of the estimators $\{\qhat_n\}$ used
in the SCP.

\subsection{Estimators of $q_n$ in SCP}
\label{sec:SCP:qnEstimators}

To ease the notation, we drop the dependence of the heuristic on $\xv$
in this section, as we do when discussing importance sampling in
Section \ref{sec:MCusingSCP:ImportanceSampling}, where typically the
heuristic distribution is not dependent on the initial state. The
dependence will recur when we present SCP-based MCMC in Section
\ref{sec:MCusingSCP:MCMC}.

\subsubsection{Preliminaries}
\label{sec:SCP:qnEstimators:Preliminaries}
Suppose we possess a collection of $N$ samples $\{\uv\samp{j} \}$. Let
$I$ be a subset of the variable indices, i.e.
$I \subseteq \{ 1, \ldots, m\}$, and $y_I$ refer to a binary state on
the corresponding variables. We define the counting function
\begin{equation}
\#[ y_I ; \{ \uv\samp{j} \}] \triangleq \sum_{j=1}^{N} \indic_{y_I}[ u_I\samp{j}]
\label{eq:countingFunction}
\end{equation}
which simply indicates the number of times the binary string $y_I$
occurs in the samples $\{ \uv\samp{j} \}$.

Consider the case in which we seek to approximate the univariate
function $q_n( y_n | y_{1:n-1})$.  A possible choice is the
\emph{histogram} or \emph{maximum likelihood estimator} (MLE)
\begin{equation}
\qhat_n( y_n | \{ \uv_n\samp{j} \}, y_{1:n-1}) = \frac{\#[ y_n ; \{ \uv_n\samp{j} \}]}{N}
\label{eq:histMLE}
\end{equation}
An undesirable property of the MLE in the present context is that it
assigns zero success probability when the observed
success count is zero. This can cause serious theoretical and
practical issues for our algorithms. 

Bayesian estimators, which obtain from minimizing the \emph{Bayes
  risk} over a suitable \emph{loss function} (often the \emph{squared
  loss}) under a hypothesized \emph{prior}, are appealing because for
many prior distributions, they never evaluate to zero for a finite
number of observations. It is not our intention to delve into the
details of Bayesian analysis here (see
e.g. \cite{bernardo2009bayesian,berger2013statistical}) but we recall
a few relevant facts. When performing inference on the success
probability of binomially-distributed data, it is common to assume
that the parameter follows a \emph{beta distribution} with two
non-negative \emph{hyperparameters.} The \emph{posterior} distribution
of the success probability will then also be beta. In our situation,
let $\alpha(y_n)$ refer to the prior parameter corresponding to
outcome $y_n \in \{+,-\} $.  The estimator that minimizes the
posterior expected \emph{squared loss} is then
\begin{equation}
\qhat_n( y_n | \{ \uv_n\samp{j} \}, y_{1:n-1}) = \frac{\#[ y_n ; \{ \uv_n\samp{j} \}] + \alpha(y_n)  }{N + \alpha(+) + \alpha(-)}
\label{eq:BayesEstBinom}
\end{equation}
Many other loss functions are possible; in this work, we use the
\emph{Kullback-Leibler loss} as discussed and justified in Section
\ref{sec:SCP:qnEstimators:RobustBayes}. The Bayes estimator with
respect to that loss happens to be the same as that under the squared
loss.

It is common to set $\alpha(+)=\alpha(-)$ in the absence of
information suggesting a preference towards either outcome; within
this family of \emph{symmetric} (or ``non-informative'') beta
distributions are the \emph{Laplace} prior ($\alpha(+)=\alpha(-)=1$)
corresponding to a uniform distribution on the success parameter, and
\emph{Jeffreys'} prior ($\alpha(+)=\alpha(-)=0.5$), which leaves the
posterior distribution invariant under a monotone transformation of
the success parameter and is hence termed the \emph{objective} prior
for the binomial parameter.

The Bayesian binomial analysis is straightforwardly generalized to the
multinomial case; the $K$-category extension of the beta distribution
is the \emph{Dirichlet distribution} with nonnegative hyperparameter
vector $\alphav \in \Reals^K$. When the set of outcomes,
defined as $\Ycal_I$, consists of all $2^{|I|}$ possible binary
configurations on variables $I$, let $\alpha(y_I)$ be the prior
hyperparameter corresponding to particular string $y_I$.  The
analogous Bayes estimator to (\ref{eq:BayesEstBinom}) for the
probability $q_n( y_I | \{ \uv_n\samp{j} \}, y_{1:n-1})$ is
\begin{equation}
\qhat_n( y_I | \{ \uv_n\samp{j} \}, y_{1:n-1}) = \frac{\#[ y_I ; \{ \uv_n\samp{j} \}] + \alpha(y_I)  }{N + \sum_{y_I}\alpha(y_I)}
\label{eq:BayesEstMultinom}
\end{equation}
Comparing the Bayes estimators (\ref{eq:BayesEstBinom}) and
(\ref{eq:BayesEstMultinom}) to the histogram estimator
(\ref{eq:histMLE}), it is commonly remarked that the Dirichlet prior
parameters can be interpreted as \emph{pseudo-counts}, or
\emph{pseudo-observations} whose influence diminishes as the amount of
data $N$ increases. 

The notion of an ``objective'' prior is considerably more involved
for the multinomial case \cite{good1965estimation}, and over the
years, several choices of $\alphav$ have been proposed. Recent work
\cite{de2011inference} has suggested that \emph{weak priors} are
preferable for multinomial problems as otherwise the set of possible
parameters for which the Bayes estimator is preferable (with respect
to squared loss) to the MLE shrinks rapidly.

The strength of the prior is particularly relevant for our algorithms
when some outcomes of the set are unobserved. An outcome $y_I$ with an
empirically-observed count of zero will have its probability
approximated as
\[
\qhat_n( y_I | \{ \uv_n\samp{j} \}, y_{1:n-1}) = \frac{\alpha(y_I)  }{N + \sum_{y_I}\alpha(y_I)}
\]
The parameter $\alpha(y_I)$ can thus be seen to control the smoothness
of the estimator; larger values of $\alpha$ result in increased
probability that the states sampled by the SCP will be unobserved in
the proposal's population, while smaller values yield increased
fidelity to the empirical statistics. Hence it is undesirable to make
$\alpha$ too large relative to $N$ as such as choice will corrupt an
inherently favorable proposal. On the other hand, some reasonable
``leakage'' probability to unseen states is desirable and indeed, as we
shall see in Section \ref{sec:MCusingSCP:ImportanceSampling},
\emph{required} to make the importance sampling estimators
theoretically correct.

An additional appealing property of the Bayesian framework is that
updating the estimators as new observed data arrive is
straightforward. More precisely, suppose that we \emph{store}, for the
clamping condition $y_{1:n}$, the counts occurring in the samples
$\{ \uv_n \samp{j} \}$ from a batch of $N$ constrained process runs;
if we draw another set $\{ \vv_n \samp{j} \}$ of size $N$ from the
same constrained process, we simply evaluate the updated estimator as
\[
\qhat_n( y_I | \{
  \uv_n\samp{j}\} \cup  \{\vv_n\samp{j} \}, y_{1:n-1}) = \frac{\#[ y_I ; \{
  \uv_n\samp{j}\}]  + \#[ y_I ;  \{\vv_n\samp{j} \}] + \alpha(y_I)  }{2N + \sum_{y_I}\alpha(y_I)}
\]
storing the modified set of counts for the next update. In the SCP
framework, this implies that one has the option of interleaving the
operations of \emph{generating} samples from estimators available at a
given time with periodic (possibly randomized) \emph{refinement} of
estimator accuracy by Bayesian updating following drawing more samples
from the appropriate constrained process. This can be understood as a
form of dynamical \emph{querying} of the unknown proposal
distribution.

\subsubsection{Robust Bayesian Estimators}
\label{sec:SCP:qnEstimators:RobustBayes}

In distinction to the usage of non-informative priors, an alternative
method of setting $\alpha$, discussed in detail in
\cite{hamze2015robust} for multinomial inference, appeals to the
notion of \emph{Bayesian robustness}\cite{berger1984robust}. In the
robust Bayesian framework, one considers the effect of prior
misspecification on an estimator's quality as the prior is allowed to
vary within some family. This departs from the traditional Bayesian
methodology but nonetheless yields estimators with appealing
properties. In this work, we will present a simple and
efficiently-computable estimator that is not necessarily optimal, but
works well enough for the purposes of the Monte Carlo algorithms. See
\cite{hamze2015robust} for more sophisticated options and an in-depth
discussion.

We define the \emph{single-pseudocount family} $\Gamma$ as the set of
Dirichlet distributions\footnote{We abuse notation and identify a
  prior distribution with its parameters} over outcome probabilities
$\{q(y_I)\}$ whose parameters $\alphav$ sum to 1:
\begin{equation}
\Gamma = \Big \{ \alphav : \sum_{y_I \in \Ycal_I} \alpha(y_I) = 1 \Big \}
\label{eq:Gamma}
\end{equation}
Assuming that the prior distribution lies within this family is not
necessarily justified when the data are very sparse, i.e. most
outcomes are observed to occur at most a few times. In our case
however, we dynamically construct tree-based \emph{partitions} over
the set of outcomes (Section
\ref{sec:SCP:qnEstimators:DynPartConstruction}), which has the effect
of controlling the sparsity. In that context, the family defined in
(\ref{eq:Gamma}) is quite appropriate. We note that the well-known
non-informative prior due to Perks \cite{perks1947some}, with
$\alpha(y_I) = \frac{1}{|\Ycal_I|}$ for all $y_I$, belongs to
$\Gamma$.

Given an estimator of outcome probabilities, the worst-case
\emph{posterior Kullback-Leibler (KL) loss} over the priors in
$\Gamma$ provides a conservative measure of its reliability. In
particular, suppose we construct an estimator $\qhat(y_I | \alphav )$
over the outcomes $\Ycal_I$, with $\alphav \in \Gamma$. The worst-case
KL loss is then defined to be
\begin{equation}
  \rhostar(\alphav) = \sup_{\alphav' \in \Gamma} \rho(
  \alphav', \alphav )
\label{eq:worstCaseKLLossDefn}
\end{equation}
with the posterior loss
\begin{equation}
  \rho( \alphav', \alphav) = \Expect{q}\Big[
  \log\frac{q}{\qhat}\Big] 
  \label{eq:KLLoss}
\end{equation}
For priors in the Dirichlet family, (\ref{eq:KLLoss}) can be shown to
be
\begin{equation}
  \rho( \alphav', \alphav) = \sum_{y_I \in \Ycal_I}
    \frac{\alpha'(y_I)+\#( y_I) }{1+N}\Big[
    \psi\big[\alpha'(y_I) +\#(y_I) +1\big]
    - \psi\big[1+N+1\big] - \log
    \frac{\alpha(y_I)+\#(y_I)}{1+N}  \Big]
  \label{eq:KLLossGamma}
\end{equation}
where $\psi(.)$ is the \emph{digamma function}
\cite{abramowitz1964handbook}. The loss (\ref{eq:KLLossGamma})
represents, up to an additive constant, the posterior expected bits of
divergence between the Bayes estimator under \emph{assumed} prior
$\alphav \in \Gamma$ and the true $q_n( . | y_{1:n-1})$ if the latter
were \emph{actually} a priori Dirichlet-distributed with parameters
$\alphav' \in \Gamma$; a smaller loss implies a more reliable
estimator. This formulation enables an analysis of estimator
robustness. The \emph{worst-case} loss defined in Equation
(\ref{eq:worstCaseKLLossDefn}) yields the maximum divergence, over all
Dirichlet priors whose parameters lie on the \emph{standard simplex},
between the true distribution and the estimator constructed by assuming
$\alphav$. In principle, one could seek an estimator that minimized
(\ref{eq:worstCaseKLLossDefn}). Unfortunately, computing such an
estimator is non-trivial, though still ``tractable.'' In this work, we
propose a simple approximation to the Bayes estimator that would yield
minimum worst-case loss over priors in $\Gamma$.  Consider the
empirical counts of all outcomes $\{ y_I\}$, and let $\#\samp{1}$
refer to the smallest among them, for example zero. Let
$\Ycal\samp{1}$ be the set of outcomes with count $\#\samp{1}$. We use
a Bayes estimator as defined in (\ref{eq:BayesEstMultinom}) with
parameters
\begin{equation} \alpha(y_I) = \left \{ \begin{array}{ll}
                                              \frac{1}{|\Ycal\samp{1}|} & y_I \in \Ycal\samp{1} \\
                                              0 & \textrm{otherwise}
                         \end{array}
                       \right.
\label{eq:gammaLFPAlphas}
\end{equation}
When there are empty outcomes, i.e. $\#\samp{1} = 0$, the probability
of generating \emph{some} outcome $y_I \in \Ycal\samp{1}$ is simply
\[
\frac{1}{1+N}
\]
with the individual outcomes in $\Ycal\samp{1}$ occurring uniformly.

To quantify reliability, we would like to evaluate
(\ref{eq:worstCaseKLLossDefn}) under our estimator.  It can be shown
that when using $\alphav$ as prescribed in (\ref{eq:gammaLFPAlphas})
the worst-case loss is achieved either when $\alpha'(y_I) = 1$ for any
single member $y_I$ of $\Ycal\samp{1}$ or for one in the set
$\Ycal\samp{2}$ of those occurring with \emph{second-lowest count},
with the rest set to zero. The maximum can therefore be computed by
evaluating (\ref{eq:KLLossGamma}) for these two choices of $\alphav'$
and taking the larger of the two values. By using appropriate data
structures, this computation is relatively inexpensive.

The KL divergence allows natural comparison between losses on domains
of \emph{different cardinality.} Given, for example, a subset $I$ of 4
variable indices, implying that $|\Ycal_I| = 16$, and another one $J$
of 5 indices, so that $|\Ycal_J| = 32$, the posterior KL losses in
both cases are nonetheless directly commensurable. This feature is
especially relevant when we dynamically determine an approximately
optimal partition of the state space subject to a loss constraint as
described in the state space sampling algorithm presented in Section
\ref{sec:SSS}. Within a sampling context, the divergence has an
additional relevant aspect: it can readily be shown by expanding the
logarithm to first order that the KL divergence is approximately
(half) the importance weight \emph{variance} resulting from using the
estimator as a proposal to the true distribution.

The strategy to construct a fine partition of the outcome set subject
to a constraint on the KL loss, which has the interpretation of imposing
a binary tree structure on the elements, is outlined in Section
\ref{sec:SCP:qnEstimators:DynPartConstruction}. This strategy is
leveraged in the sampling algorithm of Section
\ref{sec:SSS} to maximally exploit the set of heuristic samples
$\{ \Uv\samp{j} \}$ in the approximation construction so as to
minimize the number of required calls to the heuristic. The next
section discusses variance reduction and can be skipped on first
reading.

\subsubsection{Variance Reduction by Rao-Blackwellised Estimators}
\label{sec:SCP:qnEstimators:RaoBlack}
In several commonly-encountered situations, the heuristic process
follows dynamics such that the systems' individual variables are
assured of being in a \emph{local} equilibrium distribution.  This is
certainly so for simulated annealing, and can be straightforwardly
imposed for more general processes. The property can be exploited to
derive probability estimators of tremendously reduced variance, which
would improve the Monte Carlo estimates yielded by our
algorithm. Variants of this idea are used in several statistical
applications, and are collectively referred to as
\emph{Rao-Blackwellisation}\cite{casella1996rao} due the appearance of
the same principle in the fundamental Rao-Blackwell theorem of
estimation theory.

In our context, the method informally consists of replacing the
count-based Bayes estimators of the type discussed in
(\ref{eq:BayesEstMultinom}) with ones based on averaging the
\emph{exact local probabilities}, which can be computed
easily. To use the Ising model as an example, given any
state $\yv \in \Scube^{m}$, the local equilibrium probability of
variable $i$ is given by
\begin{equation}
  \Pr(x_i | \yv ) = \Pr(y_i|\yv_{\nbrs{i}} ) = \frac{ e^{-y_i (h_i +
      \sum_{j\in\nbrs(i)} J_{ij} y_j ) }}{ e^{(h_i +
      \sum_{j\in\nbrs(i)} J_{ij} y_j )} + {e^{-(h_i +
        \sum_{j\in\nbrs(i)} J_{ij} y_j )} } }
  \label{eq:isingLocalEq}
\end{equation}

We propose constructing the Rao-Blackwellised estimator to the
probability $q_n( y_I | y_{1:n-1})$ over the variable set $y_I$ in a
\emph{top-down} manner as follows. Let $(i_1, \ldots, i_{|I|})$ be a
sequence consisting of the elements of $I$, and define the sequences
\begin{displaymath}
I_k =
\left \{ \begin{array}{l l} 
1,\ldots,n-1 & \textrm{for } k = 1\\
1,\ldots,n-1,i_1,\ldots,i_{k-1}& \textrm{for } k > 1\\
\end{array}
\right.
\end{displaymath}
We define the (smoothed) Rao-Blackwellised estimator to be
\begin{equation}
  \qhatrb_n( y_I | \{ \uv_n\samp{j} \}, y_{1:n-1}) = \prod_{k=1}^{|I|} \qhatrb_n(
  y_{i_k} |  \{ \uv_n\samp{j} \}, y_{I_k}  )
  \label{eq:RBDefn}
\end{equation}
where the elements of the product in (\ref{eq:RBDefn}) are
\begin{equation}
  \qhatrb_n(y_{i_k}|\{ \uv_n\samp{j} \}, y_{I_k}) = 
  \frac{\sum_{j}\Pr(y_{i_k} | \uv_n\samp{j}) + \alpha(y_{i_k})  }{\#(y_{I_k}) + \sum_{y_I}\alpha(y_I)}
  \label{eq:RBProdElements}
\end{equation}
The summation in the numerator of (\ref{eq:RBProdElements}) is over
all sample indices $j$ such that variables $I_k$ of $\uv_n\samp{j}$
agree with the conditioning sets $y_{I_k}$, and $\#(y_{I_k})$ is the
total number of occurrences of $y_{I_k}$ in $\{ \uv_n\samp{j} \}$.

Due to the requirement of evaluating the conditional probabilities,
the Rao-Blackwellised estimator is somewhat more costly to compute
than the count-based Bayes estimator, but the corresponding
improvement in accuracy can be well worth the effort, as we will show
in Section \ref{sec:Experiments}. Indeed we remark that for systems at
very high temperature, where the multinomial sampling variance is
expected to be at its maximum, the Rao-Blackwellised estimator returns
almost perfect estimates of the true probabilities.

\subsubsection{Dynamical Construction of a Partition}
\label{sec:SCP:qnEstimators:DynPartConstruction}

We now generalize the problem of Bayesian probability estimation to
the form that will be used in the algorithm. Suppose we have clamped a
set of variables and run the constrained proposal on the remaining
variables with indices $I \subseteq \{ 1, \ldots, m \}$.  We seek to
build a distribution on as large a portion as possible of the
unconstrained state space using the samples $\{\uv\samp{j}\}$. For
readers wishing to skip the pedantic details of this section on first
reading, the idea can be summarized as follows: continue branching a
binary tree, whose leaf nodes represent partial states, computing an
estimator for the probabilities of the tree nodes using the sample
statistics, until the KL loss described in Section
\ref{sec:SCP:qnEstimators:RobustBayes} grows unacceptably large. An
example of such a tree appears in Figure \ref{fig:exSubcubeTree}

Consider the augmentation of full state space
$\Scube^m$ with the value $\Box$, meaning \emph{unassigned} (or
\emph{free}.) We define a \emph{partial state}
$ \sigmav \in \{+,-,\Box\}^m$. For example, if $m=4$, then
$\sigmav = -\Box +\Box$ corresponds to the binary state with variables
$1,3$ set to $-,+$ respectively and the rest unspecified, while
$\sigmav = \Box\Box\Box\Box$ means no variables are assigned. Each
partial state $\sigmav$ thus naturally defines a \emph{binary subcube}
$A(\sigmav)$, i.e. a set of all binary states with some subset of the
variables having fixed values. For example:\\
\begin{displaymath}
A(-\Box +\Box) = \left \{ \begin{array}{l l} 
-\fbx{+}+\fbx{+} & \qquad -\fbx{+}+\fbx{-} \\ \\
-\fbx{-}+\fbx{+} &  \qquad -\fbx{-}+\fbx{-} 
\end{array} \right \}
\end{displaymath}
where the boxes over the instantiated variables in the subcube clarify
that they were free in $\sigmav$.  Clearly the extreme cases of
$\sigmav=\{\Box\}^m$ and $\sigma_i \in \{+,-\}$ for all $i$
correspond, respectively, to subcubes of size $2^m$ and $1$
respectively. Where there is no ambiguity, we will refer to the binary
subcube $A(\sigmav)$ by its corresponding partial state
$\sigmav$. Furthermore, we can extend our previous \emph{count}
definition to that of a partial state $\sigmav$ given samples
$\{ \uv\samp{j} \}$ as \[\#[ \sigmav ; \{ \uv\samp{j} \}] \]
which returns the number of times that the specific configuration of
all \emph{assigned} variables in $\sigmav$ occurs in
$\{ \uv\samp{j} \}$

Returning to the estimation problem, define $\sigmav_R$ to be the
\emph{root} partial state corresponding to all possible outcomes of
the constrained proposal with free variables $y_I$. We desire a
maximum \emph{partition} $\Acal$ of the subcube $A(\sigmav_R)$, i.e as
large a collection as possible of pairwise disjoint subcubes
$\Acal = \{ A_i \}$ with $\cup_i A_i = A(\sigmav_R)$. The elements of
the partition are the outcomes of a discrete distribution whose
probabilities we wish to estimate. The limiting factor to the size of
the partition is the statistical quality of the resultant estimator;
if it gets too large, the quality will degrade unacceptably. We
represent the elements comprising a partition as \emph{leaf nodes} of
an \emph{incomplete} binary tree $T$, i.e. one whose leaf nodes can
lie at different depths from the root. The tree is also \emph{full}:
all its nodes have either zero or two children. Let $\nu_R$ refer to
the root node of $T$. Each internal node of the tree corresponds to
the union of its two children's subcubes; hence subcubes decrease in
size from the root to the leaves. Each node $\nu \in T$ is identified
with partial state $\sigmav(\nu)$. Furthermore, each internal
(non-leaf) node is mapped to a \emph{branch variable}, the unassigned
variable in the node's partial state that gets fixed in its two
children. The branch variable of node $\nu$ is defined as
$v(\nu) \in I $ with $v(\nu)=\emptyset$ if and only if $\nu$ is a leaf
node. An example of such a structure, which we call a \emph{subcube
  tree}, illustrating these rather natural ideas is shown in Figure
\ref{fig:exSubcubeTree}. Readers familiar with the ideas of
\emph{branch-and-bound} algorithms will recognize that the subcube
tree is simply a kind of \emph{search tree} on a binary state space;
our presentation is to emphasize the aspects that will be used in the
algorithm.

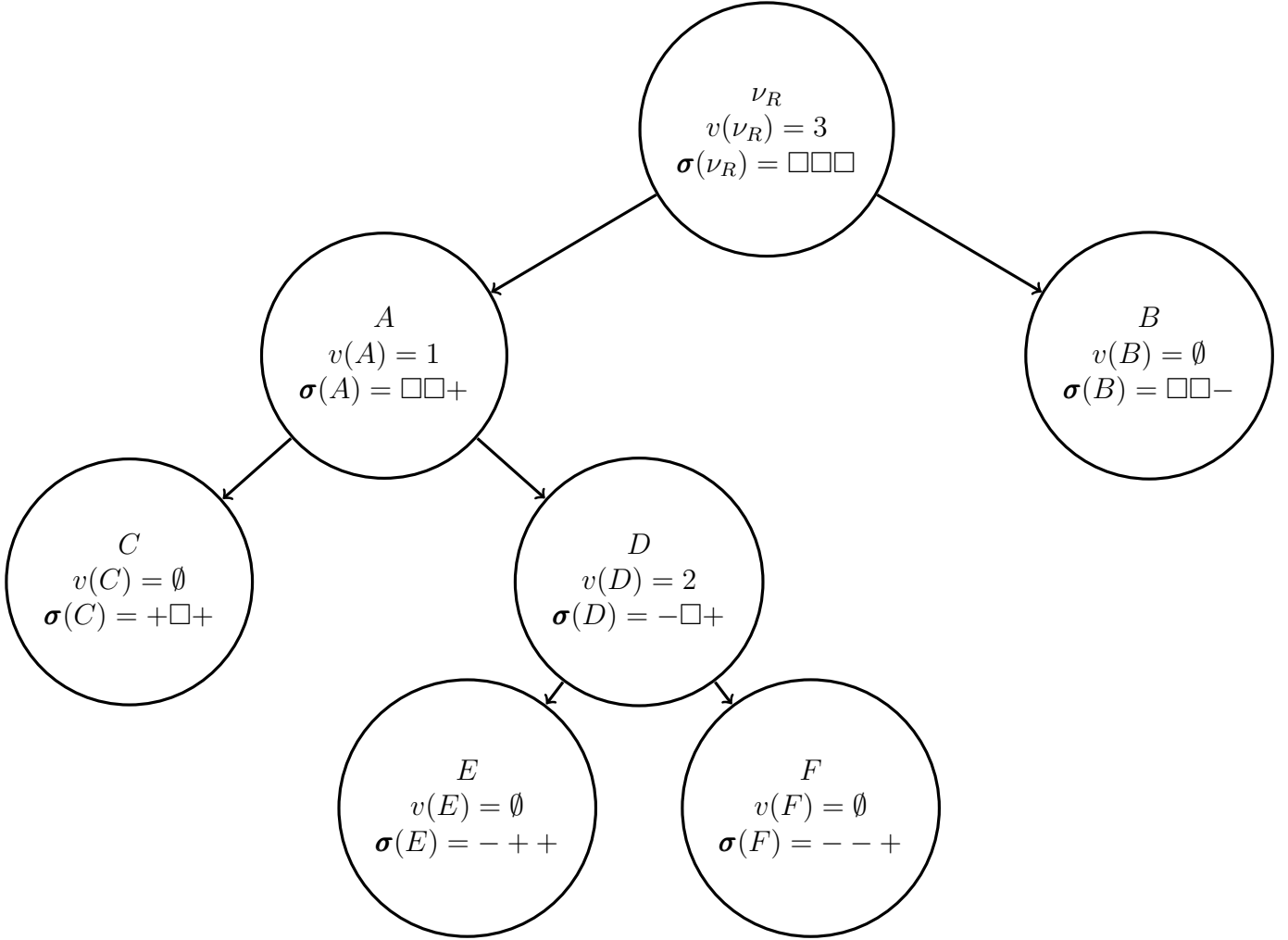
\begin{figure}[h!]
\centering
\begin{tikzpicture}[every tree node/.style={draw,circle,very thick},
  edge from parent/.append style={->,very thick},
  level distance=3.25cm,sibling distance=1.2cm, 
  edge from parent path={(\tikzparentnode) -- (\tikzchildnode)} ]
  \Tree 
    [.\node[text=black](nuRoot){$\begin{array}{c} \nu_R \\ v(\nu_R)=3 \\
\sigmav(\nu_R) = \Box\Box\Box \end{array}$};
[.\node[text=black](A){$\begin{array}{c}  A \\
v(A)=1 \\ \sigmav(A) =
\Box\Box + \end{array}$};
[.\node[text=black](C){$\begin{array}{c}
C \\
v(C)=\emptyset \\
\sigmav(C) = + \Box +
\end{array}$}; ]
[.\node[text=black](D){$\begin{array}{c}
D \\
v(D)=2 \\
\sigmav(D) = - \Box +
\end{array}$}; 
[.\node[text=black](E){$\begin{array}{c}
E \\
v(E)=\emptyset \\
\sigmav(E) = - + +
\end{array}$}; ]
[.\node[text=black](F){$\begin{array}{c}
F \\
v(F)=\emptyset \\
\sigmav(F) = - - +
\end{array}$}; ]
]
]
[.\node[text=black](B){$\begin{array}{c}
B \\
v(B)=\emptyset \\
\sigmav(B) = \Box\Box -
\end{array}$}; ]
]
\end{tikzpicture}
\caption{An example subcube tree for a state space corresponding to a
  problem with $m=3$ binary variables. Each node contains the node
  label, its branching variable $v$, and its corresponding partial
  state $\sigmav$. The symbol $\Box$ denotes an unassigned variable in
  the length $M$ state. The branch variables for all leaf nodes are
  empty. For internal nodes, $v$ specifies which variable is assigned
  in its two children. For example the left and right children $A$ and
  $B$ of root node $\nu_R$ have in their partial states variable 3 set
  to $+$ and $-$ respectively. Only nodes $E$ and $F$ correspond to
  fully-defined states. The state space partition $\Acal$ represented
  by this tree consists of the sets of states corresponding to the
  leaves. In particular,
  $\Acal = \{ +\Box + \enskip,\enskip -++\enskip,\enskip --+\enskip,\enskip \Box\Box -\}$.}
\label{fig:exSubcubeTree}
\end{figure}

To construct the tree along with the distribution over the outcomes,
we begin with a trivial tree $T_0$ containing a single node $\nu_R$
and corresponding to the vacuous partition
$\Acal_0 = \{ \sigmav_R \}$. We also maintain a \emph{priority queue}
$\Fcal$ which, inspired by the terminology of state space
search\cite{russell1995artificial,zhang1999state}, we call the
\emph{frontier}. The queue initially contains only the root node, i.e.
$\Fcal_0 = \{ \nu_R \}$. Generally, the frontier consists of leaf
nodes whose partial states are considered valid for further refinement
(or \emph{branching}.) A leaf node in the tree is called \emph{open}
if it appears in the frontier. The priorities of the nodes in $\Fcal$
correspond to their ranks when the corresponding partial states are
sorted by \emph{decreasing} count $\#[\nu ; \{ \uv\samp{j} \}]$, i.e.
more populated partial states have a higher branching
priority. Algorithms to implement priority queues are standard in
computer science \cite{cormen2009introduction}.

The algorithm thus begins. If the current frontier $\Fcal_i$ is empty,
i.e. there are no further available nodes to refine, terminate the
algorithm. Otherwise, \emph{dequeue} the most populated open leaf node
in the tree; in other words, $\nu'$ of highest priority is removed
from $\Fcal_i$. Next, from the corresponding partial state
$\sigmav(\nu')$, branch on one of the \emph{unassigned} variables
either at random, or according to some predetermined
order\footnote{Branching according to some ``greedy'' criterion, as is
  done when constructing classification trees, is not suitable in the
  present context as it will heavily bias the probability
  estimates. Classification is quite a different problem from
  parameter estimation.}.  We then compute the estimator over the
potential new partition $\Acal_{i+1}$ formed by removing
$\sigmav(\nu')$ from $\Acal_i$ and adding
$\{\sigmav(\nu_+), \sigmav(\nu_-)\}$, where $\sigma_v( \nu_+ ) = +$
and $\sigma_v(\nu_-) = -$, i.e. unassigned variable $v(\nu')$ now
explicitly takes its two possible values. The estimator at any stage
may be computed as described in Section
\ref{sec:SCP:qnEstimators:RobustBayes}. Specifically, for leaf nodes
$\nu$ corresponding to partition $\Acal_i$, employ the Bayes estimator
(\ref{eq:BayesEstMultinom}) with
\begin{equation}
\alpha(\nu) = \left \{ \begin{array}{ll}
                           \frac{1}{|\Acal_{i}\samp{1}|} & \nu \in \Acal_i\samp{1} \\
                           0 & \textrm{otherwise}
                         \end{array}
                       \right.
  \label{eq:partitionLFPAlphas}
\end{equation}
where $\Acal_i\samp{1}$ refers to the set of all members of $\Acal_i$
with minimum count $\#[\nu ; \{ \uv\samp{j} \}]$.  Alternatively, the
Rao-Blackwellised estimator discussed in Section
\ref{sec:SCP:qnEstimators:RaoBlack} given by Equation (\ref{eq:RBDefn}) using
the partition hyperparameters defined in (\ref{eq:partitionLFPAlphas})
may be preferable to lower the variance.

The cost of a partition $\Acal_i$ is defined to be the worst-case
posterior KL loss (\ref{eq:worstCaseKLLossDefn}) computed by
evaluating (\ref{eq:KLLossGamma}) with $\alpha'(\nu) = 1$ for an
arbitrary member of $\Acal_i\samp{1}$ with the remaining
$\alpha'(\nu)=0$, evaluating it for $\alpha'(\nu) = 1$ for an
arbitrary member of $\Acal_i\samp{2}$ with the rest set to zero, and
taking the maximum over the two cases. Judicious application of data
structures can considerably speed up this evaluation; this will not be
discussed here to avoid cluttering the presentation.

If the KL cost of the potential new estimator is below the specified
threshold, we accept the branch and attempt to make another
refinement; otherwise, we terminate with partition $\Acal_i$ and
estimator corresponding to $\{ \alpha(A_j) \}$ for the partition
elements $A_j \in \Acal_i$.  In the event that we continue, the
children $\nu_+$ and $\nu_-$ are only considered candidates for
further branching, if, of course, they do not correspond to full
binary states. Furthermore, in this algorithm, we do not allow
branching into aggregations whose counts fall below a threshold. We
may, for example, enforce that if $\#[\nu ; \{ \uv\samp{j} \}] = 0$
for leaf node $\nu$, then it is not eligible for further refinement.
If a node is considered valid for refinement, it is \emph{enqueued},
i.e. added by priority into $\Fcal_i$.

We now return to the task of showing that correct Monte Carlo
simulation is theoretically feasible using the SCP that uses the
approximations of the type we have discussed in this section.

\section{Monte Carlo Estimation via SCP}
\label{sec:MCusingSCP}
In the previous section, a method for sequential sampling from an
arbitrary discrete stochastic heuristic was presented such that the
resultant samples' distribution can be reasonably estimated. We have
not yet discussed, however, whether usage of approximate values of the
probabilities voids the statistical correctness of the standard
sampling methods. We now present importance sampling and MCMC methods
that are theoretically correct.

\subsection{Importance Sampling}
\label{sec:MCusingSCP:ImportanceSampling}

Consider first the case where one is interested in \emph{reweighting}
the samples generated by the SCP so as to allow statistical averages
to be computed. This is in distinction to the \emph{dynamical} MCMC
approach that will be discussed in Section \ref{sec:MCusingSCP:MCMC}. For
clarity, we will consider the one-variable-at-a-time SCP; extension to
the more general cases of constructing approximations over blocks
(groups) of variables and using the dynamical partitioning methods
discussed in Section \ref{sec:SCP:qnEstimators:DynPartConstruction} is left as an exercise.

Suppose we have generated a population of $N$ binary configurations
$\{\yv\samp{i}\}$ using the SCP Algorithm \ref{alg:SCPBasic}. We define the
importance weight of the $i^{th}$ sample
\begin{equation}
\label{eq:NDMCImpWeights}
w(\yv\samp{i}, \{ \uv_{1:m} \samp{j} \} ) = \frac{\pi(\yv\samp{i})}{
  \qhat_1(y_1\samp{i}|\{\uv_{1}\samp{j}\} )
  \qhat_2(y_2\samp{i}|\{\uv_2\samp{j}\}, y_{1}) \ldots
  \qhat_m( y_m\samp{i} |\{\uv_m\samp{j}\},y_{1:m-1} )}
\end{equation}
We then have the following straightforward result.
\begin{proposition}
  Suppose the product of the approximating distributions
  $\prod_{n=1}^{m} \qhat_n( y_n | \{ \uv \samp{j} \}, y_{1:n-1} ) \neq 0$ for all possible sets of
$\{\uv_{1}\samp{j}\}, \ldots, \{\uv_{m}\samp{j}\}$ wherever $\pi(\yv)
\neq 0 $. Define the \emph{importance sampling estimator}
\begin{equation}
\hat{h} = \frac{1}{N} \sum_{i} w(\Yv\samp{i}, \{ \Uv_{1:m}\samp{j} \}) h(\Yv\samp{i}) 
\label{eq:NDMCImpSamp}
\end{equation}
Then (\ref{eq:NDMCImpSamp}) is an unbiased estimator of $\Expect{\pi}[h(\Yv)]$.
\end{proposition}
\begin{proof}
Consider the term $ w(\Yv, \{ \Uv_{1:m}\samp{j} \}) h(\Yv) $ and take
its expectation over all its underlying random variables:
\[
\Expect{}[w(\Yv, \{ \Uv_{1:m}\samp{j} \}) h(\Yv)] = 
\sum_{\{ \uv_{1:m}\samp{j} \}} \sum_{ \yv } \Bigg[  w(\yv, \{
\uv_{1:m} \samp{j} \} ) h(\yv) 
\prod_{n=1}^{m} q_n(\{\uv_{n}\samp{j}\}|y_{1:n-1})
\prod_{n=1}^{m}\qhat_n(y_n | \{\uv_{n}\samp{j}\}, y_{1:n-1} ) \Bigg ]
\]
By our assumption, $w$ is finite for all $\yv$ and $\uv$ used to
construct $\qhat$. Cancellation of the $\qhat$ terms and summing
out terms depending on $\uv_{1:m}$ yields the result.
\end{proof}

We now observe why Bayesian estimators are of interest in this work; their
property of ensuring nonzero probability for all outcomes ensures that
their support will always include that of the target $\pi$. Usage of
the MLE as the approximating distributions $\qhat_n$ will, in theory,
lead to a biased estimator.

Proposition \ref{eq:NDMCImpSamp} states that the approximation is
correct in expectation, but the \emph{accuracy} of an IS estimator is
typically measured by its \emph{variance.} Although the variance of
the estimator will depend on the function $h$, it is often instructive
to consider the variance of the importance weights themselves:
\begin{equation}
\textrm{var}\Bigg[ w(\Yv, \{ \Uv_{1:m}\samp{j} \}) \Bigg] =
\sum_{\{ \uv_{1:m}\samp{j} \}} \sum_{ \yv } \frac{\pi(\yv) \prod_{n=1}^{m} q_n(\{\uv_{n}\samp{j}\}|y_{1:n-1})}{\prod_{n=1}^{m}\qhat_n(y_n | \{\uv_{n}\samp{j}\},y_{1:n-1} )}
- 1
\label{eq:NDMCImpWeightVar}
\end{equation}

We observe that since $\{ \qhat_n \} $ are assumed to be consistent
estimators, i.e.
$\qhat_n( y_n| \{ \Uv_n\samp{j} \}, y_{1:n-1} ) \Rightarrow q_n( y_n |
y_{1:n-1} ) $
in probability, as we expect, the variance of $w$ approaches that
resulting from using the \emph{exact} SCP. Unfortunately, in the case
of using an arbitrary stochastic proposal as opposed to one such as
MCMC enjoying asymptotic convergence properties, we cannot make a
stronger statement. In particular it is not strictly true that more
accurate estimates $\qhat$ will result in lower overall variance. This
is understandable: the proposal itself could be fundamentally poor
with respect to $\pi$, in the sense that sampling from and evaluating
it exactly would yield high variance. It is quite conceivable in those
situations that the approximations may in fact be favourable to the
exact values. This highlights an important aspect of our framework: it
allows principled use of an arbitrary proposal, but it cannot
compensate for an ill-fitting one. Design of an appropriate proposal
is the task of the domain practitioner.

The present importance sampling framework is an especially good fit
for Bayesian estimation in an additional sense to naturally solving
the theoretical issue of coverage of the support of $\pi$. Suppose we
possessed the capacity to \emph{store} the approximating distributions
$\qhat_n$, so that instead of building them from scratch by renewed
generation from the constrained proposal prior to drawing each sample,
we maintain them in memory and sample from them. Furthermore, with
some specified period, we can \emph{update} the Bayesian estimators as
discussed in Section \ref{sec:SCP:qnEstimators:Preliminaries} by
\emph{querying} the proposal for more data. As long as we keep
updating the estimates, this implies (assuming we have sufficiently
large memory) that the approximating importance sampling estimator
will approach that of the SCP using its exact $q(\yv)$. Of course,
``sufficiently large memory'' will in general be of exponential size;
Section \ref{sec:SSS} discusses our ultimate solution of performing
this type of sampling by maximally exploiting finite memory resources.

We conclude this section by noting for completeness that in many
cases the target $\pi$ is only known up to some intractable normalization
constant (or \emph{partition function}) $Z_\pi$. Specifically, if
\[
\pi(\xv) = \frac{\tilde{\pi}(\xv)}{Z_{\pi}}
\]
we can only evaluate $\tilde{\pi}$.
In such situations, one can adapt a well-known method to dealing with
this in conventional importance sampling. Define the
\emph{unnormalized} importance weights as
\begin{equation}
  \tilde{w}( \yv, \{ \uv_{1:m} \samp{j} \}) = \frac{\tilde{\pi}(\yv)}{\prod_{n=1}^{m}\qhat_n(y_n | \{\uv_{n}\samp{j}\},y_{1:n-1} )}
  \label{eq:NDMCImpWeightsUnnorm}
\end{equation}
The following \emph{biased} IS estimator can then be used in place of (\ref{eq:NDMCImpSamp}):
\begin{equation}
  \Expect{\pi}[h(\Yv)] \approx \frac{ \sum_i \tilde{w}(\Yv\samp{i}, \{ \Uv_{1:m}\samp{j} \}) h(\Yv\samp{i}) } { \sum_i \tilde{w}(\Yv\samp{i}, \{ \Uv_{1:m}\samp{j} \}) }
  \label{eq:NDMCImpSampUnnorm}
\end{equation}
Furthermore, an unbiased estimate for the partition
function is given by
\begin{equation}
  Z_{\pi} \approx \frac{1}{N} \sum_i \tilde{w}(\Yv\samp{i}, \{ \Uv_{1:m}\samp{j} \})
  \label{eq:NDMCISPartitionFuncEst}
\end{equation}

\subsection{Markov Chain Monte Carlo}
\label{sec:MCusingSCP:MCMC}

We now examine the feasibility of implementing a correlated dynamical
process whose equilibrium distribution is $\pi$ using the SCP. In
distinction to the importance sampling scenario, as long as the
resultant Markov chain is \emph{ergodic}, there is no requirement that
the proposal cover the full support of $\pi$; the process will
asymptotically visit the regions of the state space with the correct
probabilities. Generally, an MCMC proposal will depend on
the current state $\xv$ but may be a large-neighbourhood generalization
of local or single-site MCMC. For example one may build a
\emph{locally-modelling} distribution about the $\xv$. Alternatively,
one may only apply the proposal to some subset of the variables as it
may not adequately cover $\pi$ over the full state space.

This discussion is for reference as in the present work, we have not
attempted to implement a MCMC sampling using SCP. It is
nonetheless noteworthy, as in the Importance Sampling case, that usage
of \emph{approximate} distributions results in a correct methodology.

Algorithm \ref{alg:MCMCSCP} describes our MCMC routine. In contrast to
IS, at each step, a single joint sample is drawn from the proposal and
a random decision is made whether to accept the sample or remain at
the current point. To enforce \emph{detailed balance} (DB), which
ensures that the process converges to the correct distribution $\pi$,
the probabilities of making the \emph{forward} sampling move and its
\emph{reversal} must be computed. Neither of these two can in general
be evaluated exactly for sufficiently complex heuristic proposals; in
particular, note that in Algorithm \ref{alg:MCMCSCP}, only the
\emph{approximate} values of the proposal are used in computing the
acceptance probability \( \alpha \). Nonetheless, the method satisfies
DB, as shown in the following.

\begin{algorithm}
  \caption{Markov Chain Monte Carlo by SCP}
  \label{alg:MCMCSCP}
  \begin{algorithmic}
    \Statex
    \State Given state $\xv$, generate (correlated) updated state $\xv'$
    \Statex
    \State \textbf{Draw forward sample}
    \State Sample \( \{ \Uv_1\samp{j} \}  \sim q_1(\uv | \xv ) \)
    \State Evaluate \( \qhat_1(y_1 | \{ \uv_1\samp{j}\}, \xv  ) \)
    \State  Sample \( Y_1 \sim \qhat_1(y_1 | \{ \uv_1\samp{j}\}, \xv  ) \)
    \For{\( n = 2 \ldots m \)}
      \State Sample \( \{ \Uv_n\samp{j} \}  \sim q(u_{n:m} | y_{1:n-1},
      \xv ) \)
      \State Evaluate \( \qhat_n(y_n | \{ \uv_n\samp{j}\}, y_{1:n-1}, \xv  ) \)
      \State Sample \( Y_n \sim \qhat_n(y_n | \{ \uv_n\samp{i}\},
      y_{1:n-1}, \xv  ) \)
    \EndFor
    \Statex
    \State \textbf{Compute reverse proposal}

  \State Sample \( \{ \Zv_1\samp{j} \}  \sim q(\zv | \yv ) \)
  \State Evaluate \( \qhat_1(x_1 | \{ \zv_1\samp{j}\}, \xv  ) \)
  \For{\( n = 2 \ldots m \)}
    \State Sample \( \{ \Zv_n\samp{j} \}  \sim q_n(z_{n:m} | x_{1:n-1},
    \yv ) \)
    \State Evaluate \( \qhat_n(x_n | \{ \zv_n\samp{j}\}, x_{1:n-1},
    \yv  ) \)
  \EndFor
  \Statex
  \State \textbf{Acceptance step}
  \State  Evaluate \( \alpha = \min\Big(1, \frac{\pi(\yv)  \prod_{n=1}^{m}  \qhat_n(x_n | \{ \zv_n\samp{j}\}, x_{1:n-1},
    \yv  ) }
  { \pi(\xv) \prod_{n=1}^{m} \qhat_n(y_n | \{ \uv_n\samp{j}\}, y_{1:n-1},
    \xv  )     }  \Big )\)
  \State Sample $R \sim U[0,1]$
  \If{$R < \alpha$}
    \State  $\xv' \gets \yv$
  \Else
    \State  $\xv' \gets \xv$
  \EndIf
\end{algorithmic}
\end{algorithm}

\begin{proposition}
Algorithm \ref{alg:MCMCSCP} satisfies detailed balance. In other words:
\begin{equation}
\pi( \xv ) q( \{\uv \}, \{\zv \}, \yv |
\xv) \alpha(\xv; \{\uv \}, \{\zv \}, \yv
) = 
\pi( \yv ) q( \{\zv  \}, \{\uv \}, \xv |
\yv) \alpha(\yv; \{\zv \}, \{\uv \}, \xv
)
\label{eq:NDMCDetBal}
\end{equation}

\end{proposition}
\begin{proof}
  Suppose without loss of generality that the LHS \( \alpha < 1 \)
  so that $\alpha$ on the RHS, corresponding to the \emph{reverse}
  process of generating the sets of samples \( \{ \zv_n \} \),
  \( \xv \)
  using the approximate proposals computed using these samples, and
  finally the sets \( \{ \uv_{n} \} \),
  is therefore 1. Then the overall forward transition probability
  under invariant
  distribution \( \pi \), i.e. the LHS of (\ref{eq:NDMCDetBal}), is
  \[
  \pi(\xv)  \prod_{n=1}^{m} q_n(\{\uv_{n}\samp{j}\} | y_{1:n-1},
  \xv)\qhat_n(y_n|\{\uv_n\samp{j}\}, y_{1:n-1}, \xv ) \prod_{n=1}^{m} q_n(\{
  \zv_{n}\samp{j} \} | x_{1:n-1}, \yv ) \times \alpha 
  \]
  which by substituting the definition of \( \alpha \) yields
  \[
  \pi(\yv)\prod_{n=1}^m q_n(\{
  \zv_{n}\samp{j} \} | x_{1:n-1}, \yv ) \qhat_n(x_n |
  \{ \zv_n\samp{j} \}, x_{1:n-1}, \yv )  \prod_{n=1}^{m} q_n(\{\uv_{n}\samp{j}\} | y_{1:n-1},
  \xv)
  \]
  Under the assumption of \( \min( 1, \frac{1}{\alpha} ) = 1 \), we see
  that this is precisely the reverse transition probability from \(
  \yv \) to \( \xv \) via the \( \{\zv_n\} \) and \( \{\uv_n \} \)
  auxiliary variables, i.e. the RHS of (\ref{eq:NDMCDetBal}). Hence
  detailed balance is satisfied over all random variables.
\end{proof}

Before leaving this section, we remark that one's options for storing
and modifying the proposal distribution in the MCMC setting are more
limited than they are for importance sampling, where such strategies
do not invalidate the algorithm. This will be discussed more
thoroughly in future work examining the SCP for performing MCMC; for
now, we note that Algorithm \ref{alg:MCMCSCP} requires recomputing the
forward and reverse proposals for each move attempt.

\section{State Space Sampling: Putting it all together}
\label{sec:SSS}

  
  

\subsection{Algorithm Summary}
\label{sec:SSS:AlgSummary}

We have presented ideas for sequentially constructing approximations
to an unknown stochastic process with which we desire to sample from
discrete target distribution $\pi$. We have shown that both importance
sampling and MCMC are in principle feasible using the general
approach. Practical matters have so far been deferred or only
peripherally mentioned. In this section, we finally orchestrate the
disparate components into a workable importance sampling
algorithm. Generalization to MCMC will be considered in future work.

The method we propose is clearly motivated by the ideas of \emph{state
  space search} \cite{russell1995artificial,zhang1999state} from
computer science. In that classical idea, a \emph{heuristic function}
is used to guide exploration of a state space with the objective of
finding some \emph{goal state}, for example one of minimum cost. In
this work, we use our results on the statistical correctness of
importance sampling with the SCP to generalize state space search to
the stochastic setting, and we hence call the algorithm, somewhat
insipidly, \emph{state space sampling}(SSS). The approximate statistics
gathered from runs of the proposal process will serve as the
randomized heuristic decision rules. In direct analogy to the search
case, the quality of the resultant estimates depend on the suitability
of the proposal to the target distribution.

We will illustrate the key aspects in Section \ref{sec:SSS:comicBook}
using a running example; pseudocode describing a rudimentary version
is included in Section \ref{sec:SSS:algPseudoCode}. Here we quickly
summarize the method.

An outer loop seeks to generate, say, $K$ samples
$\{\Yv\samp{i} \in \Scube^m \}$for use in an importance sampling
estimator. Initially, no statistical model of the proposal
distribution exists, so we instantiate samples $\{ \Uv\samp{j}\}$ from
the process and dynamically determine as large a state space partition
as possible subject to a conservative constraint on the KL loss as
described in Section
\ref{sec:SCP:qnEstimators:DynPartConstruction}. In other words, we
grow a subcube tree and determine a Bayesian count-based or
Rao-Blackwellised estimator over its nodes' occurrence
probabilities. All leaf nodes then have nonzero probability of being
visited, including those for which no samples have been observed.

We next sample a leaf node from the partial tree; if this node
corresponds to a \emph{full} binary state, we return it along with its
probability\footnote{in practice, its logarithm  to avoid numerical underflow}
and proceed to draw the next state beginning from the root of the existing
tree. If, on the other hand, the leaf node corresponds to an
\emph{incomplete} state, we recursively call the constrained proposal
from the partial state corresponding to the visited leaf node, growing
a new subcube tree and statistical representation from the leaf
node. We call this a \emph{refresh} of a leaf node; all points in the
tree that result from a renewed call to the constrained process are
called \emph{refresh points.} Once the subcube tree has been extended,
the sampling resumes in this manner until a full state has been
generated, which is returned along with its probability.

Due to \emph{caching} in computer memory of the overall tree between
subsequent state requests, considerable savings in the number of
required proposal process runs can be achieved in some circumstances.
These runs constitute the dominant computational bottleneck for
problems with rough energy functions. As all leaf nodes occur with
nonzero probability however, the tree will eventually grow
unmanageably large as an increasing number of states are requested. To
circumvent this, the subcube tree is kept at \emph{a bounded size} by
\emph{prioritized
  deletion}. Once a predefined size limit is reached, the least likely
elements are removed to make room for new nodes. Such memory bounding
strategies have been used in state space search algorithms for some
time \cite{russell1992cient}. We call an internal node both of whose
children are leaves a \emph{subleaf node.} The operation of converting
a subleaf node into a leaf by deleting its two children is called
\emph{retraction}. By maintaining a retraction \emph{priority
  queue} of all subleaf nodes, where a node's retraction priority
\emph{decreases} with increasing visitation probability, we can
efficiently look up candidates for retraction as a new part of the
tree is grown. When the tree attains its maximum size, retraction is
interleaved with branching of new nodes, so some care is required to
prevent removing nodes that may interfere with the growth of a
subtree. The pseudocode in Algorithm \ref{alg:dynamicTrees}
illustrates the steps to ensure correct subtree construction.

\subsection{State Space Sampling: The Comic Book}
\label{sec:SSS:comicBook}

As promised, we now illustrate a simple run of the algorithm. We wish
to generate samples on an $m=3$ variable binary state space. All
constrained proposal runs will draw the same number of
samples ($N=10$), though as discussed in Section
\ref{sec:SCP:SeqConstProc}, this is not required. To demonstrate the
memory-bounding aspect, we impose a size limit on the subcube tree of
$9$ nodes.

\begin{figure}
  \centering
  \begin{subfigure}[b]{0.4\textwidth}
    \centering
    \includegraphics[scale=1.15]{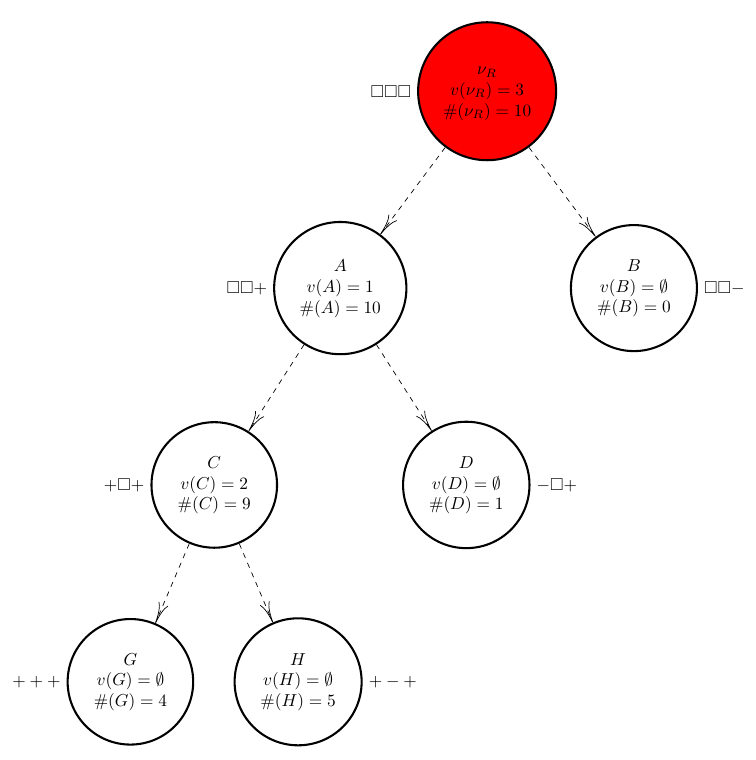}
    \caption{}
    \label{fig:stateSpaceSample1}
  \end{subfigure}
  \hspace{20mm}
  \begin{subfigure}[b]{0.4\textwidth}
    \centering
    \includegraphics[scale=1.15]{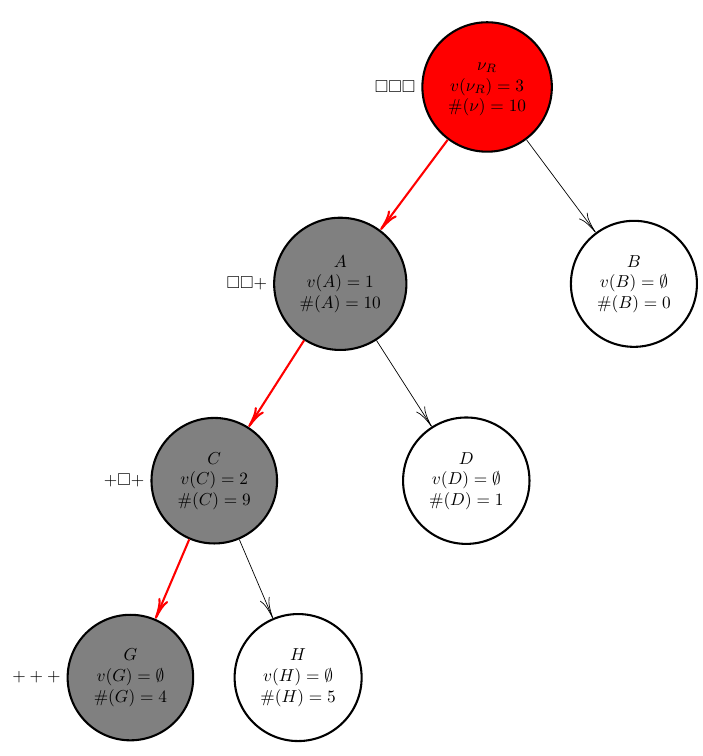}
    \caption{}
    \label{fig:stateSpaceSample2}
  \end{subfigure}
  \caption{First stages of the procedure on a toy binary system of
    $m=3$ variables. Figure \ref{fig:stateSpaceSample1} illustrates
    construction of the initial subcube tree from a population
    $\{\Uv\samp{j}\}$ drawn from an unconstrained proposal
    process. Labels, branch variables, and populations are shown in
    each tree node; corresponding partial binary states appear
    adjacently. The root node is red to signify that it is a refresh
    point; a maximum tree size of $9$ nodes must be ensured in the
    subsequent algorithm. In Figure \ref{fig:stateSpaceSample2}, the
    full binary state $+++$ happens to be sampled from the model.}
\end{figure}

Prior to drawing the first sample, we of course possess no statistical
representation, hence we draw a population $\{ \Uv\samp{j} \}$ from
the proposal with no constraints applied. A subcube tree is
constructed as shown in Figure \ref{fig:stateSpaceSample1}. Each tree
node displays its label (e.g. $\nu_R$ for the root,) its branching
variable $v$, and the population of samples in the data with the
corresponding partial state. Leaf nodes do not possess a branch
variable.  In Figure \ref{fig:stateSpaceSample1}, we indicate the
partial states beside each tree node; note that only nodes $G$ and $H$
represent fully-instantiated states. Furthermore, node $B$ is excluded
in our algorithm from further branching because it is empty; node $D$
was prevented from branching because it would have violated the KL
loss constraint (the tree would have become unreliable.) The root node
is coloured red to signify that it was a refresh point. We construct
the Bayes estimator over the leaf nodes of the subcube tree as
described in Section \ref{sec:SCP:qnEstimators:DynPartConstruction},
where in this case the partition element of minimum count is node $B$,
and recursively set all internal node probabilities. Node $C$ is a
subleaf node, and so it is entered into the retraction queue.

We then proceed to sample from the representation in Figure
\ref{fig:stateSpaceSample1}; the red arrows and sequence of gray nodes in
Figure \ref{fig:stateSpaceSample2} show the sampling of a
leaf node that happens in this case to represent the full state
$+++$. This state and its probability under the model can then be used
in an importance sampler. 

\begin{figure}
  \centering
  \hspace{-10mm}
  \begin{subfigure}[b]{0.4\textwidth}
    \centering
    \includegraphics[scale=1.15]{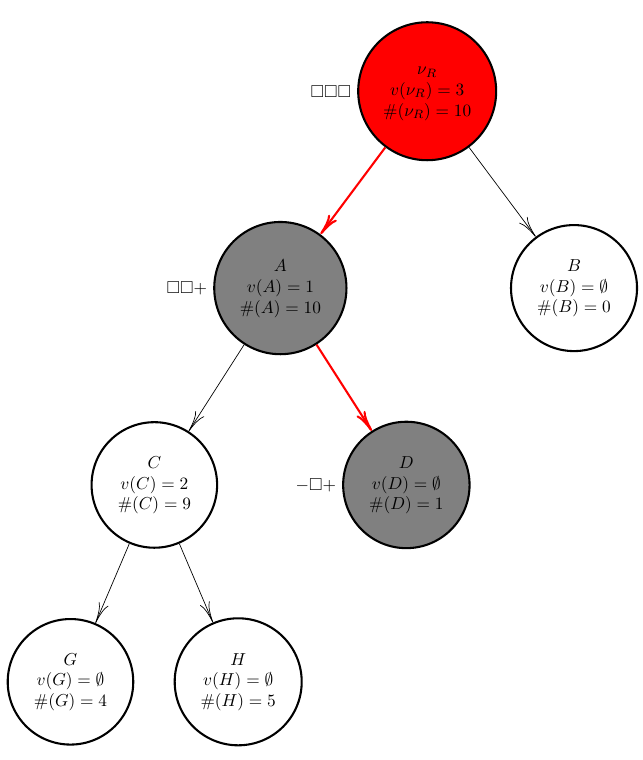}
    \caption{}
    \label{fig:stateSpaceSample3}
  \end{subfigure}
  \hspace{10mm}
  \begin{subfigure}[b]{0.4\textwidth}
    \centering
    \includegraphics[scale=1.15]{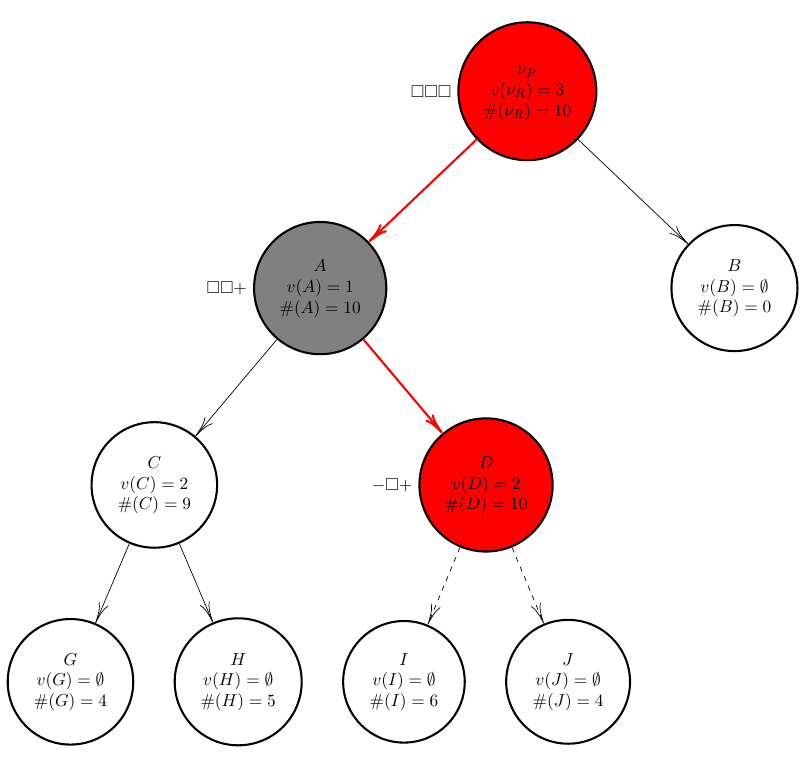}
    \caption{}
    \label{fig:stateSpaceSample4}
  \end{subfigure}
  \caption{Figure \ref{fig:stateSpaceSample3} shows a sample from the
    subcube tree encountering leaf node $D$, paired with incomplete
    binary state $- \Box +$. Hence, sampling pauses as we \emph{query}
    for more data by running the proposal with constraining condition
    $- \Box +$. The subcube tree is \emph{extended} in Figure
    \ref{fig:stateSpaceSample4}, with node $D$ becoming a refresh
    point. The tree is now at its maximum allowable size $9$.}
\end{figure}

Suppose that, as shown in Figure \ref{fig:stateSpaceSample3}, in a
later request for a sample we happen to visit node $D$, mapping to the
incomplete state $-\Box +$. The sampling pauses as \emph{extension} of
the tree is required.  In Figure \ref{fig:stateSpaceSample4}, node $D$
becomes a refresh point; $10$ runs of the proposal are performed with
clamping condition $-\Box +$.  Nodes $I$ and $J$ are generated to
extend the subtree. As a result of the extension, $D$ becomes a
subleaf and is thus placed in the retraction queue.
\begin{figure}[h!]
  \centering
  \includegraphics[scale=1.15]{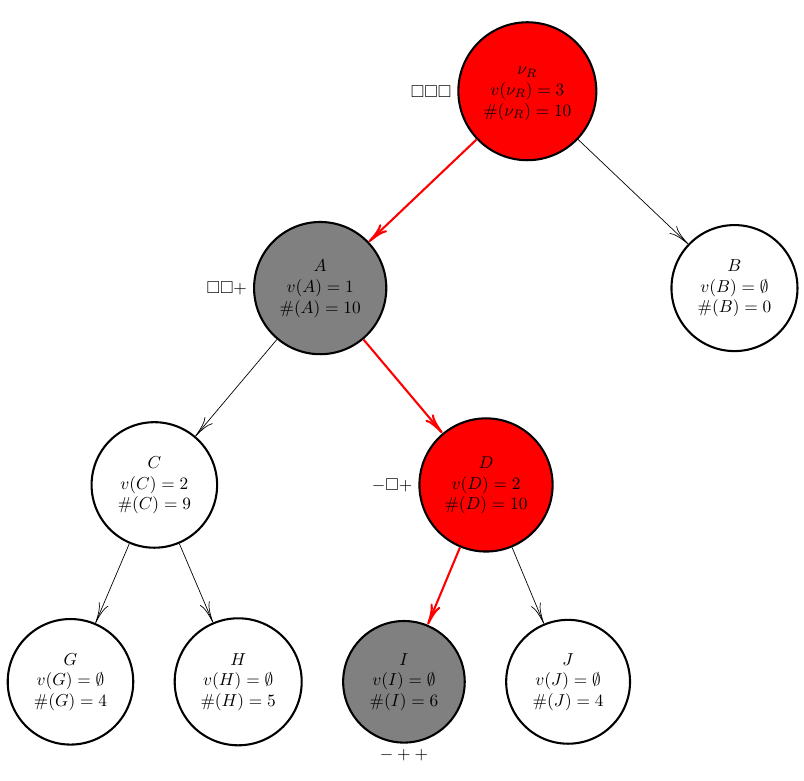}
  \caption{Continuation of the sampling paused in Figure
    \ref{fig:stateSpaceSample3}; in this case, full binary state $-++$
    results.}
  \label{fig:stateSpaceSample5}
\end{figure}
Sampling resumes in Figure \ref{fig:stateSpaceSample5}; the full state
$-++$ results. Its returned probability is simply the product of the
conditional probabilities along the path from the root, irrespective
of which model was used to approximate the conditionals. Note that the
overall number of generated nodes is now $9$, so the tree is at its
full capacity. Any further node generations will now require
corresponding retractions.

\begin{figure}[h!]
  \centering
  \hspace{-30mm}
  \begin{subfigure}[b]{0.4\textwidth}
    \centering
    \includegraphics[scale=1.15]{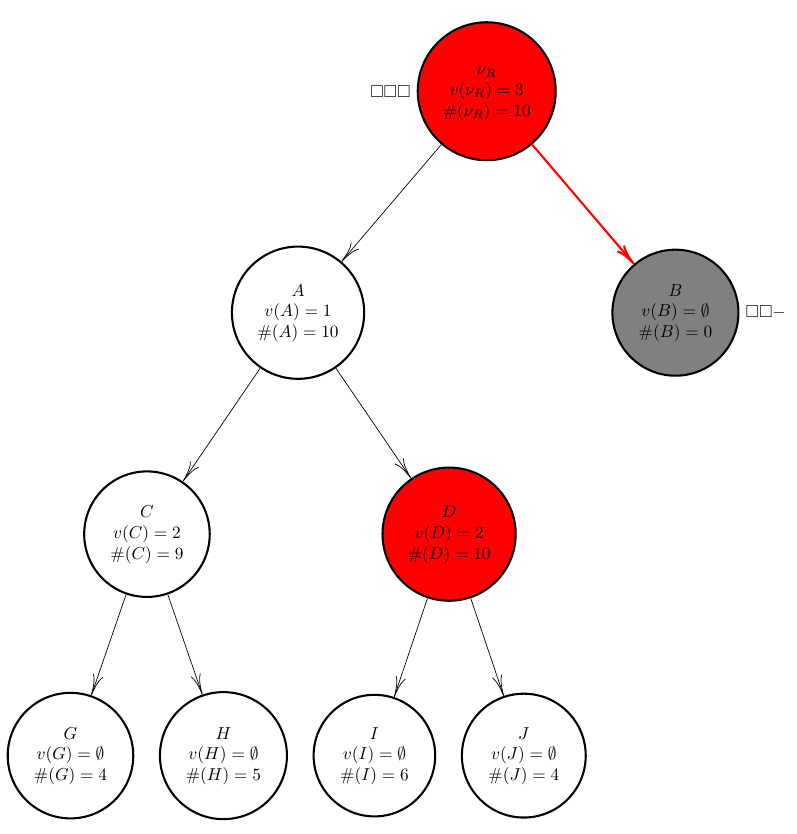}
    \caption{}
    \label{fig:stateSpaceSample6}
  \end{subfigure}
  \hspace{20mm}
  \begin{subfigure}[b]{0.4\textwidth}
    \centering
    \includegraphics[scale=1.15]{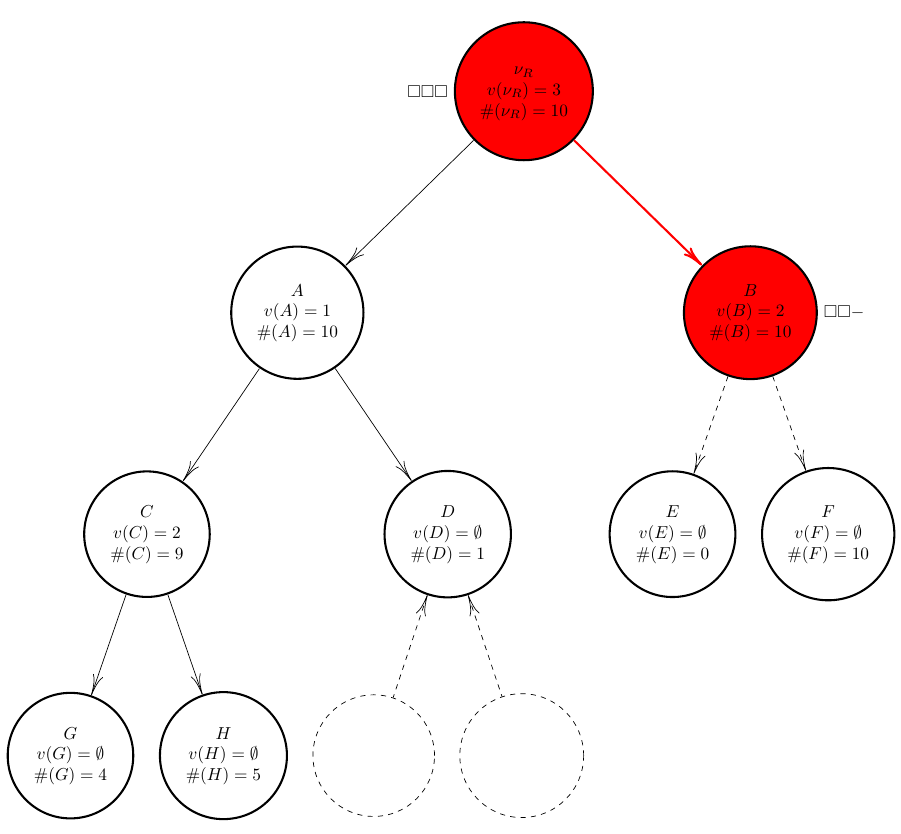}
    \caption{}
    \label{fig:stateSpaceSample7}
  \end{subfigure}
  \caption{Incomplete state $\Box\Box -$ is sampled in
    \ref{fig:stateSpaceSample6}; a refresh is performed from node $B$,
    but extension of the tree now requires retraction to respect the
    9 node size restriction. Of the two eligible \emph{subleaf} nodes,
    i.e. those with two children, both of which are leaves, node $D$
    happens to have a higher priority for retraction because it has a
    lower visitation probability. Nodes $I$ and $J$ are thus deleted
    to make room for $E$ and $F$. }
\end{figure}

In Figure \ref{fig:stateSpaceSample6}, we happen to visit node $B$,
with partial state $\Box\Box -$. Once again, sampling pauses at node
$B$ while it is refreshed with a call to the proposal with condition
$\Box \Box -$ imposed. In Figure \ref{fig:stateSpaceSample7} the first
growth step using the new population generates nodes $E$ and $F$ by
branching on variable $v=2$; to enforce the tree size constraint, a
subleaf node is retracted. In this case, it happens that node $D$ has
a lower visitation probability than $C$, so its two children $I$ and
$J$ are deleted, maintaining the tree size at $9$. Node $B$ is a
subleaf now, but is not yet considered a retraction candidate to avoid
interfering with the subtree under construction.

\begin{figure}[h!]
  \centering
  \hspace{-20mm}
  \begin{subfigure}[b]{0.4\textwidth}
    \centering
    \includegraphics[scale=1.15]{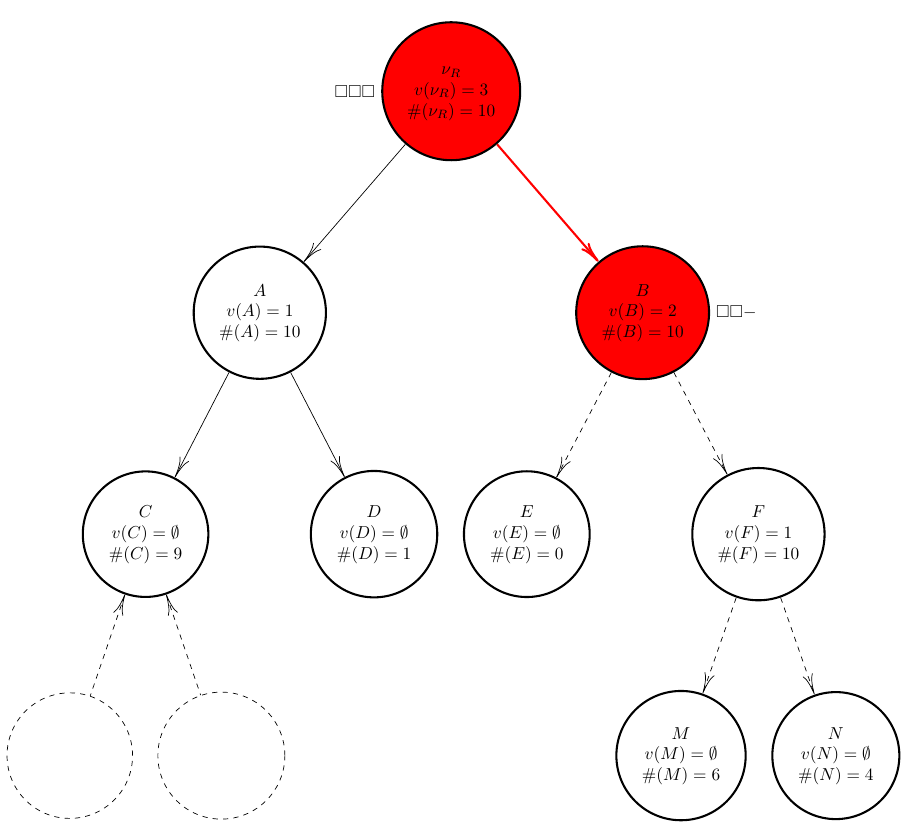}
    \caption{}
    \label{fig:stateSpaceSample8}
  \end{subfigure}
  \hspace{30mm}
  \begin{subfigure}[b]{0.4\textwidth}
    \centering
    \includegraphics[scale=1.15]{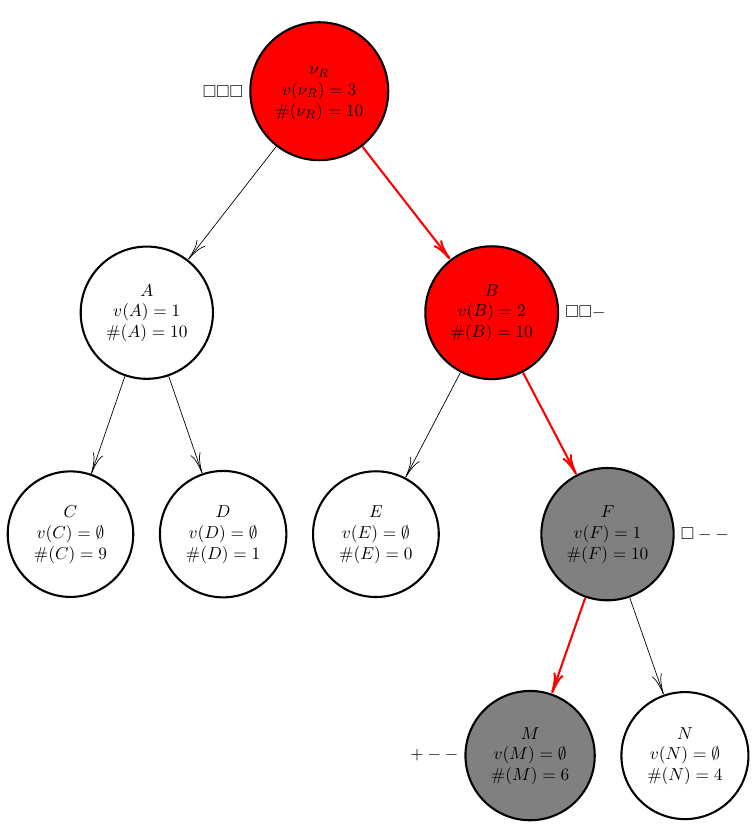}
    \caption{}
    \label{fig:stateSpaceSample9}
  \end{subfigure}
  \caption{Extension of the tree begun in
    \ref{fig:stateSpaceSample7} continues, still using the proposal
    samples generated from node $B$. Subleaf $C$ is retracted to
    free memory. At this point the extension is complete, and the
    sampling paused in Figure \ref{fig:stateSpaceSample6} resumes to
    generate state $+--$}
\end{figure}

The tree is continued from $F$ to yield nodes $M$ and $N$ in Figure
\ref{fig:stateSpaceSample8}. The sole remaining subleaf, node $C$, is
retracted, rendering $A$ a subleaf. But the refresh from $B$ is now
complete ($E$ is empty.) All subleaves of the new subtree are now
considered eligible for retraction; in this case, $F$ is enqueued. The
sampling which paused at $B$ in Figure \ref{fig:stateSpaceSample6} now
resumes; Figure \ref{fig:stateSpaceSample9} shows the continuation to
node $M$, i.e. full state $-++$.

Having illustrated the essential ideas, we next present pseudocode of
a variant of SSS that was used in our experimental validation.

\subsection{Algorithm Pseudocode}
\label{sec:SSS:algPseudoCode} 

The collection of functions comprising SSS is presented in
Algorithm \ref{alg:dynamicTrees}. The pseudocode must necessarily
trade off transparency of presentation with efficiency of
implementation; we certainly do not claim that the routines are
optimal. Not every variable and function in the pseudocode is
explicitly defined; we hope that, in conjunction with the example run
discussed in Section \ref{sec:SSS:comicBook} and the discussion of the
algorithm principles throughout the paper, what the methods are doing
will be clear from the context. We ask the reader to mind the
distinction between a node $\nu$ and a branch variable $v$. Following
a standard convention, elements of a data structure are referred to by
their name preceded by a period; for example the left child of a tree
node $\nu$ is $\nu$.\texttt{LeftChild}.

The code presented computes the estimated state probabilities using
the count-based robust Bayes estimator described in Section
\ref{sec:SCP:qnEstimators:RobustBayes}; extension to
Rao-Blackwellisation as discussed in Section
\ref{sec:SCP:qnEstimators:RaoBlack} is not difficult.

The main function is \algoname{SampleS}, which recursively returns discrete states and their
(log) probabilities. It takes as input a tree
node $\nu$, a structure called \texttt{status} containing dynamical
data related to the generation process, and a structure called
\texttt{params} containing the user-specified simulation
parameters. The \texttt{status} structure is used to track the state
under construction and its log probability.

Once a call to \algoname{SampleS} arrives at a leaf node corresponding
to an incomplete partial state, the function \algoname{extendTree}
performs the required tree extension. This involves running the SCP
with the user-specified heuristic process
\texttt{params}.\texttt{simulationProcess} and constructing a robust
Bayes estimator on a state space partition derived from the SCP
samples $\{ \uv\samp{i} \}$. An \emph{expansion queue} prioritizes the
order in which to branch the tree nodes; variables are selected from
the unassigned set $\nu.\texttt{freeVariables}$ according to some rule
called \algoname{chooseBranchVariable}. Examples of such a branching
rule are to follow a fixed or totally random order, or for problems
with a topological structure, to employ a parallelisation strategy as
discussed in Appendix \ref{sec:Appendix:partitioning}. For Ising-type
problems, an additional possibility which we used in our
implementation is to allow a randomized order, but only from among
unassigned variables neighbouring those in the assigned set.

Extension of the tree terminates when the posterior KL loss exceeds
threshold $\theta$; nodes are only considered candidates for branching
if the population of the corresponding partial states is larger than
$\#_{thresh}$. The function \algoname{retractWorstSubleaf} bounds the
tree size at specified value \texttt{maxTreeSize} by prioritized
deletion from a \emph{retraction
  queue} of subleaf nodes.

Before presenting numerical results, we mention that the possibility
alluded to in Section \ref{sec:MCusingSCP:ImportanceSampling} of
periodic Bayesian estimator updating following drawing more SCP
samples from the partial state distributions corresponding to subtrees
was left to future work.  The idea is briefly touched upon in Appendix
\ref{sec:Appendix:BayesianUpdating}.

\begin{algorithm}
  \caption{State Space Sampling}
  \label{alg:dynamicTrees}
  \begin{algorithmic}
    
    \Function{sampleS}{$\nu$,\texttt{status},\texttt{params}}
      
      \If{\texttt{status}.State is a full state}
        \State \texttt{status}.$\textrm{logQ} \gets \nu.\textrm{logQ}$
        \State \textbf{return}
      \EndIf
      
      \If {$\nu$ is a leaf node}
        \State \textbf{extendTree}($\nu$, \texttt{status},\texttt{params})
      \EndIf
      
      \State $\nu_{next} \gets $ \textbf{sampleChild}($\nu$)
      
      \If {$\nu.\textrm{\texttt{LeftChild}} = \nu_{next}$}
        \State $S_{next} \gets 1$
      \Else
        \State $S_{next} \gets -1$
      \EndIf
      
      \State \texttt{status}.State[$\nu.v$] $\gets S_{next}$
      \State \textbf{sampleS}($\nu_{next}$, \texttt{status}, \texttt{params})
      
    \EndFunction
    \Statex           

    \Function{sampleChild}{$\nu$}
      \State $\triangleright$ Only valid for internal nodes
      \State $P_+ \gets \exp( \nu.\textrm{\texttt{LeftChild}.logQ} - \nu.\textrm{logQ}  )$
      \State Sample $R \sim U[0,1]$
      \If{$R< P_+$}
        \State \textbf{return} $\nu$.\texttt{LeftChild}
      \Else
        \State \textbf{return} $\nu$.\texttt{RightChild}
      \EndIf

    \EndFunction
    
    \Statex

    \Function{setTreePs}{$\nu,N,\textrm{logQ}_r$}

      \If{$\nu$ is a leaf}
        \State $\nu.\textrm{logQ} \gets \log[\#(\nu)+\alpha(\nu)] - \log(1 + N) + \textrm{logQ}_r$
        \State \textbf{return}
      \Else
        \State \textbf{setTreePs}($\nu.\textrm{\texttt{LeftChild}},N,\textrm{logQ}_r$)
        \State \textbf{setTreePs}($\nu.\textrm{\texttt{RightChild}},N,\textrm{logQ}_r$)
        \State $\nu.\textrm{logQ} \gets \log[ \exp(\nu.\textrm{\texttt{LeftChild}.logQ}) + \exp( \nu.\textrm{\texttt{RightChild}.logQ}) ]$
          
      \EndIf
    \EndFunction

    \algstore{alg:dynamicTrees}
    
  \end{algorithmic}
\end{algorithm}

\begin{algorithm}
  \caption{State Space Sampling (cont'd)}
  \begin{algorithmic}
    \algrestore{alg:dynamicTrees}

    \Function{extendTree}{$\nu$,\texttt{status},\texttt{params}}
      \State $\triangleright$ $\nu$ is root node of subtree being grown

      \If {$\nu$ is not root}  \Comment{Avoid root Node}
        \If {$\nu.\textrm{\texttt{Parent}}$ is a subleaf node}
          \State Remove $\nu.\textrm{\texttt{Parent}}$ from \texttt{status}.$\Rcal$
        \EndIf
      \EndIf
      
      \State $\{\uv\samp{i}\} \gets $\texttt{params}.\texttt{simulationProcess}(
      \texttt{status}, \texttt{params} )  \Comment{Sample from process}

      \State $\nu.n \gets |\{\uv\samp{i}\}| $ \Comment{Number of samples may be determined by the heuristic}

      \State $\Fcal \gets$ \textbf{newExpansionQueue}
      \State \textbf{enQueue}($\nu, \Fcal$)

      \While{$\Fcal$ is not empty} \Comment{There exist frontier nodes}

        \State \texttt{retractionFailed} $\gets$ \textbf{retractWorstSubleaf}(\texttt{status}, \texttt{params}, $\nu$.\texttt{Parent})
        \If{\texttt{retractionFailed}}
          \State \textrm{break}
        \EndIf
      
        \State $\nu' \gets$ \textbf{deQueueBest}($\Fcal$) \Comment
        {Most populated tree node is selected for branching}     
        \State $v' \gets$ \textbf{chooseBranchVariable}($\nu', \{
        \uv\samp{i}\}$) \Comment {E.g. randomly from unassigned variables}
        \State \textbf{branchTree}($\nu'$, $v'$, $\{ \uv\samp{i}\}$)
        \Comment {Node $\nu'$ is extended by branching on $v'$}
        \State $\alphav' \gets$ \textbf{robustBayesEstimator}($\nu$)
        \Comment{New potential $\alphav$ matching extended tree}
        \State $KL \gets$ \textbf{evalKLCost}( $\nu$, $\alphav'$)
        \If{ $KL > \textrm{\texttt{params}.}\theta$} 
          \Comment{Reject the branch on $v'$}
          \State \textbf{retractTree}($\nu'$)
          \State \textbf{break}
        \EndIf        
        \State $\alphav \gets \alphav'$ \Comment{Accept the
          new estimator at the leaf nodes}
        \If{$|\nu'.\textrm{\texttt{freeVariables}}| > 1$} \Comment{Did we assign last
            possible variable?}
          \State $\nu_+ \gets \nu'.\textrm{\texttt{LeftChild}}$
          \State $\nu_- \gets \nu'.\textrm{\texttt{RightChild}}$
          \If{$\#(\nu_+) > \textrm{\texttt{params}.}\#_{thresh}$} \Comment{Only if $\nu_+$ is
              sufficiently populated}
            \State \textbf{enQueue}($\nu_0$, $\Fcal$)
          \EndIf
          \If{$\#(\nu_-) > \textrm{\texttt{params}.}\#_{thresh}$} \Comment{Only if $\nu_-$ is
              sufficiently populated}
            \State \textbf{enQueue}($\nu_1$, $\Fcal$)
          \EndIf
        \EndIf
  
      \EndWhile
      \algstore{alg:dynamicTrees}
    \end{algorithmic}
  \end{algorithm}
  
\begin{algorithm}
  \caption{State Space Sampling (cont'd)}
  \begin{algorithmic}
    \algrestore{alg:dynamicTrees}
      \State \textbf{setTreePs}($\nu,\nu.n,\nu.\textrm{logQ}$) \Comment{Recursively set internal probabilities of new subtree}
      \State \textbf{updateRetractionQueue}($\nu$, \texttt{status}) \Comment{Enter subleaves of new subtree into retraction queue}
    \EndFunction
    \Statex
    \Function{robustBayesEstimator}{$\nu$}
      \State $\#_{min} \gets \infty$
      \For{$\nu' \in T(\nu)$} \Comment{All nodes in subtree rooted at $\nu$}
        \If{$\nu'$ is a leaf node}
          \If{$\#(\nu') < \#_{min}$}
            \State{$\#_{min} \gets \#(\nu')$}
          \EndIf
        \EndIf
      \EndFor
      \For{$\nu' \in T(\nu)$}
        \If{$\nu'$ is a leaf node}
          \If{$\#(\nu') = \#_{min}$}
            \State{$\alpha(\nu') \gets 1/\#_{min}$}
          \Else
            \State{$\alpha(\nu') \gets 0$}
          \EndIf
        \EndIf
      \EndFor
    \EndFunction
    \algstore{alg:dynamicTrees}
  \end{algorithmic}
\end{algorithm}

\begin{algorithm}
  \caption{State Space Sampling (cont'd)}
  \begin{algorithmic}
    \algrestore{alg:dynamicTrees}
    \Function{updateRetractionQueue}{$\nu$,\texttt{status}}
      \If{$\nu$ is a leaf}
        \State \textbf{return}
      \ElsIf{$\nu$ is a subleaf}
        \State \textbf{enQueue}($\nu$, \texttt{status}.$\Rcal$)
        \State \textbf{return}
      \Else
        \State \textbf{updateRetractionQueue}($\nu$.\texttt{LeftChild}, \texttt{status})
        \State \textbf{updateRetractionQueue}($\nu$.\texttt{RightChild}, \texttt{status})
      \EndIf
    \EndFunction
    
    \Statex

    \Function{retractWorstSubleaf}{\texttt{status}, \texttt{params}, $\nu_{protected}$}
      \State $\triangleright$If tree size at limit, try to delete
      worst subleaf, avoiding $\nu_{protected}$
      \State $\triangleright$If no subleaves are available to
      retract, returns failure
      \If{\texttt{status}.treeSize $>$ \texttt{params}.\texttt{maxTreeSize}}
        \If{\texttt{status}.$\Rcal$ is not empty} \Comment{Available nodes
            to retract?}
          \State $\nu_{worst} \gets $\textbf{deQueueWorst}(\texttt{status}.$\Rcal$)
          \State \textbf{retractTree}$(\nu_{worst})$ \Comment{Delete
            the two children of $\nu'$}
          \State \texttt{status}.treeSize $\gets$ \texttt{status}.treeSize$-2$
          \If{$\nu_{worst}$.\texttt{Parent} is now a subleaf and is not $\nu_{protected}$}
            \State \textbf{enQueue}($\nu_{worst}$.\texttt{Parent}, \texttt{status}.$\Rcal$)
          \EndIf
          \State \textbf{return} Success 
        \Else
          \Comment{No available retractions}
          \State \textbf{return} Failure
        \EndIf
      \Else
        \Comment{Tree not at full capacity, no retraction needed}
        \State \textbf{return} Success
      \EndIf
    \EndFunction
  \end{algorithmic}
\end{algorithm}

\section{Experimental Results}
\label{sec:Experiments}

\subsection{Overview}
\label{sec:Experiments:Overview}

We now turn to an experimental validation of the SSS methodology. We
analyze the algorithm on four different classes of Ising model
by running a simulated annealing (SA) process under varying sets of
heuristic parameters. The analytically-solvable cases of
\emph{independent} and \emph{chain-structured} (one-dimensional) Ising
variables are considered, as are the more complex
\emph{fully-connected} (or \emph{Sherrington-Kirkpatrick}) and
\emph{three dimensional} models.

Given that one of the stated points of this work is to allow for
equilibrium sampling using processes that do not necessarily converge
to the target distribution, it may seem unsatisfying that we have
chosen to use SA as the heuristic in demonstrating SSS. We defend the
choice in this preliminary work by noting that the SSS algorithm is
\emph{agnostic} to the particular heuristic chosen, but that using an
equilibrating process like SA is an informative \emph{diagnostic} of
how well SSS is performing. Recall that when using SA, each SCP run
will converge to the conditional Boltzmann probability of the free
variables given the constraining set. Consequently, it should be
possible to show the degree of convergence of the final SSS
distribution to the Boltzmann target. We make the cautionary reminder
however, that strictly speaking, this does \emph{not} show convergence
of a conventional (unconstrained) SA run, from which statistics are
gathered only to build the initial portion of the state space tree.
The overall SSS distribution generally depends strongly on the
simulation parameters of each call in the sequence used to construct
the tree. In particular, if SA is run with the \emph{same} simulation
length (number of sweeps) for each clamping condition, as was done
here for simplicity, the final distribution will tend to be
\emph{closer} to the Boltzmann distribution than the one arising from
a plain SA run. The reason is quite simple and was touched upon in
Section \ref{sec:SCP:SeqConstProc}; as the simulated subproblems
decrease in size with increased variable clamping, local equilibration
on the subproblems becomes more rapid. Nonetheless, convergence to the
target $\pi$ when using SA within SSS \emph{does} (asymptotically)
occur, justifying SA as a natural preliminary benchmark case for our
algorithm. It is reasonable to wonder if, for example, the sources of
statistical error arising from constructing models to the heuristic
distribution are large enough to severely increase the importance
sampling variance. Hence, by demonstrating correct performance under
SA, one can be more confident that when using another heuristic,
failure of the algorithm is a signature of an inappropriate heuristic,
not of errors inherent to the SSS method. When we mention
``convergence of SA'' in the following results, we intend it as a
shorthand for ``convergence of the SSS distribution under SA'' with
the preceding caveat in mind.

An exploration of the relative performance of the vast set of possible
heuristics is a substantial undertaking and will be dealt with in
future work.

As the Boltzmann distribution is sought by the SA Markov chains, it is
natural to compare when possible, $\log \qhat(\yv)$, the log
probability of each state $\yv$ returned by SSS, with that of
$\log \pi(\yv) = -E(\yv) - \log Z_\pi$. Of course, $\log Z_\pi$ is not
always available, but in such cases one can heuristically assess
Markov chain convergence by using an estimate of $Z_\pi$. The
\emph{scatterplot} of the samples' energies against their log
probabilities yields visually useful information on the extent to
which the target distribution has been reached. Ideally, $\log \qhat$
will be a strict linear function of $E$, with slope $-1$ (or $-\beta$
if the energy is scaled independently) and an intercept of
$-\log Z_\pi$. We call this relation the \emph{Boltzmann line}. The
fidelity of the SSS samples' $E$ versus $\log \qhat$ scatterplot to
this line determines the accuracy of an importance sampling estimator
based $\qhat$. For example excessive variance about the relation
implies that samples with very similar energies relative to the
temperature, which ought to occur with nearly identical probability,
will contribute quite differently to the estimator compensate for
their relative difference in actual probability. In essence, this
reduces the \emph{effective} number of samples and hence the
estimator's reliability.

\subsection{Methodological Details}
\label{sec:Experiments:Methodological}

We implemented SSS following the algorithm description in Section
\ref{sec:SSS:algPseudoCode}. Both count-based and Rao-Blackwellised
estimation were allowed for. For all simulations, the SCP used
$N=2000$ samples from the heuristic to construct the dynamic state
space tree; the posterior KL loss was bounded at a threshold value of
$\theta=0.05$ nats. No Bayesian updating was performed; the tree size
was bounded at a maximum of $10^5$ nodes by prioritized retraction of
low-probability nodes. The branching rule employed for tree
construction was to select a random unassigned variable from among
those that neighbour assigned variables (or a totally random variable
if none have been assigned.)

The \emph{independent} (\instancename{IsingIndep}),
\emph{one-dimensional} (\instancename{Ising1D}), and
\emph{full-connected} (\instancename{IsingSK}) models all had $m=100$
variables; the \emph{three-dimensional} (\instancename{Ising3D})
problem was of size $m=6\times 6\times 6 = 216$ nodes with periodic
boundary conditions. The systems are all of relatively modest size,
but this was necessitated by the lengthy simulation times when given a
serial implementation such as ours. An implementation exploiting the
trivially parallel nature of the population simulation, as well as the
parallelism inherent to the topology of problems such as
\instancename{Ising3D} as described in Section
\ref{sec:Appendix:partitioning}, will yield enormous speedups. The
results we present nonetheless show the considerable power of the
method: using a relatively simple idea, we can resolve the
probabilities of states generated by a generic stochastic process with
fairly high precision, despite these states occurring with extremely
small probability.

Problem \instancename{IsingIndep} was generated by sampling local
fields $h$ from a unit variance, zero mean Gaussian:
$h_i \sim N(0,1) $. All other problems had zero local field. For
\instancename{Ising1D} and \instancename{Ising3D},
$J_{ij} \sim N(0,1) $ for edges $\{i,j\}$ in the model, while with
\instancename{IsingSK}, $J_{ij} \sim N(0,1/\sqrt m) $ for all
$\{i,j\}$. The SA simulation parameters were chosen in a
problem-dependent manner as described in the following section.

\begin{figure}[H]
  \centering
  \begin{subfigure}[b]{0.5\textwidth}
    \includegraphics[width=\textwidth]{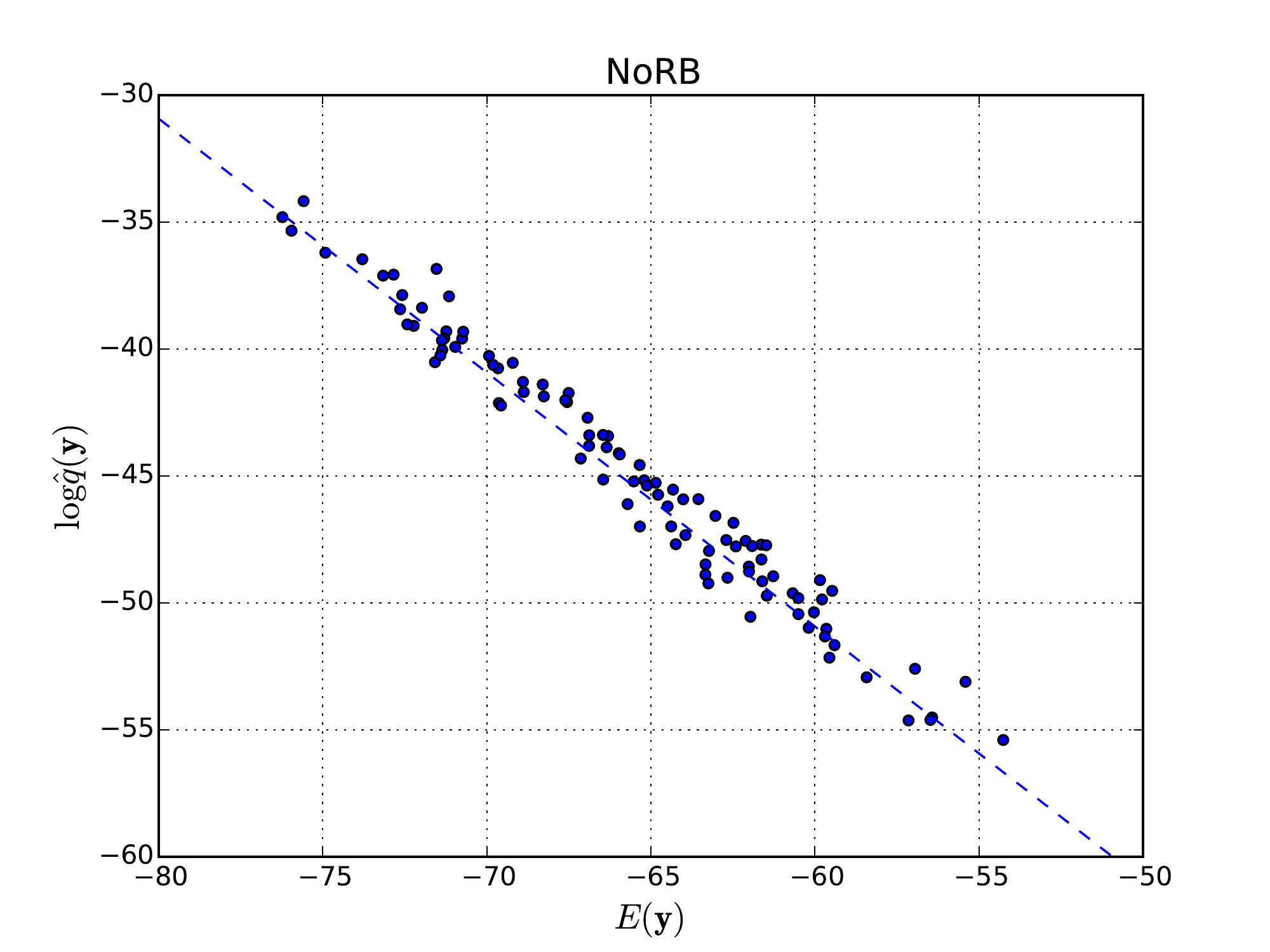}
    \caption{Count-based estimator}
    \label{fig:zeroJIsing_ElogQ_NoRB}
  \end{subfigure}
  \hspace{-5mm}
  \begin{subfigure}[b]{0.5\textwidth}
    \includegraphics[width=\textwidth]{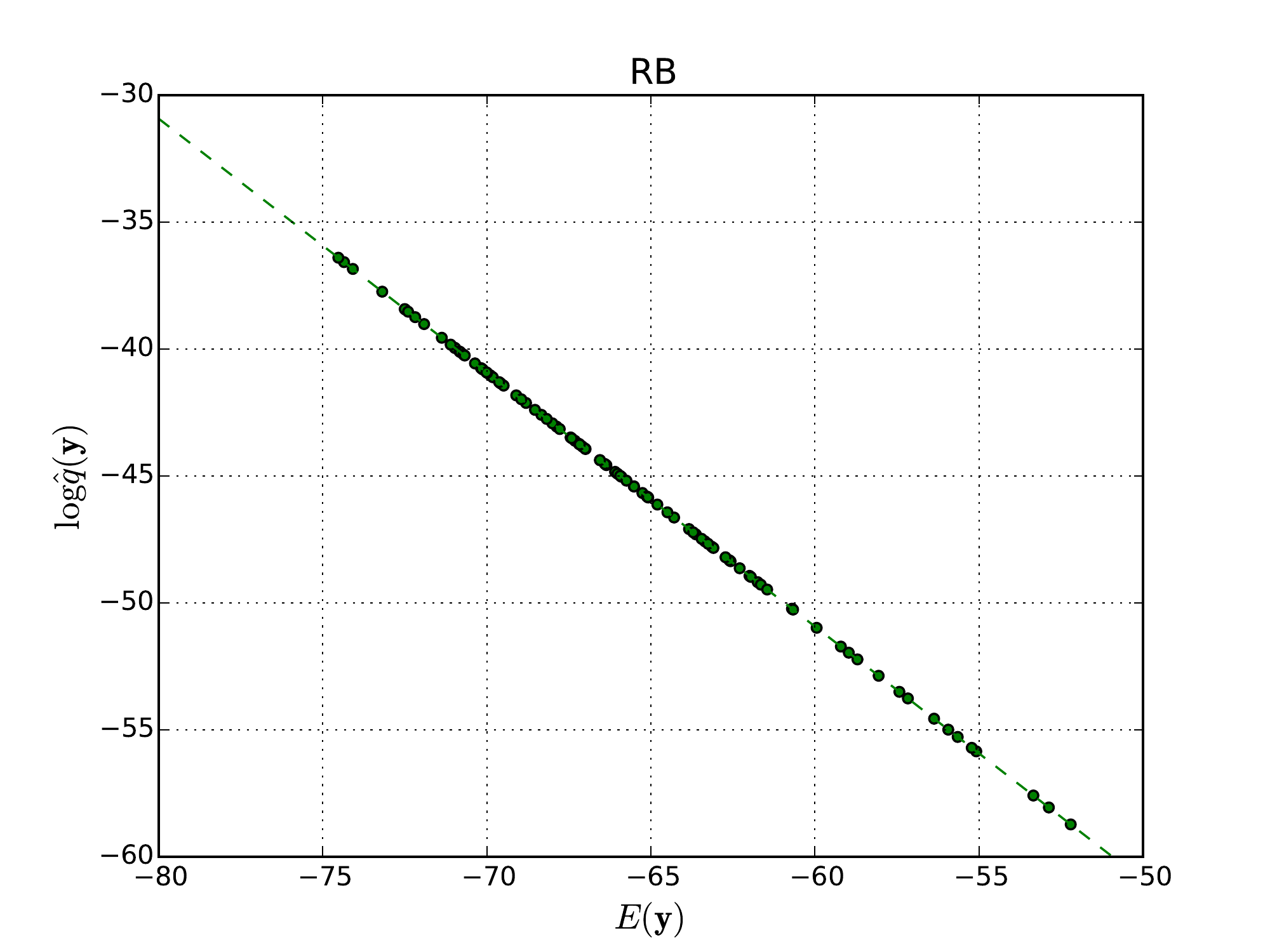}
    \caption{Rao-Blackwellised Estimator}
    \label{fig:zeroJIsing_ElogQ_RB}
  \end{subfigure}
  \caption{Scatterplot of configuration energies against log
    probabilities for $N_{draws}=100$ samples generated using SSS as
    described in the text on the independent-spin Ising model
    (\instancename{IsingIndep}). The ``heuristic'' in this case was an exact sampler
    from the target distribution.  The dotted line, termed the
    \emph{Boltzmann line} in this text, shows where the points would
    ideally lie. Figure \ref{fig:zeroJIsing_ElogQ_NoRB} shows the
    results using a count-based robust Bayesian estimator, while
    \ref{fig:zeroJIsing_ElogQ_RB} uses Rao-Blackwellisation, which for
    this case, completely eliminates the sampling noise as
    expected. However even for the results using the count-based
    estimator, the importance weight variance with respect to the
    exact true distribution is a fairly modest value of
    $\approx 0.927$. This is noteworthy considering the minuscule
    probabilities of the events themselves, which were evaluated using
    a relatively tiny number of samples from the distribution. Direct
    resolution of these events, whose probabilities are down to
    $\sim e^{-55}$ for this ensemble, is infeasible.}
  \label{fig:zeroJIsing_ElogQ}
\end{figure}

\subsection{Results}
\label{sec:Experiments:Results}

We first consider the \instancename{IsingIndep} case. This is a very simple problem
from which sampling exactly merely requires simulating $m$ independent
Bernoulli variables. Nonetheless, SSS does not make use of this
simplifying structure, and hence comparison with the known, exact
values of the samples' probabilities serves as a good initial
illustration of the reliability of the algorithm with a reasonable set
of parameters, as well as of the strength of Rao-Blackwellisation. The
inverse temperature for this system was taken as $\beta=1$, and exact
samples were generated to construct the state space tree. The results,
shown in Figure \ref{fig:zeroJIsing_ElogQ}, show the scatterplot of
the energies of $N_{draws} = 100$ samples queried from the state space
tree using count-based and Rao-Blackwellised estimators; the dotted
line with slope of $-\beta$ and intercept of exactly-calculated
$\log Z_\pi$ in the two subfigures show where exact samples taken from
the Boltzmann distribution should lie. As expected,
Rao-Blackwellisation eliminates statistical variance, but it should be
noted that despite the sampling noise, \emph{merely using the
  configuration statistics} can yield quite reliable estimates of the
log probabilities. The states appearing in Figure
\ref{fig:zeroJIsing_ElogQ_NoRB} occur with probabilities on the order
of $e^{-50}$; directly approximating their values with, for example a
simple histogram is simply out of the question. The importance weight
variance\footnote{not the variance of the log of the weights} for the
count-based estimator is $\approx 0.927$, quite a modest value
considering the vanishingly small probabilities being dealt with. This
variance can be reduced by increasing the population size $N$ and
decreasing the posterior KL threshold $\theta$ used in constructing
the state space tree. In general, the Rao-Blackwellised estimator will
not eliminate the variance as it does in this simple case, but as it
is not costly to compute and usually yields substantial
variance reduction, it is used in all subsequent experiments.

\begin{figure}[H]
  \centering
  \begin{subfigure}[b]{0.5\textwidth}
    \includegraphics[width=\textwidth]{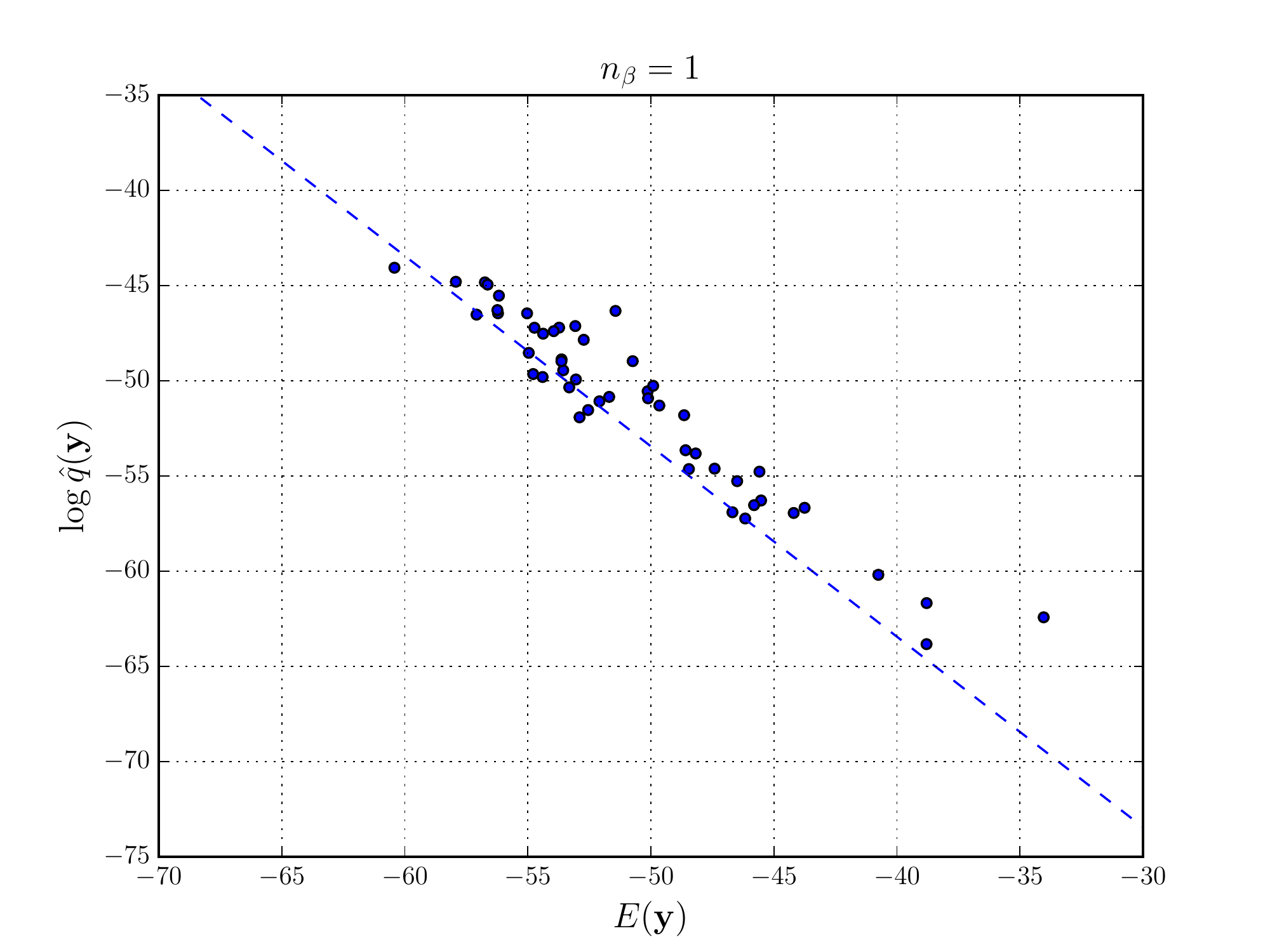}
    \caption{$n_{\beta}=1$}
    \label{fig:chainIsing_ElogQ_nBetas_1}
  \end{subfigure}
  \hspace{-5mm}
  \begin{subfigure}[b]{0.5\textwidth}
    \includegraphics[width=\textwidth]{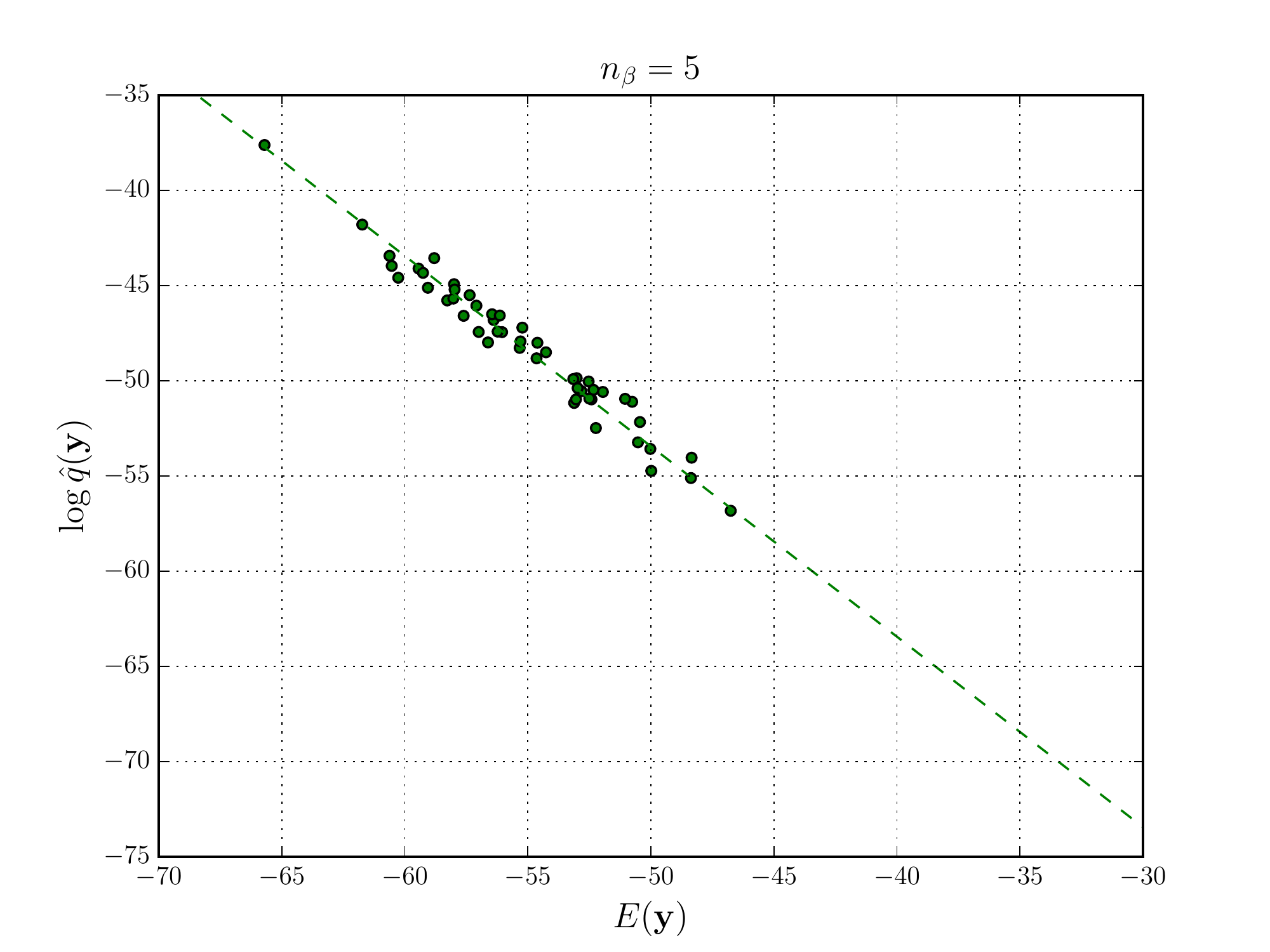}
    \caption{$n_{\beta}=5$}
    \label{fig:chainIsing_ElogQ_nBetas_5}
  \end{subfigure}
  \vspace{1mm}
    \begin{subfigure}[b]{0.5\textwidth}
    \includegraphics[width=\textwidth]{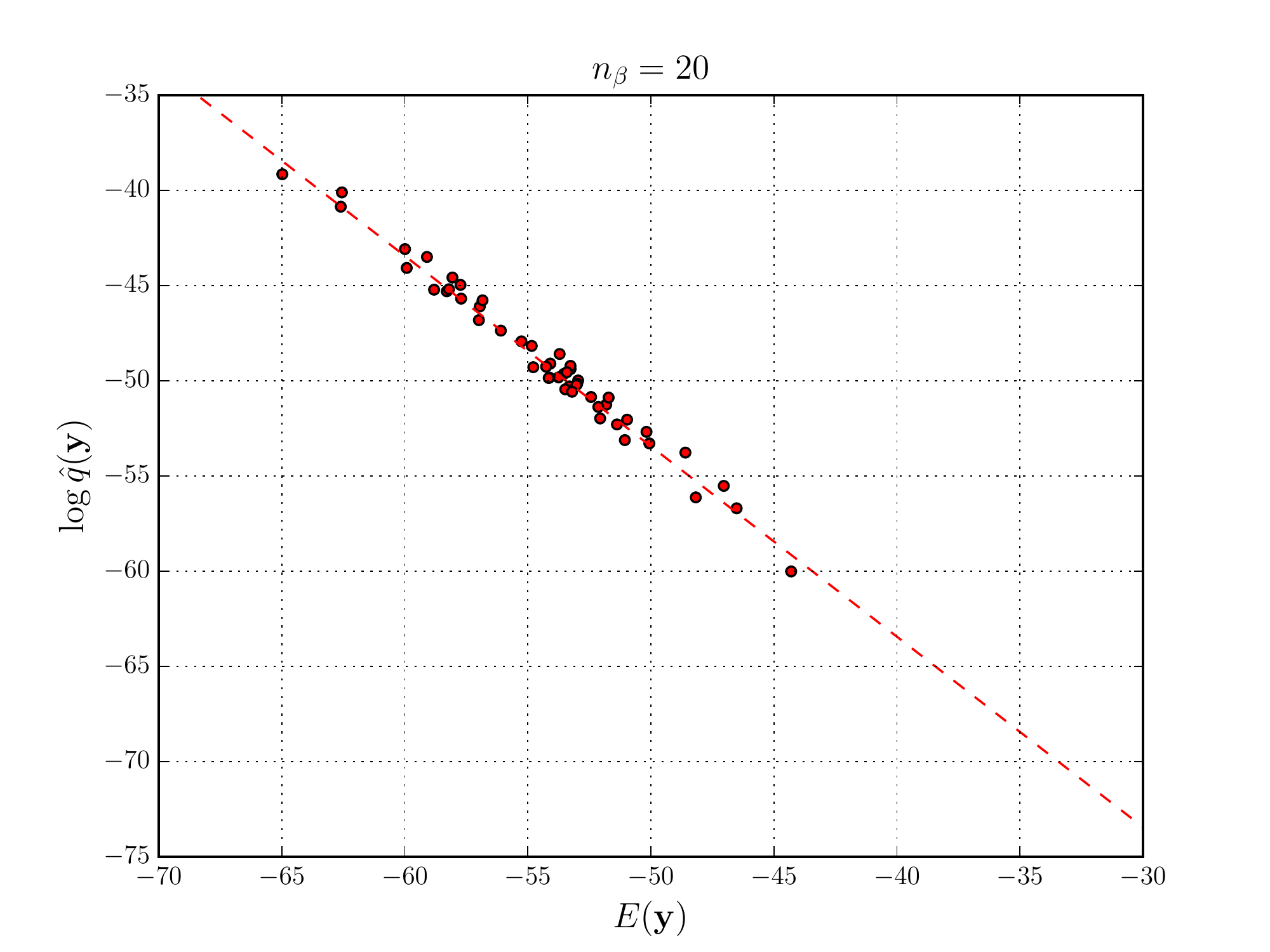}
    \caption{$n_{\beta}=20$}
    \label{fig:chainIsing_ElogQ_nBetas_20}
  \end{subfigure}
  \hspace{-5mm}
  \begin{subfigure}[b]{0.5\textwidth}
    \includegraphics[width=\textwidth]{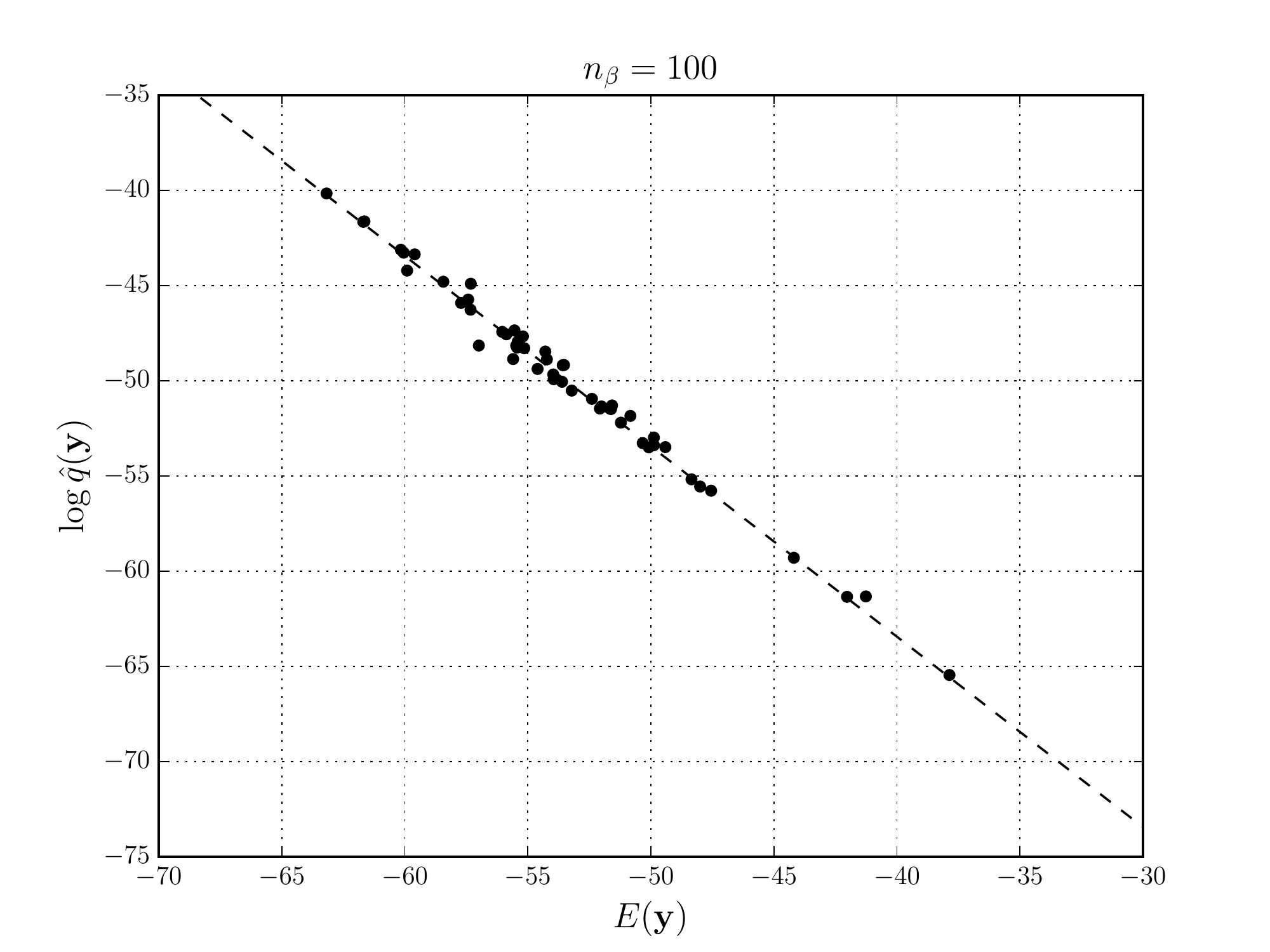}
    \caption{$n_{\beta}=100$}
    \label{fig:chainIsing_ElogQ_nBetas_100}
  \end{subfigure}
  \caption{$E$ versus $\log\hat{q}$ scatterplot of $N_{draws}=50$
    states generated by SSS using simulated annealing (SA) on the one
    dimensional Gaussian $J_{ij}$ Ising model (\instancename{Ising1D}). The number of
    SA sweeps used to anneal $\beta$ from $0.1$ to $1$ is stated below
    the corresponding plot. As expected, as the number of sweeps
    increases, the samples increasingly lie on the Boltzmann line,
    indicating that $\hat{q}$ is converging to $\pi$. See text for
    more details.}
  \label{fig:chainIsing_ElogQ}
\end{figure}

The next set of experiments examines \instancename{Ising1D}. The
target inverse temperature was set at $\beta=1$. Due to its low
treewidth, this model is still tractable in the sense that both exact
sampling from its equilibrium distribution and partition function
computation are straightforward. However to show some interesting
behaviour, we instead run SA on this model for a set of varying
simulation lengths. Interestingly, due to the continuity of $J_{ij}$
leading to so-called ``kinks'', there exist many local minima for this
problem\cite{li1981structure}. The SSS methodology allows us to
monitor the nonequilibrium statistics of the SA samples. We performed
SSS using a set of $n_\beta \in \{ 1, 5, 20, 100\}$ simulated
annealing sweeps with a linearly-increasing inverse temperature
schedule beginning with $\beta=0.1$ (the case of $n_\beta=1$ was
simply a single simulation sweep at the final temperature). SSS was
called $N_{draws} = 50$ times per simulation length. The results for
these four SA regimes are shown in Figure \ref{fig:chainIsing_ElogQ}.
As was the case for \instancename{IsingIndep}, the exact partition
function was computed via dynamic programming (or the \emph{transfer
  matrix} method), and the Boltzmann line is shown in each plot. We
see clearly in Figure \ref{fig:chainIsing_ElogQ_nBetas_1} that a
single Monte Carlo sweep at the target temperature is enough to bring
the population quite close to the Boltzmann distribution, but not
quite enough to reach it. The population statistics are quite varied
at this point, as we would expect from such a short simulation run.
Five and $20$ sweeps, as shown in \ref{fig:chainIsing_ElogQ_nBetas_5},
move the set further towards the Boltzmann line, while after
$n_\beta=100$ sweeps (Figure \ref{fig:chainIsing_ElogQ_nBetas_100})
many of the samples start to lie right on the line.

\begin{figure}[H]
  \centering
  \begin{subfigure}[b]{0.5\textwidth}
    \includegraphics[width=\textwidth]{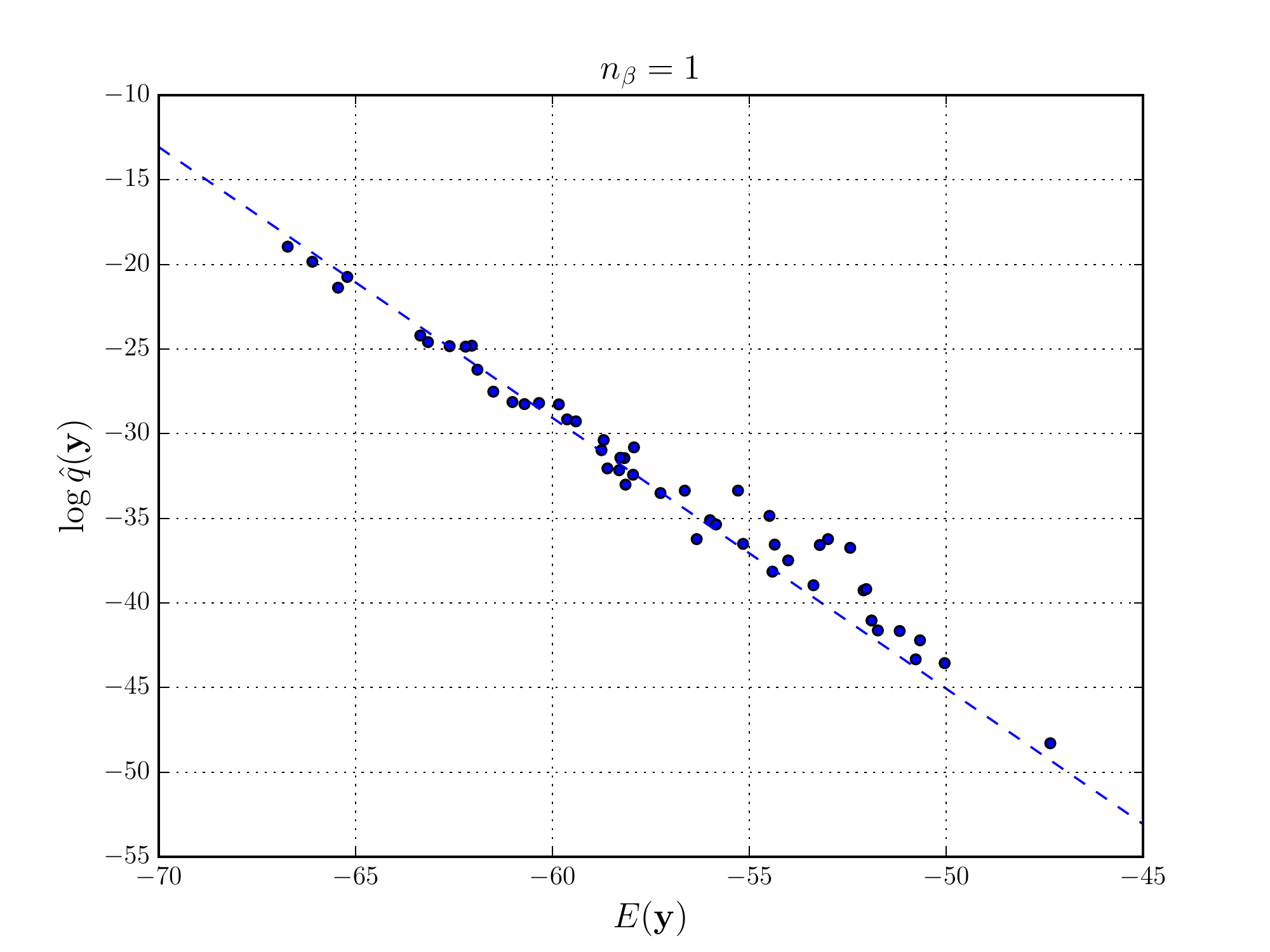}
    \caption{$n_{\beta}=1$}
    \label{fig:SKIsing_ElogQ_nBetas_1}
  \end{subfigure}
  \hspace{-5mm}
  \begin{subfigure}[b]{0.5\textwidth}
    \includegraphics[width=\textwidth]{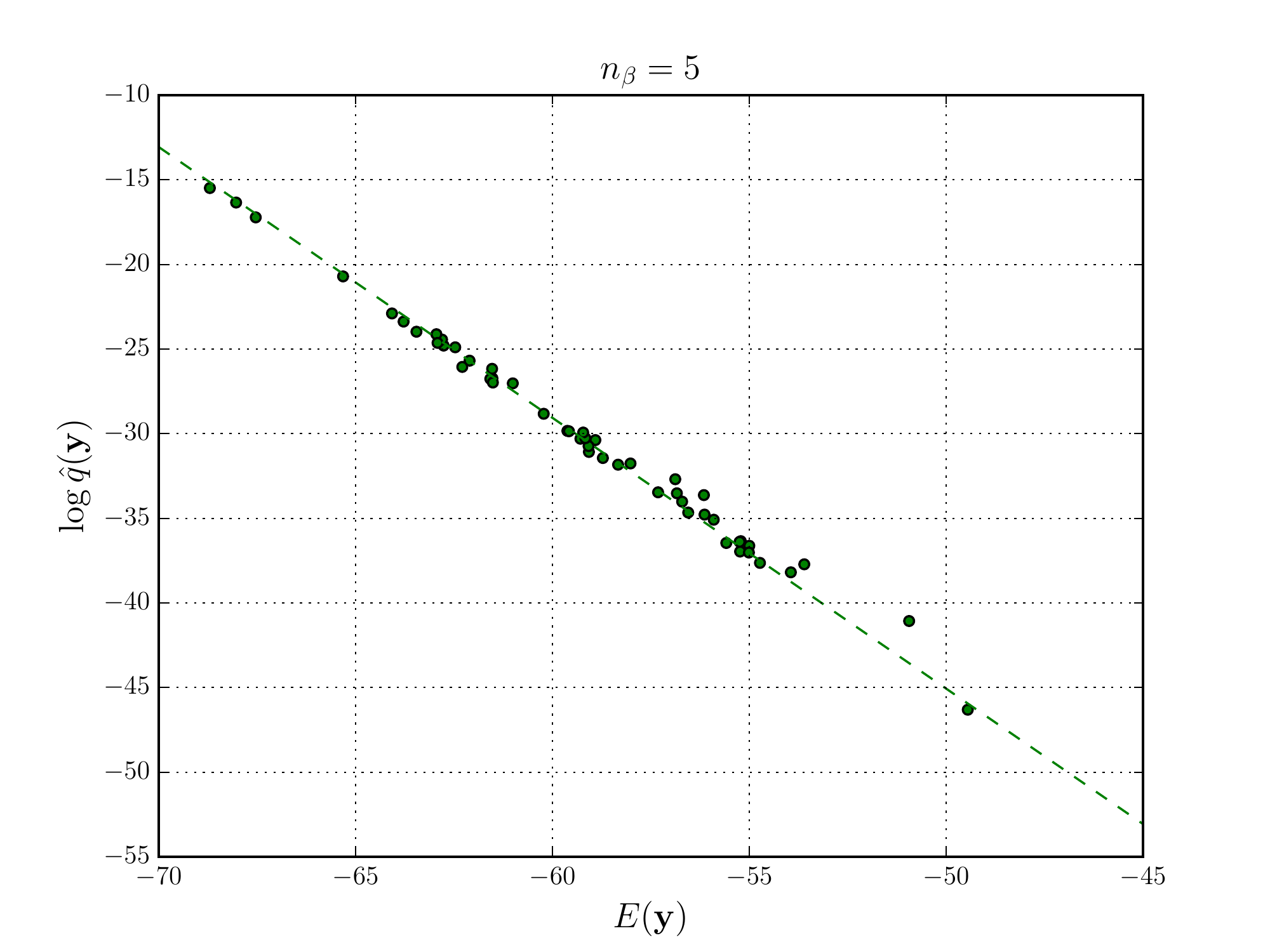}
    \caption{$n_{\beta}=5$}
    \label{fig:SKIsing_ElogQ_nBetas_5}
  \end{subfigure}
  \vspace{1mm}
    \begin{subfigure}[b]{0.5\textwidth}
    \includegraphics[width=\textwidth]{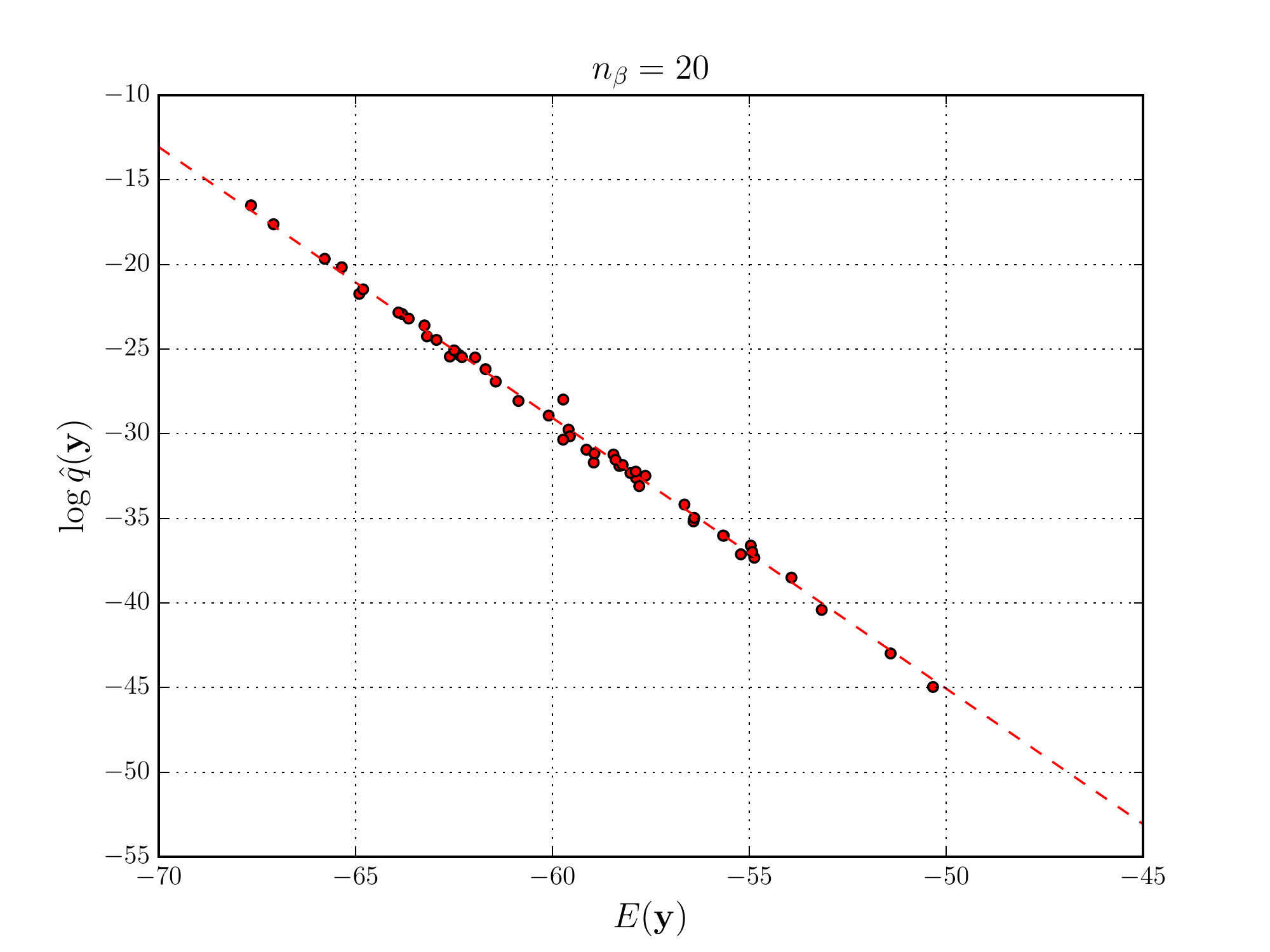}
    \caption{$n_{\beta}=20$}
    \label{fig:SKIsing_ElogQ_nBetas_20}
  \end{subfigure}
  \hspace{-5mm}
  \begin{subfigure}[b]{0.5\textwidth}
    \includegraphics[width=\textwidth]{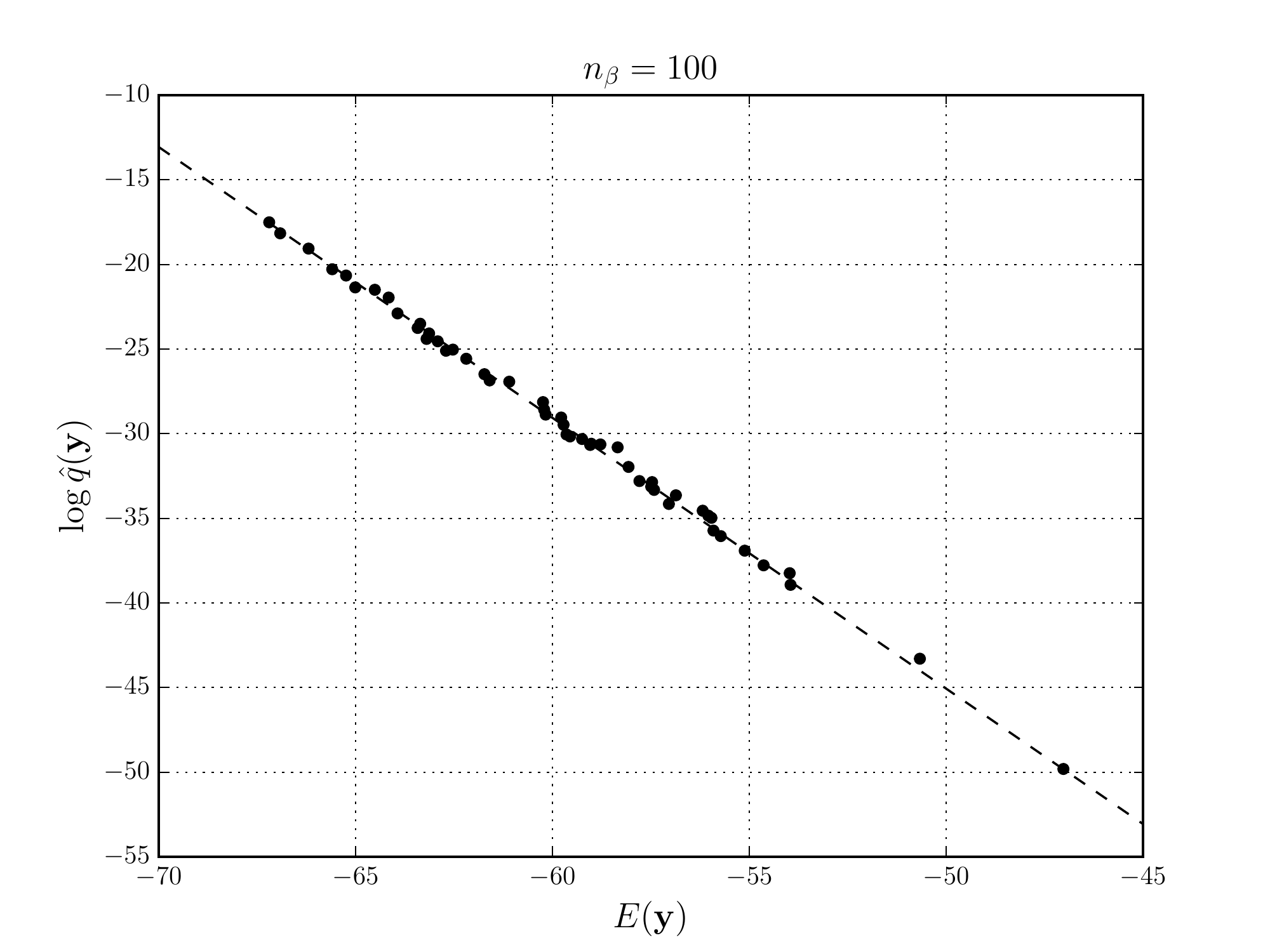}
    \caption{$n_{\beta}=100$}
    \label{fig:SKIsing_ElogQ_nBetas_100}
  \end{subfigure}
  \caption{$E$ versus $\log\hat{q}$ scatterplot of $N_{draws}=50$
    states generated by SSS using simulated annealing (SA) on the
    \emph{fully-connected} Gaussian $J_{ij}$ Ising model
    (\instancename{IsingSK}). For this instance, SA was run for each of a set of
    $n_\beta$ batches of $6$ sweeps, during which $\beta$ was linearly
    increased from $0.1$ to $1.6$; $n_\beta$ is indicated below each
    graph.  We see the same expected pattern of convergence of
    $\hat{q}$ to the distribution with increasing $n_\beta$. See text
    for details.}
  \label{fig:SKIsing_ElogQ}
\end{figure}

We proceed next to two examples for which it is no longer possible to
compare against analytical results. The first of these is the
$100$-variable fully-connected problem \instancename{IsingSK}; $\beta$
was set to $1.6$, corresponding to a temperature below the spin glass
transition temperature \cite{binder1986spin}. As for the case of
\instancename{Ising1D}, a set of increasing SA lengths were compared
against each other. We considered linearly increasing beta from an
initial value of $\beta=0.1$ to $\beta=1.6$ over the course of
$n_\beta \in \{ 1, 5, 20, 100\}$ steps, but as this problem is more
computationally-demanding than the previous two considered, each step
consisted of $6$ sweeps. Figure \ref{fig:SKIsing_ElogQ} shows the
outcomes. Plotting the Boltzmann line required an estimate of
$\log Z_\pi$, which was obtained with the importance sampling
estimator (\ref{eq:NDMCISPartitionFuncEst}) using the SSS samples
under the $n_\beta=100$ regime; the target value of $\beta=1.6$
defined the slope. We stress that the line was \emph{not} obtained by
linear regression, lessening the chance that a seemingly good relation
is an overfitting artifact. The plots show the same pattern of
convergence of the samples to $\pi$. Following a single set of sweeps
at the target $\beta$, the system has relaxed to the Boltzmann line,
but substantial variance still exists. The fit improves as the number
of SA steps increases, until $n_\beta=100$ appears to yield an
essentially converged final distribution. The reader is once again
asked to note that SSS is returning individual configurations with
quite accurate probability estimates across a range of many orders of
magnitude (between $\sim e^{-15}$ and $\sim e^{-50}$ for this case.)

\begin{figure}[H]
  \centering
  \begin{subfigure}[b]{0.5\textwidth}
    \includegraphics[width=\textwidth]{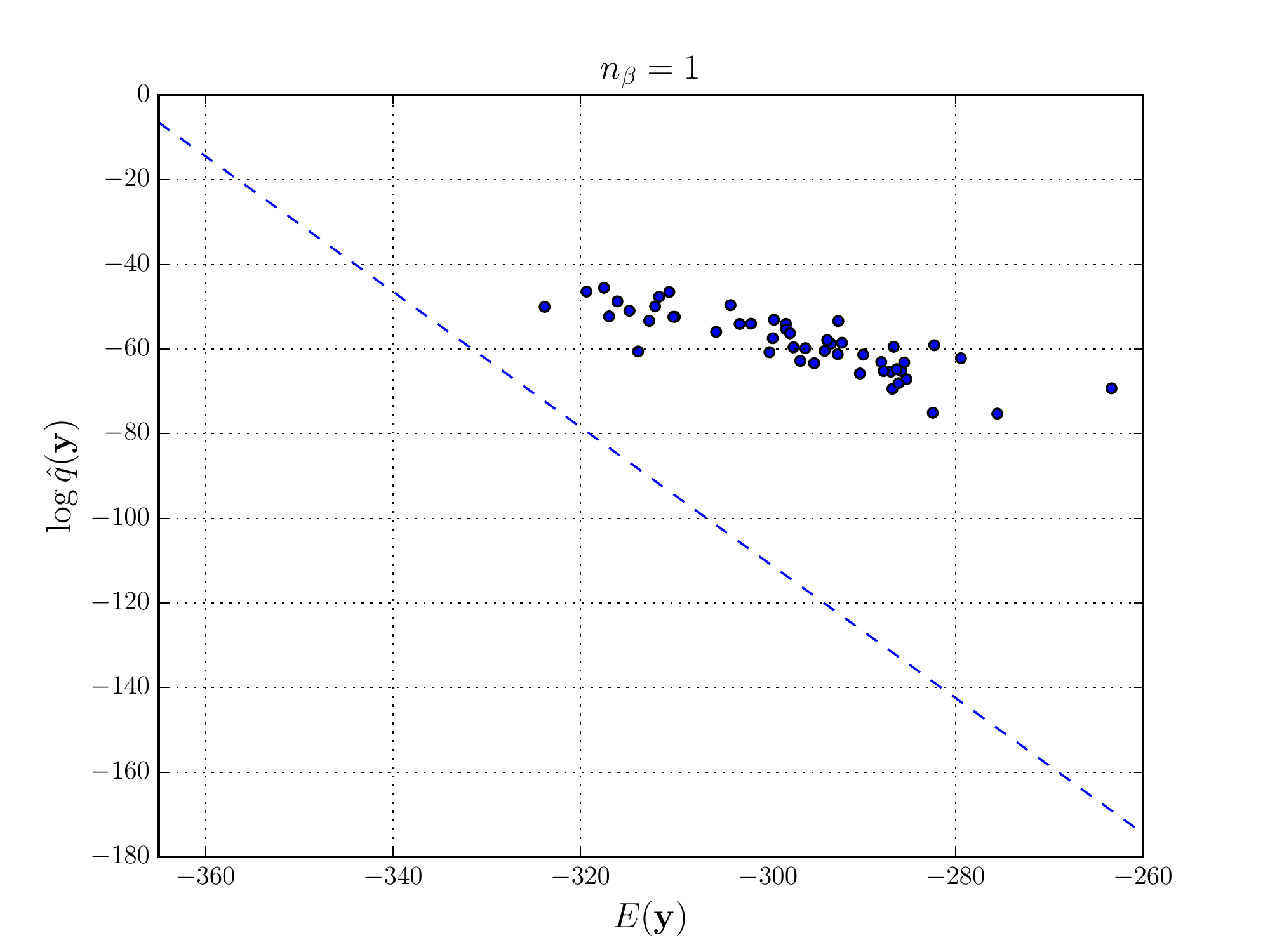}
    \caption{$n_{\beta}=1$}
    \label{fig:EA3DIsing_ElogQ_nBetas_1}
  \end{subfigure}
  \hspace{-5mm}
  \begin{subfigure}[b]{0.5\textwidth}
    \includegraphics[width=\textwidth]{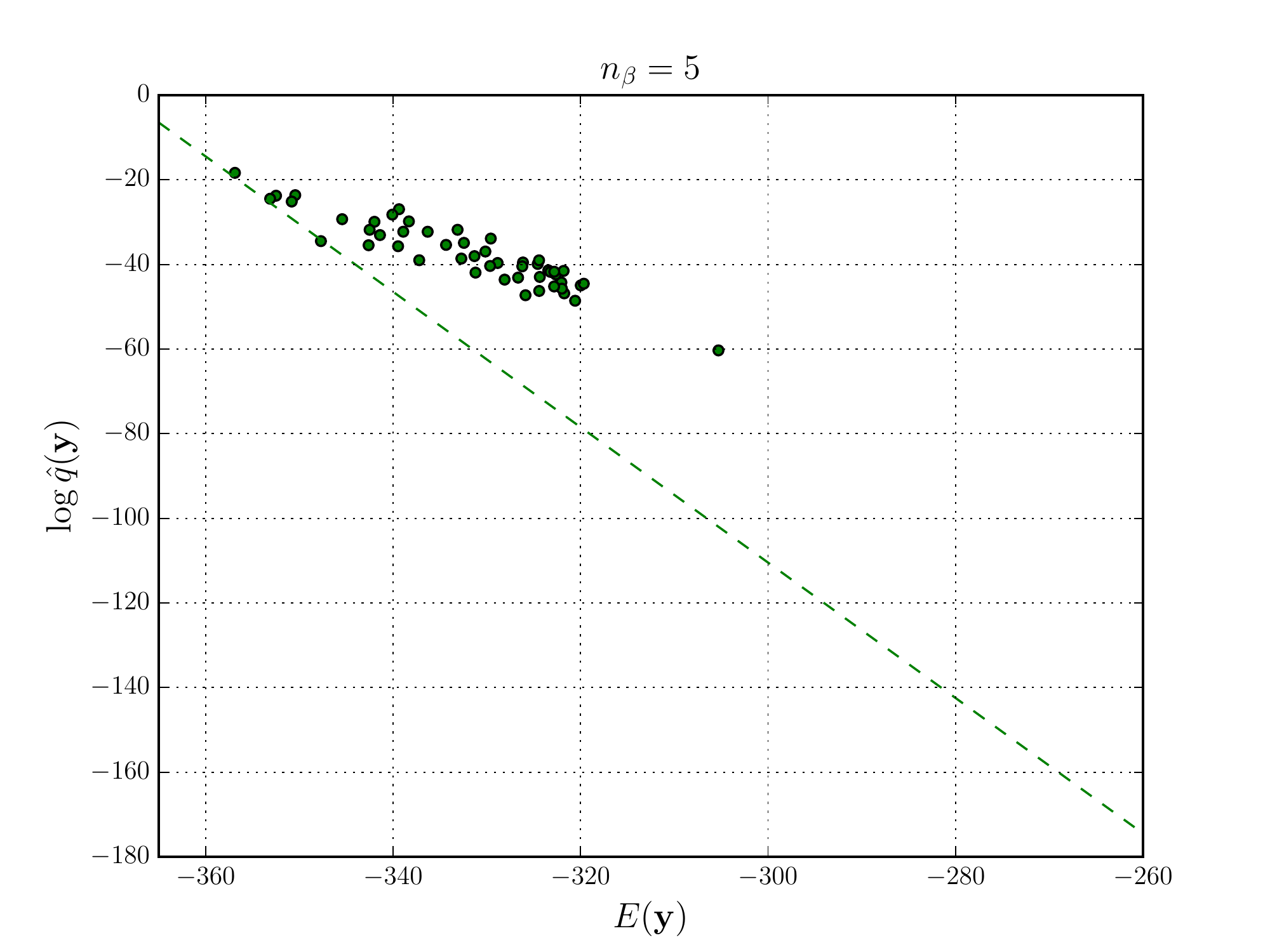}
    \caption{$n_{\beta}=5$}
    \label{fig:EA3DIsing_ElogQ_nBetas_5}
  \end{subfigure}
  \vspace{1mm}
    \begin{subfigure}[b]{0.5\textwidth}
    \includegraphics[width=\textwidth]{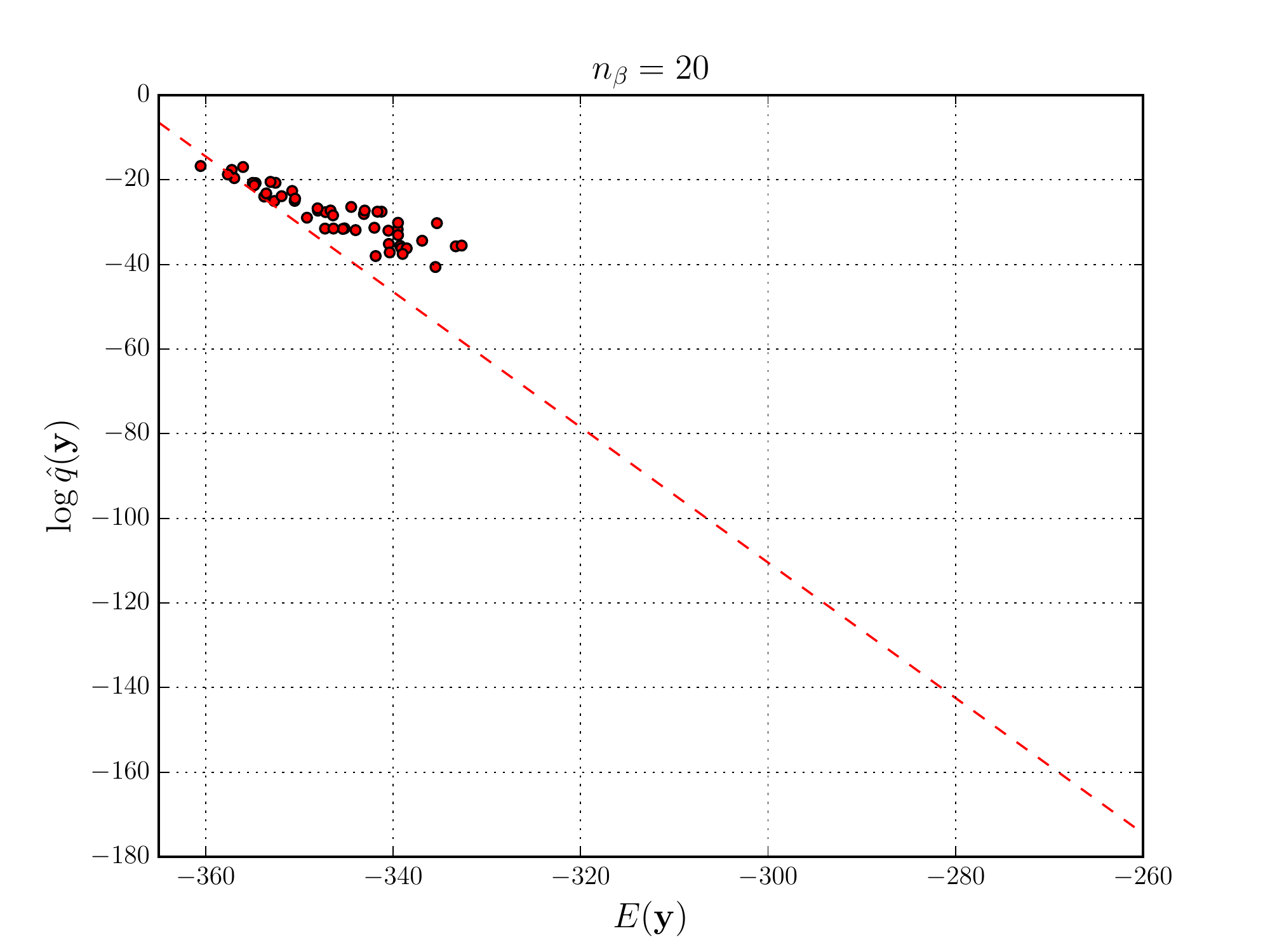}
    \caption{$n_{\beta}=20$}
    \label{fig:EA3DIsing_ElogQ_nBetas_20}
  \end{subfigure}
  \hspace{-5mm}
  \begin{subfigure}[b]{0.5\textwidth}
    \includegraphics[width=\textwidth]{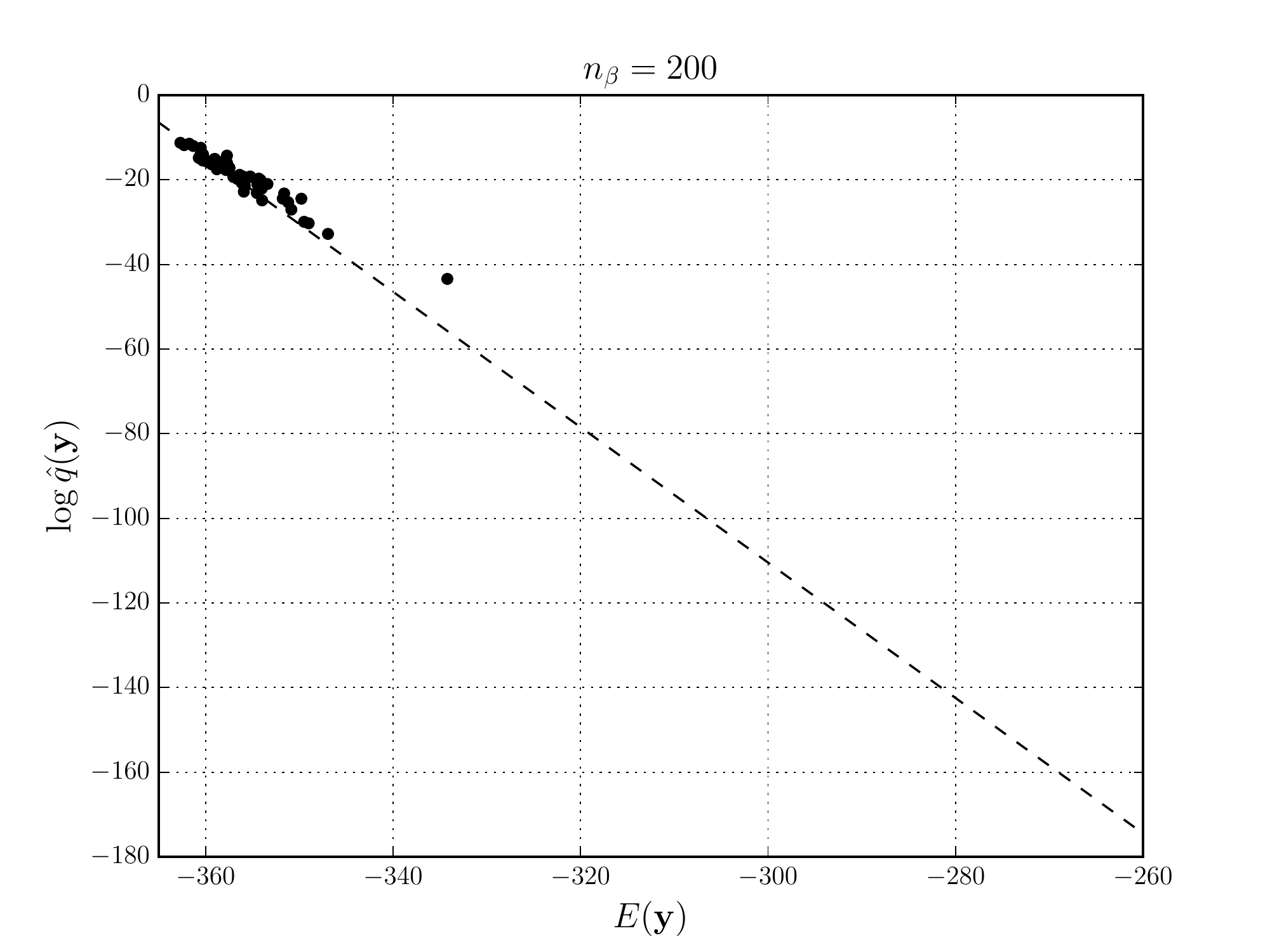}
    \caption{$n_{\beta}=200$}
    \label{fig:EA3DIsing_ElogQ_nBetas_200}
  \end{subfigure}
  \caption{$E$ versus $\log\hat{q}$ scatterplot of $N_{draws}=50$
    states generated by SSS using simulated annealing (SA) on the
    \emph{three dimensional} Gaussian $J_{ij}$ Ising model
    (\instancename{Ising3D}). As for \instancename{IsingSK}, SA was run for each of a set of
    $n_\beta$ batches of $6$ sweeps, with linear $\beta$ increase from
    $0.1$ to $1.6$. Due to the relative difficulty of the problem, the
    longest run used $n_\beta=200$. This problem exhibits hallmarks of
    slow SA convergence; the short runs display severe deviation from
    the Boltzmann line, which, unsurprisingly, becomes milder as the
    run is lengthened. Even the longest cannot be said to have fully
    converged, but it this situation one may expect importance
    sampling to yield satisfactory estimation.}
  \label{fig:EA3DIsing_ElogQ}
\end{figure}

Our final problem, the three dimensional model \instancename{Ising3D},
shows most dramatically how the sample statistics can depend on the
simulation parameters and the relative difficulty SA begins to have.
The set of SA steps was in this case $n_\beta \in \{ 1, 5, 20, 200\}$,
with as before, each step consisting of $6$ sweeps. Annealing was
performed by increasing $\beta$ linearly from $0.1$ to $1.6$, which
again corresponds to being below this model's glass transition
temperature \cite{katzgraber2006universality}. The longest simulation
used $n_\beta=200$ instead of $n_\beta=100$ as was done for the
previous instances because this problem was more challenging, and
equilibration seen to be relatively slow. The Boltzmann line intercept
($-\log Z_\pi$) was approximated using the longest simulation as it
was for \instancename{IsingSK}. The results are shown in Figure
\ref{fig:EA3DIsing_ElogQ}. After $n_\beta=1$ batches of SA sweeps, the
ensemble exhibits considerable deviation from the Boltzmann line.
While convergence to a linear relation appears to have commenced, the
slope of the relation seems to suggest a higher than correct
temperature and its intercept is quite at odds with the target value.
Furthermore, a substantial variance exists. Over the course of
$n_\beta=5,20$ and $200$ SA sweep batches, the ensemble begins
migrating to the Boltzmann line. Note that even for the $n_\beta=200$,
i.e. $1200$ SA sweeps, considerable discrepancy still exists, and one
would be cautious in pronouncing the system to have equilibrated.
Nonetheless, importance sampling can be used to perform estimation at
this point.

We close this section by cautioning that convergence of the $E$ versus
$\log \qhat$ relation to the Boltzmann line constructed with an
\emph{approximation} to $\log Z_\pi$, as would be required in most
problems of interest, is no longer guaranteed to show equilibration,
but merely \emph{quasi-equilibration,} or equilibration on a portion
of the state space. The reason, of course, is that it is always
possible that some statistically-vital portion of the state space was
missed by the heuristic, resulting in $\log Z_\pi$ being
underestimated. Nonetheless, convergence to a linear relation with the
correct inverse temperature slope implies that a Boltzmann
distribution among states of high probability under the heuristic has
been reached.

\section{Discussion}
\label{sec:Discussion}

We have seen that SSS can yield a powerful query mechanism into an
essentially unknown stochastic process that allows construction of an
accurate proposal distribution derived from the process. Here, we
discuss a few relevant points.

First, we remark that the various aspects of the methodology interact
with each other to determine its overall behaviour and performance. An
obvious example of such a dependence is between the heuristic
algorithm parameters, which, in turn dictate the final distribution,
and the SSS tree construction. For SA, of course, the simulation
length and annealing schedule predicate the extent to which the
ensemble has converged onto the statistically-relevant parts of the
target state space. A short SA run will produce a diffuse distribution
of high entropy, which may obviously be mismatched to the target. A
further consequence, however, is that a relatively large number of
calls to the SA process will be required to construct the SSS tree.
This is due to the posterior KL loss growing more rapidly during the
construction when the distribution being approximated has higher
entropy. On the other hand, allowing SA to concentrate onto the
relevant parts of the final distribution means that its final entropy
is reduced, necessitating fewer calls to the heuristic. This was
illustrated clearly in Figure \ref{fig:EA3DIsing_ElogQ}; the
population for the longest SA run is seen localize to a set of much
higher probability than that of shorter runs. Consequently, while each
SA run was of course more demanding, fewer SA calls were required for
the longest run. Once again, we emphasize that reliability of the SSS
meta-algorithm is ultimately decided by the appropriateness of the
heuristic to the distribution of interest.

Other important design parameters, discussed in Appendix
\ref{sec:Appendix:tuningSSSParams}, are the values of the tree
construction KL loss threshold $\theta_0$ and the population size $N$.
As one may expect, the population size must generally scale with the
problem dimension provided the problem's entropy does as well. We
argue in the Appendix that $N$ must scale linearly with the number of
variables to maintain an overall reliability at a fixed level. Though
this requirement is parallelisable, and Rao-Blackwellisation can often
dramatically mitigate it, albeit in a problem-dependent manner, it is
consequently prudent whenever possible to reduce the number of
variables under direct simulation by SSS. One possibility to obtain
such a reduction is the use of \emph{low-treewidth} methods. The idea
is illustrated in Figure \ref{fig:lowTreewidth}; displayed is the top
layer of a $14 \times 14 \times 14 $ Ising lattice. The set of nodes
shaded in black, which repeat in each layer of the lattice, represent
the variables on which we perform SSS sampling using a heuristic of
choice. Once these variables are instantiated, a dynamic
programming/transfer matrix algorithm can be used to sample exactly
from the blocks of $4 \times 4 \times 14$ variables. The probability
of the sample would then simply be that returned by SSS over the black
variables multiplied by the exact conditional probability over the
remaining variables computed during the block sampling. The number of
variables treated by SSS has thus been reduced from the full size
($2744$ variables) to $728$ variables. This technique can be used in
conjunction with the parallelisation method discussed in Appendix
\ref{sec:Appendix:partitioning} to yield further speedup.

\def\svgwidth{0.5\columnwidth}
\begin{figure}
  \centering
  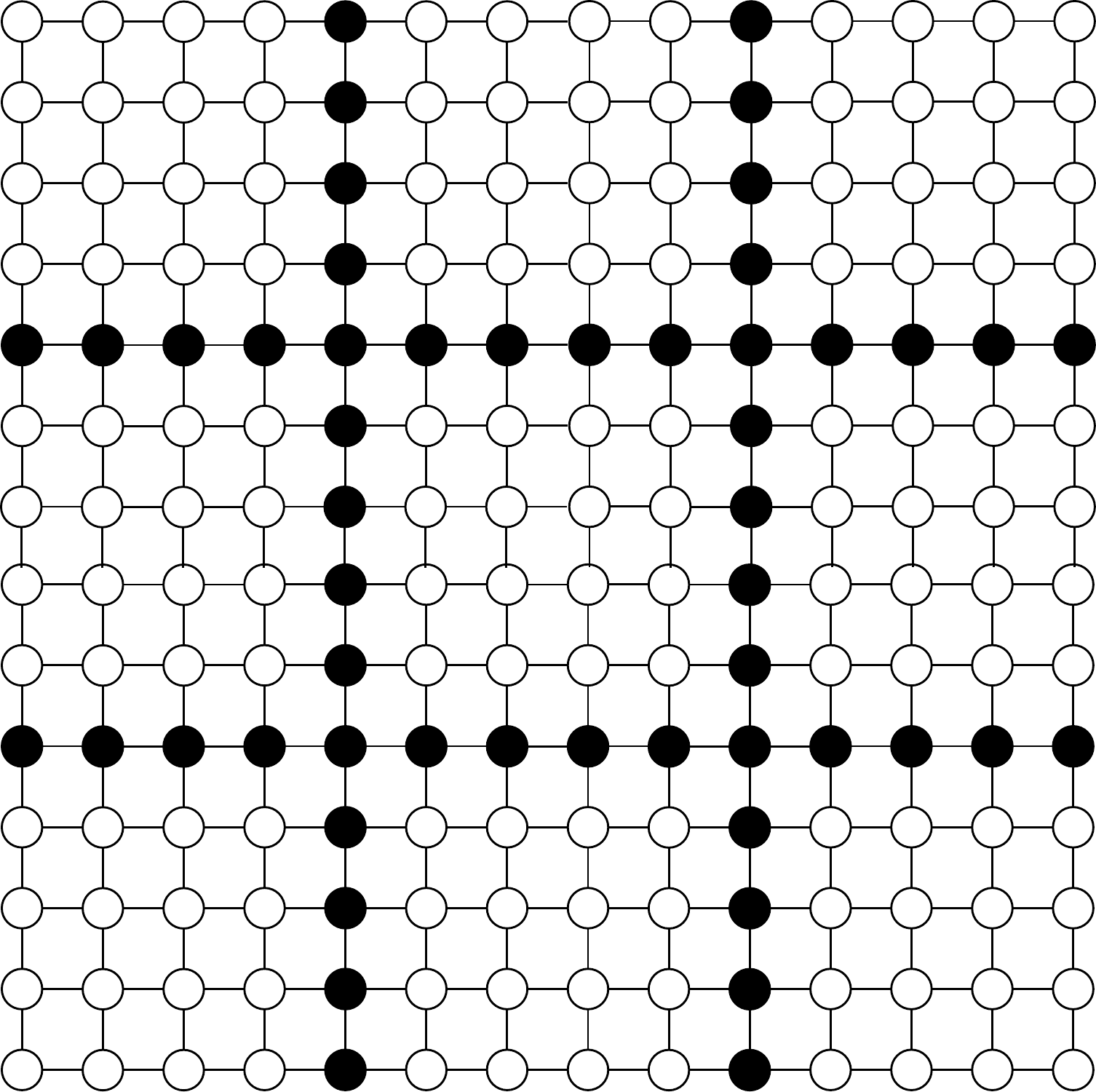
  
  \caption{Top view of a $14\times 14\times 14$ Ising lattice; the black
    nodes, repeated in each layer of the lattice, indicate the set of variables on which SSS would perform
    sampling by using a chosen heuristic. Once these nodes have been
    fully sampled and instantiated, the remaining blocks of $4 \times 4 \times 14$
    variables, filled in white, can be sampled exactly by exploiting
    their low treewidth. The number of variables under explicit
    consideration by SSS is thus reduced from $m=14^3=2744$ to $m=728$
    variables, namely, the total number black nodes in the $14$ layers.
    Noting furthermore that once two sections of the system have been
    broken into disjoint components, simulation within SSS can proceed
    in parallel, an additional possibility for speedup exists as described
    in Appendix \ref{sec:Appendix:partitioning}.}
\label{fig:lowTreewidth}
\end{figure}

Another observation is that the technique of caching the tree, as
described in Section \ref{sec:SSS}, only yields an advantage when the
target distribution is highly peaked; in the presence of large
randomness, only the root subtree is typically repeatedly visited.
The subtrees at lower levels in the tree are overwhelmingly visited
once for a reasonable overall number of SSS draws.

Despite the possibilities for reduction in number of variables and the
exploitation of parallelism, the conclusion that SSS is quite
computationally costly may be unavoidable. However this cost is the
price paid for using an analytically-unknown or uncharacterized
proposal distribution. A comparison with \emph{population
  annealing} \cite{hukushima2003population,machta2010population}, a
state-of-the-art SMC method for the Ising model may yield some
perspective. This method approximates expectations by weighted
averaging over independent estimators, each derived from a population
of \emph{interacting} SA processes. In population annealing, while the
SA simulation length must increase with system size to adequately
probe the low energy structure of a complex problem, the populations
used to compute the estimators can be kept more or less constant. Our
algorithm, on the other hand, generally requires linearly scaling the
population size and necessitates further recursive calls to the
heuristic as the system size increases. In effect, it requires
polynomially more work than SMC using annealing. Obviously then, if
one were using SA as the heuristic, population annealing is the
preferred method. The advantage, of SSS is, of course, is that a more
powerful heuristic than SA may be used, which in turn offsets the
additional cost. As an extreme, though unrealistic example of this,
suppose that one were provided with an oracle that happened to
generate, unbeknownst to the user, \emph{perfect} uncorrelated samples
from a complex target defined over a large number of variables. In
such a scenario, SSS would yield reliable estimators, including those
of the log partition function, using relatively modest resources. SA
on the other hand, would give exponentially-degrading performance.

We turn briefly to a discussion of the various heuristics that can
conceivably be used with the SSS algorithm. As we have seen, some
natural examples that can be employed to compute equilibrium
expectations are: \begin{itemize} \item
  Stochastic local search methods \item Physically-implemented
  heuristic algorithms \item Evolutionary/genetic algorithms
\end{itemize}

However the list can be expanded to include more abstract and
mathematically-founded routines. We recall that within the field of
combinatoric \emph{optimization}, several algorithms exist that are
known to yield good \emph{approximate} solutions
\cite{hochbaum1996approximation,vazirani2013approximation}. Among these are: 
\begin{itemize}
\item Linear programming with solution rounding
\item Primal-dual algorithms
\item Polynomial-time approximation schemes (when they exist)
\end{itemize}
Furthermore, in some situations such as the \emph{graph matching}
problem, computing a minimum solution to a cost function is
polynomially feasible \cite{papadimitriou1998combinatorial} but
sampling uniformly over \emph{all} such minimizers (or more generally,
over low-cost configurations) is highly nontrivial. Usage of one of
these methods for the sampling problem is tempting, but several are
\emph{deterministic} in their basic forms and so would require a
\emph{randomization} strategy to generate a distribution of possible
solutions. One obvious such strategy is to randomly perturb the
problem parameters prior to each invocation of the heuristic. The
resultant states would hence be generated from an analytically unknown
proposal distribution, which may serve as a reasonable surrogate to
the target distribution at low temperature. As a concrete example,
which can easily be seen to be a form of partition function
estimation, consider the problem of counting maximum weight matchings
(MWM) on a given graph. One may invoke a polynomial time
algorithm\cite{edmonds1965paths} to compute MWMs of a set of
randomized edge weight problems, rejecting all solutions whose costs
with respect to the original problem are less than that of the (known)
maximum value. A generally nonuniform and unknown distribution over
maximum weight matchings to the graph is instantiated; invoking
such an algorithm within the SSS framework allows the nonuniformity
bias to be corrected for.

Finally, we note that a natural direction of future work, which was
not pursued experimentally in this paper, is the analysis of different
heuristics within the MCMC context presented in Section
\ref{sec:MCusingSCP:MCMC}.

We believe that the ideas presented in this paper and the SSS
methodology in particular have an appropriate place in the toolbox of
practitioners of numerical Monte Carlo simulation. They open the
possibility of using arbitrary stochastic heuristics to sample from
discrete-valued complex systems. Of course, SSS should never be used
when a better algorithm is known, for example where exact or
fast-mixing sampling algorithms exist. As a tool, it should be viewed
very much as a ``crowbar'' with which to pry open the properties of
stochastic heuristics. Given this recommendation, and the extensive
usage of the concept of recursion in this work, we can think of no
better summary of the appropriate attitude to SSS than by recursively
paraphrasing A. Sokal's \cite{sokal1997monte} advice on Monte Carlo
strategies in general as follows: \begin{guideline}
  State space sampling is an extremely \emph{bad} Monte Carlo method. It
  should be used only when all alternative Monte Carlo methods are worse.
\end{guideline}

\section{Acknowledgements}
We acknowledge discussions with and the helpful insights of W.
Macready, J. Raymond, H. Katzgraber, P. Carbonetto, M. Amin, and A.
Smirnov. F.H. wishes to particularly thank the Vancouver Board of
Parks and Recreation for its masterful job in maintaining the beaches
under its purview so that work on this project in the Summer of 2015
was a true pleasure.

\section{Appendix}

\subsection{Recursive Partitioning for Sparse Models}
\label{sec:Appendix:partitioning}

In many situations, such as the simulation of two or three dimensional
Ising models, the distribution of interest can be interpreted as
possessing certain topological properties. We outline a methodology for
exploiting the structure of a target distribution to impose a
\emph{parallelisability} to the algorithm distinct from the trivial one
arising from the population-based aspect of the SCP.
The notion of \emph{conditional independence} is key here. Given three sets of
variable indices $ A, B, C $,
\( x_A \)
and \( x_B \)
are said to be conditionally independent given \( x_C \) iff
\[
\pi( x_A, x_B | x_C ) = \pi(x_A | x_C) \pi(x_B|x_C)
\]
As an example to illustrate the idea, we consider the case of an Ising
model defined on an \( L \times L \) grid.

The following discussion is illustrated in Figure
\ref{fig:MultilevelNDMC}. Let us initially define the full set of
nodes in the grid to be \( R_{11} \)
and take the set of graph locations \( A_{11} \in R_{11} \)
to be the set of nodes that vertically bisect (as closely as possible)
the grid into two equal-sized rectangular regions. Call the left and
right regions \( R_{12} \)
and \( R_{22} \)
respectively. We note that in the target distribution,
\[ \pi( y_{R_{12}}, y_{R_{22}} | y_{A_{11}} ) = \pi( y_{R_{12}} |
y_{A_{11}} ) \pi(y_{R_{22}} | y_{A_{11}} ) \]
If, after constraining variables \( y_{A_{11}} \)
the proposal is then run independently on the two regions, then it
will of course, also have this property, i.e.
\[ q( u_{R_{12}}, u_{R_{22}} | y_{A_{11}}, \xv ) = q( u_{R_{12}} |
y_{A_{11}}, \xv ) q( u_{R_{22}} | y_{A_{11}}, \xv ) \]
This suggests that judiciously selecting the variable ordering allows
one to \emph{divide and conquer} the problem. In particular, to
simulate
\( \{ U_{R_{12}}, U_{R_{22}}\}\sim q( u_{R_{12}} u_{R_{22}} |
y_{A_{11}}, \xv ) \),
conditional independence allows us to simultaneously generate
\begin{eqnarray*}
  \{U_{R_{12}} \} \sim q( u_{R_{12}} | y_{A_{11}}, \xv ) \\
  \{U_{R_{22}} \} \sim q(u_{R_{22}} | y_{A_{11}}, \xv )
\end{eqnarray*}
as well as variables \( Y_{R_{12}}, Y_{R_{22}} \)
using parallel worker threads. As always, the reduction in number of
free variables allows one to use fewer resources to simulate the
consecutive distributions.

The idea can be applied recursively: suppose in turn that \( A_{12} \)
is a set of nodes that \emph{horizontally} bisects \( R_{12} \)
into regions \( R_{13}, R_{23} \)
respectively, and that \( A_{22} \)
horizontally bisects \( R_{22} \)
into regions \( R_{33}, R_{43} \),
and we now have \emph{four} conditionally independent sets of
variables. The second index in the subscript can be seen to refer to a
``level'' in a hierarchy, with the first being an identifier at the
level. In general, if \( A_{i,l} \)
are a set of nodes that horizontally bisect grid region \( R_{i,l} \)
into \( \{ R_{j,l+1}, R_{k,l+1} \} \)
then \( A_{j,l+1}, A_{k,l+1} \)
will be vertical bisectors of those regions, and vice versa. Let us
examine the complexity of the na\"ive, one-at-a-time SCP algorithm on
the $L \times L$ lattice. Initially, the \( L \)
variables in \( A_{11} \)
are assigned sequentially. Once \( Y_{A_{11}} \)
have been assigned, exploiting conditional independence effectively
allows us to process sets \( A_{12}, A_{22} \),
in \( O(L/2) \)
instead of \( O(L) \)
time. Clearly, in the next step we can sample from four bisecting sets
in effective \( O(L/4) \)
time, and so on. Hence, by exploiting the inherent parallelism
emerging from the problem structure and variable ordering, the
complexity of the one-at-a-time SCP algorithm is effectively reduced
from \( O(L^2) \) to \( O(L\log L) \).

\def\svgwidth{0.5\columnwidth}
\begin{figure}
\centering
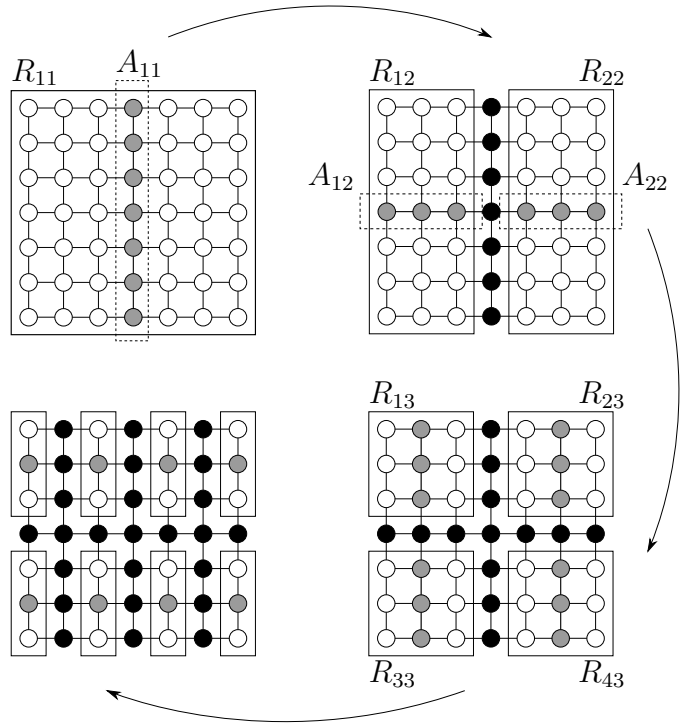
\caption{Exploiting conditional independence to accelerate sampling on
  a \( 7 \times 7 \) grid. Initially, samples are sequentially drawn
  from the vertical shaded column in the top left, called \( A_{11} \). The two sets of variables separated by the column are
  \( R_{12}, R_{22}\). Once the variables in \( A_{11} \) are sampled,
  the variables in \( R_{12}, R_{22} \) are independent. In the top
  right, sampling within these regions proceeds simultaneously along
  \( A_{12}, A_{22} \), the horizontal bisectors of the regions. The
  bisection-based ordering introduces an
  accelerating ``parallelism'' in that as the graph becomes
  increasingly partitioned, an increasing number variables can be
  sampled simultaneously. The effective complexity of the
  sampling process is reduced from \( O(L^2)\) to \(O(L \log L) \),
  where $L$ is the length of the grid.}
\label{fig:MultilevelNDMC}
\end{figure}

For graphs of general topology, the task of finding a set of variables
that maximize the inherent parallelism of the sampling problem is
known as the \emph{vertex separator problem(VSP).} In general, VSP is
an NP-hard problem \cite{bui1992finding}, but there exist several useful
heuristics \cite{de2005vertex} that may give a satisfactory variable ordering.

\subsection{Selecting the parameters for SSS}
\label{sec:Appendix:tuningSSSParams}

Recall from the SSS algorithm description that we set a KL loss
threshold $\theta_0$ and grow the state space tree representation
until this threshold is violated, recursively calling the heuristic to
continue growing the tree as needed. One may naturally wonder how to
set $\theta_0$ and the population size $N$ to ensure a given accuracy
at any system size. After all, if $\theta_0$ and $N$ are held constant
while larger systems are simulated, the overall divergence between the
estimate and the true distribution will increase. In this section, we
argue that $N$ should be allowed to increase linearly with system
size, though we remind the reader that this is a highly parallelisable
burden.

Suppose we observe in an $m$-variable system that SSS typically
requires around $N_{calls}$ calls to the heuristic process to
generating a sample. The value of $N_{calls}$ depends on the target
distribution entropy, with larger entropy necessitating larger
$N_{calls}$. Under KL loss threshold $\theta_0$, the \emph{overall} KL
divergence between the target and the estimator will be
$\sim N_{calls}\theta_0$. Suppose for simplicity that the target is
uniform. Then doubling the system size $m$ will consequently double
the required $N_{calls}$ and hence the total divergence. If, however,
we \emph{halved} the partition construction tolerance for the larger
system to $\theta_0/2$ and \emph{doubled} the population size $N$,
then $N_{calls}$ will still double over that of the $m$-variable
system. This is because the sizes of the constructed subtrees under
parameters $\theta_0/2$ and $2N$ will be approximately the same as
under $\theta_0$ and $N$, which in turn follows from the fact that the
posterior KL loss essentially becomes $O(1/N)$ for moderate-to-large
$N$. The effect is therefore that the \emph{total} divergence reverts
to $\sim N_{calls}\theta_0$.

The message is that in general, scaling a system size by a factor of
$c$ within a problem class where the entropy increases linearly with
size requires reduction and increase of $\theta_0$ and $N$ by a factor
of $c$ respectively in order to maintain the KL divergence at a fixed
value.  Rao-Blackwellisation as described in Section
\ref{sec:SCP:qnEstimators:RaoBlack} can often dramatically reduce this
number in practice but the discussion in this section at least
provides a conservative general guideline. Experimentation on the
problem class may be required to select an efficient population size.

\subsection{Bayesian Updating of State Space Trees}
\label{sec:Appendix:BayesianUpdating}

As mentioned in the text, one of the advantages of Bayesian estimators
is their natural ability to adapt to more data that may arrive. In
this section, we consider informally how to implement this
functionality within the SSS methodology.  Doing so is not onerous but
does require a certain amount of bookkeeping and care, especially when
combined with memory-bounding.

It is more appropriate in the present context to consider the full
state space tree as a collection of \emph{concatenated} subtrees
corresponding to state space partitions, each endowed with probability
estimates over its nodes. When the root node of such a subtree,
receives a certain number of visits, further runs to the SCP with the
appropriate condition are made, but only the probability estimates
corresponding to nodes of the \emph{original} partition are
updated. Nodes belonging to downstream subtrees require their own
refresh. When retracting a subleaf node, the Dirichlet aggregation
property implies that the prior parameter of the new leaf node is
simply the sum of its two deleted childrens'. Hence, Bayes estimators
on pruned trees resulting from retraction can be updated without
issue.

One complication with Bayesian updating in the presence of retraction
is that in principle, once probability estimates in a subtree are
modified, the probabilities of \emph{all} nodes descended from this
tree are changed as well. This implies that the subleaf probability
values comprising the deletion priority in the retraction queue will
no longer correspond to their actual probabilities.  Updating these
values can become quite expensive, so a natural option is to simply
use the ``stale'' initially-computed probabilities as deletion
priorities. This does not affect the correctness of the algorithm;
subleaf nodes could even have been dropped at random. The idea behind
the prioritization is to simply avoid removing representations that
are likely to be used multiple times. The initial approximations are
likely accurate enough to effect this behaviour satisfactorily.

\bibliography{NDMCBibliography}{}
\bibliographystyle{abbrv}

\end{document}